\documentclass[showpacs,aps,superscriptaddress,notitlepage,reprint,longbibliography]{revtex4-1}

\usepackage[dvips]{graphicx}
\usepackage{amsmath,amssymb,amsthm,mathrsfs,amsfonts,dsfont,amscd,keyval}
\usepackage{mathtools,mathrsfs}
\usepackage{yfonts} 
\usepackage{subfigure, epsfig}
\usepackage{titletoc} 
\usepackage[subfigure]{tocloft}
\usepackage{titlesec}
\usepackage[linesnumbered,ruled,vlined]{algorithm2e}

\usepackage{braket}
\usepackage{tensor}
\usepackage{bm}
\usepackage{enumerate,enumitem}
\usepackage{color}
\usepackage{hyperref}
\usepackage{tabularx}
\usepackage{times}
\hypersetup{breaklinks=true}
\hypersetup{colorlinks=true, linkcolor=magenta, citecolor=magenta, urlcolor=blue, linktocpage=true}
\usepackage[hyphenbreaks]{breakurl}
\usepackage{multirow}
\usepackage{float}

\usepackage{subcaption}
\usepackage{tcolorbox}

\usepackage{array,booktabs,tabularx}
\newcommand{\bvec}[1]{\bm{#1}}

\newenvironment{proofsketch}
  {\par\noindent{\it Proof Sketch:}\hspace{0.5em}\rm}
  {\hfill$\square$\par}

\renewcommand\vec{\mathbf}

\newtheorem{m-theorem}{Theorem}

\theoremstyle{definition}
\newtheorem{definition}{Definition}[section]
\newtheorem*{example}{Example}
\newtheorem{proposition}[definition]{Proposition}

\theoremstyle{plain}
\newtheorem{theorem}{Theorem}
\newtheorem{lemma}[definition]{Lemma}
\newtheorem{corollary}{Corollary}

\theoremstyle{remark}
\newtheorem*{remark}{Remark}

\newcommand{\comments}[1]{}

\renewcommand\vec{\mathbf}

\SetKwInput{KwIn}{Input}
\SetKwInput{KwOut}{Output}
\SetKwInput{KwResources}{Resources}
\SetKwInput{KwPre}{Preconditions}
\SetKwInput{KwPost}{Postconditions}
\SetKwInput{KwGoal}{Goal}

\SetAlgoNlRelativeSize{-1}      
\SetAlFnt{\small}               
\SetKwComment{tcc}{\% }{}       
\DontPrintSemicolon              

\SetKw{Prepare}{prepare}
\SetKw{Repeat}{repeat}
\SetKw{Measure}{measure}
\SetKw{Apply}{apply}
\SetKw{Reset}{reset}
\SetKw{Wait}{wait}
\SetKw{Barrier}{barrier}
\SetKw{And}{and}
\SetKw{Or}{or}

\SetKwBlock{Parallel}{\textbf{parallel}}{} 
\SetKwBlock{Sample}{\textbf{sample}}{}
\SetKwRepeat{Repeat}{repeat}{until}            

\newcommand{\Round}[1]{%
  \tcc*[r]{\textbf{Round }#1}%
  \vspace{0.25ex}\hrule\vspace{0.45ex}%
}

\newenvironment{protocol}[2]{%
  \begin{algorithm}[h!]\caption{#1}\label{#2}%
}{\end{algorithm}}

\makeatletter
\def\l@subsubsection#1#2{} 
\makeatother

\begin{document}

\title{Efficient learning of logical noise from syndrome data}

\author{Han Zheng}
\email{hanzuchicago@gmail.com}

\affiliation{Pritzker School of Molecular Engineering, The University of Chicago, Chicago, IL 60637, USA}
\affiliation{Department of Computer Science, The University of Chicago, Chicago, IL 60637, USA}

\author{Chia-Tung Chu}
\affiliation{Pritzker School of Molecular Engineering, The University of Chicago, Chicago, IL 60637, USA}

\author{Senrui Chen}
\affiliation{Pritzker School of Molecular Engineering, The University of Chicago, Chicago, IL 60637, USA}

\affiliation{Institute for Quantum Information and Matter, California Institute of Technology, Pasadena,
California 91125, USA}

\author{Argyris Giannisis Manes}
\affiliation{Pritzker School of Molecular Engineering, The University of Chicago, Chicago, IL 60637, USA}

\author{Su-un Lee}
\affiliation{Pritzker School of Molecular Engineering, The University of Chicago, Chicago, IL 60637, USA}

\author{Sisi Zhou}
\affiliation{Perimeter Institute for Theoretical Physics, Waterloo, Ontario N2L 2Y5, Canada}
\affiliation{Department of Physics and Astronomy, Department of Applied Mathematics and Institute for Quantum
Computing, University of Waterloo, Ontario N2L 2Y5, Canada}

\author{Liang Jiang}
\email{liangjiang@uchicago.edu}
\affiliation{Pritzker School of Molecular Engineering, The University of Chicago, Chicago, IL 60637, USA}

\begin{abstract}

Characterizing errors in quantum circuits is essential for device calibration, yet detecting rare error events requires a large number of samples. This challenge is particularly severe in calibrating fault-tolerant, error-corrected circuits, where logical error probabilities  are suppressed to higher order relative to physical noise and are therefore difficult to calibrate through direct logical measurements. Recently, Wagner \emph{et al}. [PRL 130, 200601 (2023)] showed that, for phenomenological Pauli noise models, the logical channel can instead be inferred from syndrome measurement data generated during error correction. Here, we extend this framework to realistic circuit-level noise models. From a unified code-theoretic perspective and  spacetime code formalism, we derive necessary and sufficient conditions for learning the logical channel from syndrome data alone and explicitly characterize the learnable degrees of freedom of circuit-level Pauli faults. Using Fourier analysis and compressed sensing, we develop efficient estimators with provable guarantees on sample complexity and computational cost. We further present an end-to-end protocol and demonstrate its performance on several syndrome-extraction circuits, achieving orders-of-magnitude sample-complexity savings over direct logical benchmarking. Our results establish syndrome-based learning as a practical approach to characterizing the logical channel in fault-tolerant quantum devices.

\end{abstract}

\maketitle


\section{Introduction}

Building useful quantum computing architectures requires precise manipulation of quantum processors and protection of the information they carry from unwanted error sources. Fault-tolerant quantum computing will be executed on logical qubits protected by quantum error correction (QEC)~\cite{nielsen2010quantum, terhal2015quantum, gottesman2010introduction, campbell2017roads}, which has already recently witnessed many exciting experimental breakthroughs~\cite{Acharya2023, Bluvstein2024, Sivak2023, Krinner2022, Postler2022, Acharya2025, Rodriguez2025}. Quantum error correction deals with noise by using multiple physical qubits to encode logical information, making it less sensitive to the impact of the noise. The defining feature of such fault-tolerant architecture is the performance of \emph{syndrome extraction}, which detects faults during the circuit execution without interfering with the logical information. Logical computations are typically comprised of Clifford operations--which can be achieved through transversal Clifford unitaries or logical Pauli-based computation (PBC)~\cite{Bravyi_2016}. Non-Clifford operations are either implemented in the form of transversal (constant-depth) non-Clifford gates~\cite{vasmer2019three, paetznick2013universal, bombin2015gauge} or in preparing resource states such as magic state distillation~\cite{Bravyi_2012} and cultivations~\cite{gidney2024magicstatecultivationgrowing}. These operations inevitably introduce faults, together with additional errors arising from idling, state preparation, and (syndrome) measurements. Collectively, these processes define the \emph{circuit-level} fault models associated with syndrome-extraction circuits. 

Central to fault-tolerant quantum information processing is the ability to characterize and benchmark the resulting logical error probability. Direct logical benchmarking, however, becomes increasingly costly as logical errors are suppressed to higher order. This motivates the question of whether syndrome data—already collected during error correction—can be used to efficiently predict logical error probabilities at the end of circuit execution. This viewpoint has been explored in a series of works by Wagner \emph{et al.}~\cite{Wagner_2021, Wagner_2022, Wagner_2023}, which showed that under phenomenological noise models and idealized syndrome extraction, the pre-decoding (or effective) logical error probability~\cite{Wagner_2023, Beale_2018, takou2025estimatingdecodinggraphshypergraphs} can be learned from syndrome data alone. These results motivate a broader question: under realistic circuit-level noise, to what extent—and with what efficiency—can syndrome data predict logical error probabilities?

In parallel, there has recently been a surge of interest in learning the detector error model from syndrome data~\cite{blumekohout2025estimatingdetectorerrormodels, arms2025estimatingdetectorerrormodels, iyer2025enhancingdecodingperformanceusing, takou2025estimatingdecodinggraphshypergraphs}. Learning detector error models can be seen as learning the \emph{prior} distribution, primarily discussed in the decoder community such as Ref.~\cite{Sivak_2024}. A timely and relevant question can be formulated: to what extent does the prior distribution associated with the detector error model give a prediction for logical error probability? 

One can view these questions as learnability questions: what are the learnable and unlearnable degrees of freedom associated with a (circuit-level) fault model, given a (fixed) restricted access of measurement instruments. More precisely, in our case with a circuit execution of logical computation, accessible resources are limited to operations that do not collapse the encoded information--namely, the syndrome measurement--which introduces additional new gauge degrees of freedom. Hence, we need to develop a unified theoretical framework for learning the logical channel of QEC circuits, taking the circuit-level fault models into consideration.  The existence of gauge degrees of freedom draws an intuitive similarity to that studied in the context of gate-set tomography~\cite{chen2024efficient, Chen_2023_learnability, Zhang_2025, combes2014insitucharacterizationquantumdevices,Nielsen_2021, Erhard_2019, Flammia_2020}, whereupon identifying the gauge action is crucial to learnability (of Pauli channels) from state preparation and measurement (SPAM) noise.

The first contribution in this work is to formulate the necessary and sufficient learnability conditions from syndrome data, based on the classification of logical equivalences for any given stabilizer code, or any (logical) Clifford circuit implementing a base stabilizer code~\cite{bacon2017sparse, delfosse2023spacetimecodescliffordcircuits, gottesman2022opportunitieschallengesfaulttolerantquantum}. Using representation-theoretical approaches~\cite{TaoVu2006AddComb, goodman_wallach_2009}, we can unify the discussion on the learnability of (logical) Pauli faults from a Clifford quantum circuit with syndrome extractions and that from a base stabilizer code on the same footing. This generalizes the results from Refs.~\cite{Wagner_2022, Wagner_2023} studied assuming phenomenological noise models. As a simple corollary, our necessary and sufficient learnability condition explains precisely when the prior distribution predicts the logical error probability. The extension of learnability to circuit-level fault models is characterized precisely by a (Clifford) circuit-to-(subsystem) spacetime code mapping, which takes inspiration from, and generalizes, the spacetime mapping schemes in Refs.~\cite{delfosse2023spacetimecodescliffordcircuits, bacon2017sparse}.

Second, we develop efficient protocols with provable guarantees on sample complexity and computational costs. In quantum information processing, a pivotal question can be framed as follows. Given an unknown quantum process and an oracle that records some measurement data, what is the minimum  number of measurement data required to recover or learn the quantum process? This question takes many diverse forms: in the case of random unitaries~\cite {Hayden_2007, XuSwingle_2024} and shadow tomography~\cite{Huang_2020}, randomized Pauli measurements~\cite{Gross_2010, liu2011universallowrankmatrixrecovery}, Hamiltonian learning~\cite{Huang_2023, ma2024learningkbodyhamiltonianscompressed}, and randomized and cycle benchmarking~\cite{Erhard_2019, Knill_2008, Zhang_2025, Helsen_2022}. In our case, we wish to learn an unknown fault process $\Lambda$ from syndrome measurement outcomes from a Clifford circuit implementing a base stabilizer code.  Via the Fourier analysis on the Abelian measurement group and utilizing the compressed sensing technique, we provide a provably efficient guarantee on learning (logical) Pauli channel from syndrome data. In real experiments, the obtained syndrome expectation values may have sampling imprecision. Using the property of compressed sensing or the restricted isometry property (RIP)~\cite{haviv2015restrictedisometrypropertysubsampled, FoucartRauhut2013, Vershynin2018HDP, RudelsonVershynin2008}, we control the propagation of errors from sampling imprecision of syndrome expectation values in learning the (physical) logical error. 

Third, we explicitly characterize the learnable degrees of freedom with circuit Pauli faults, even if it is impossible to learn every fault probability of the entire circuit up to logical equivalence. We give an optimal and variance-aware sampling complexity from collecting syndrome expectations to estimating logical error probability in an end-to-end protocol. We demonstrate the end-to-end estimation protocol in several syndrome extraction circuit examples, and we observe an orders-of-magnitude sample complexity saving in estimating the logical error probability.

The paper is constructed as follows. In Section~\ref{main-sec: main-results} we state the main results of efficiently learning logical faults from syndrome data. In Section~\ref{main-sec: circuit-learnability}, we present a form of circuit-to-code mapping in \ref{main-subsec: circuit-to-code-mapping} from which we could unify the discussion on learnability with circuits and codes on the same footing from representation theory of the Boolean group.  In Section~\ref{main-sec: sampling-complexity}, we detail the protocol of optimal sample complexity in estimating logical error probabilities, utilizing techniques from  compressed sensing (Section~\ref{main-subsec: error-propagation-rip}) and optimal, variance-aware sampling techniques, inspired from Ref.~\cite{lee2025efficientbenchmarkinglogicalmagic} (Section~\ref{main-subsec: background-free-sampling}). Finally, we detail an end-to-end estimation framework in estimating the logical error probability from the Algorithm~\ref{algo: end-2-end-framework-from-syndrome} and Section~\ref{main-subsec: lep-sampling-complexity}.

\begin{figure*}[pt]
\centering

\includegraphics[width=0.9\textwidth]{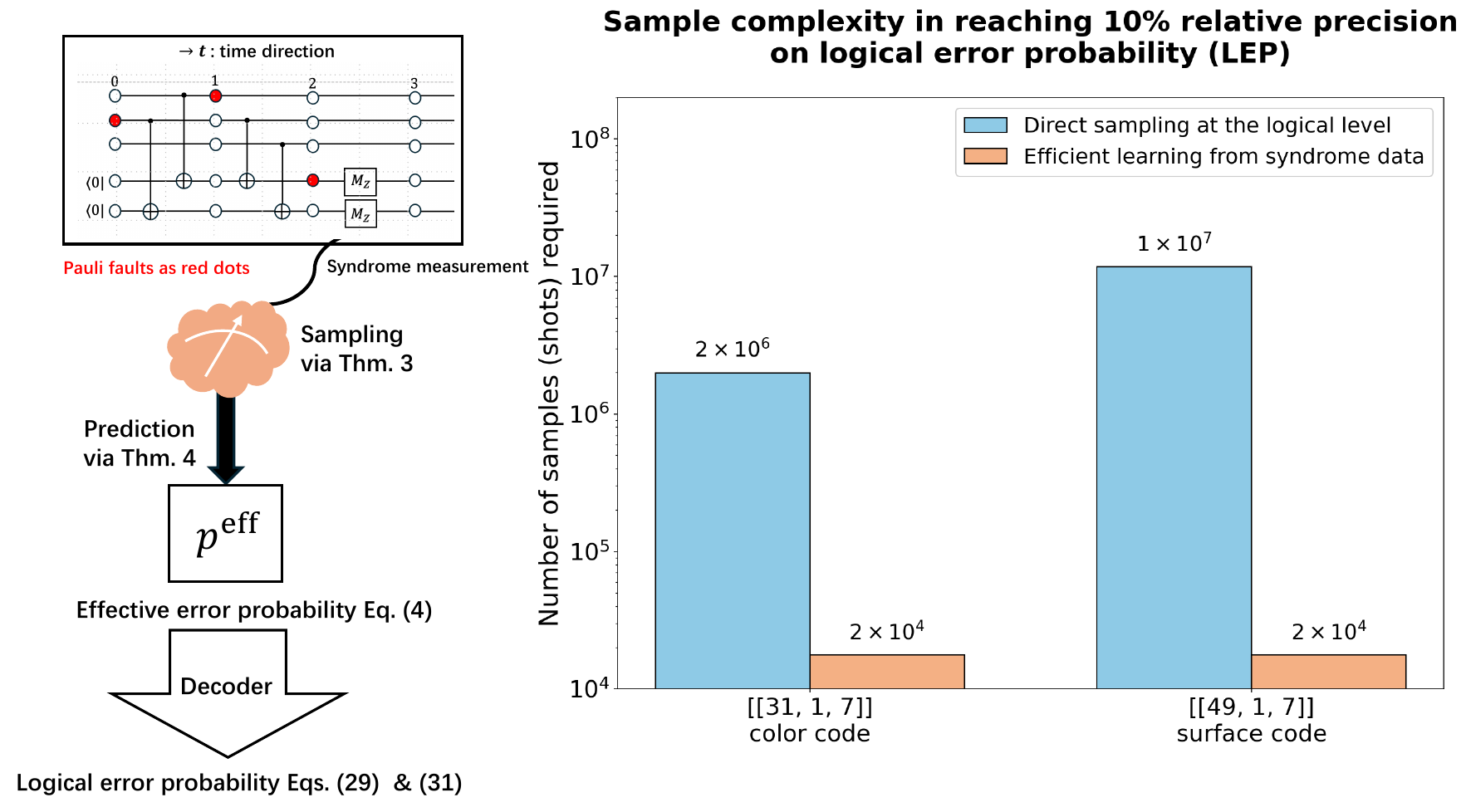}

\caption{Illustration of the learnability problem studied in this work. On the left, we give a schematic illustration of utilizing the syndrome data to infer the effective error probabilities. We provide the necessary and sufficient learnability condition in Theorem~\ref{thm: main-physical-learnability} and Theorem~\ref{thm: main-logical-learnability}. The efficient sampling and robust prediction procedure are given respectively from Theorem~\ref{thm: main-row-subsampled-rip-A} and Theorem~\ref{thm: background-free-estimation-prior}. Finally, the logical error probability conditioned on a given circuit-level decoder is computed through Eqs.~\eqref{eq: main-logical-error-probability} \&~\eqref{eq: main-prior-distribution-logical-error}. The entire learning procedure is summarized in Algorithm~\ref{algo: end-2-end-framework-from-syndrome}. On the right, we show that, under a circuit-level fault model with memory experiments, the required sample complexity in estimating the logical error probability within $10\%$ of relative precision. For the square-octagon $(4.8.8)$ color code~\cite{bombin2007exact} and rotated surface code with distance $7$, we demonstrate orders-of-magnitude savings in estimating the logical error probability to $10\%$ of relative precision. }
\label{Fig: simulation}
\end{figure*}

\section{Preliminary and main results}\label{main-sec: main-results}
In this section, we set up the preliminaries necessary to initiate discussion in the following and state the main results. First, we introduce the following notation. Let $x, y \in \mathbb{R}$. We denote $x \approx_{\tau, \varepsilon} y$ if $x \in [(1-\tau)y - \varepsilon, (1+\tau)y + \varepsilon]$. We also provide a short notation table necessary for initiating the following discussion. A complete notation table is given in the Appendix.

\newcolumntype{N}{>{\centering\arraybackslash$\displaystyle}p{0.05\textwidth}<{$}}
\newcolumntype{E}{>{\raggedright\arraybackslash}X}

\begin{table}[htbp]
  \centering
  \renewcommand{\arraystretch}{1.5}
  \begin{tabularx}{\textwidth}{N E}
    \toprule
    \multicolumn{1}{c}{\textbf{Notation}} & \textbf{Meaning} \\
    \midrule
    \mathcal{A}  & Boolean group as spacetime Pauli group~\cite{delfosse2023spacetimecodescliffordcircuits, bacon2017sparse}. \\

    \langle \cdot, \cdot \rangle & The (spacetime) symplectic inner product. \\

    \mathcal{S} & The stabilizer group. \\

   \mathcal{C}_T(\mathcal{S}) & Quantum circuit implementing the base stabilizer group. \\

   \Lambda & The spacetime Pauli fault distribution. \\

    \mathscr{G}, \mathscr{M} & The gauge and measurement (center) subgroup from $\mathcal{C}_T(\mathcal{S})$. \\

    \sigma  & The syndrome vector associated with $\mathscr{M}$.  \\
    
   \mathcal{K}  & Independent fault processes.  \\

   \mathcal{K}_{\mathscr{M}} & Independent fault processes of nonzero, distinct syndromes.\\

   1- \varepsilon & the probability where there is no fault from $\Lambda$. \\
   
    \delta_K & The $K+1$th order restricted isometric constant.  \\

    p_a & Pauli error rate (probability that the fault $a$ occurs). \\

    q_a & The parameterization coefficient in Eq.~\eqref{eq: main-character-fourier-basis}. \\

q_{[a]} &  The parameterization coefficient of the prior distribution in Eq.~\eqref{eq: main-prior-distribution}. \\

       \bottomrule
  \end{tabularx}
  \caption{Elementary notation used in this manuscript. Other notation used in more technical analysis is consistent with the above notational conventions.}
  \label{tab: notation}
\end{table}

Throughout the manuscript, we assume that $\mathscr{M}$ is generated by $M$ independent generators so that $\mathscr{M} \cong \mathbb{F}^M_2$. For a slight abuse of notation, we also denote $\sigma: \mathscr{M} \rightarrow \mathbb{F}^M_2$ given by the decomposition of an element in $\mathscr{M}$ to its generators. We denote the number of independent fault processes of distinct nonzero syndromes by $K = |\mathcal{K}_{\mathscr{M}}|$. Note that $I \notin \mathcal{K}$.

\subsection{Preliminary}\label{main-subsec: preliminary}

A $n$-qubit Pauli group is generated by $\mathrm{P}^n = \langle X,  Z\rangle^{\otimes n}$ modulo the phase. Pauli faults represent the most common fault type in the existing quantum circuits, which can be represented as a Pauli channel 
\begin{align}
    \Lambda = \sum_{a \in \mathrm{P}^n} p_a a \cdot a,
\end{align}
where we denote $p_a$ (Pauli error rate) the probability at which the Pauli fault $a$ occurs. One can always perform the Walsh-Hadamard (Fourier) transform to obtain the Pauli eigenvalues $\Lambda = (1/2^n)\sum_{a \in \mathrm{P}^n} \lambda_a \operatorname{Tr}(a \cdot)$~\cite{Chen_2022}, related by
\begin{align}
    \lambda_b = \sum_{a \in \mathrm{P}^n} p_a \chi_b(a); \quad \chi_b(a) = (-1)^{\langle a, b \rangle}.
\end{align}

The inner product $\langle a, b \rangle$ denotes the $\mathbb{F}_2$-valued inner product over the symplectic representation of Pauli elements such that $\langle a, b \rangle =1$ if they anti-commute and otherwise $0$. The Pauli eigenvalue $\lambda_b$ takes values from $[-1, 1]$. In this work, we assume that $p_I > 1/2$ so that $1\geq \lambda_b > 0$ for all $b \in \mathrm{P}^n$. This assumption reflects the standard assumption that the faults do not dominate the system and, mathematically, removes a potential sign ambiguity, similar to Refs.~\cite{Wagner_2022, Wagner_2023}. More precisely, we assume that there exists some small parameter $\varepsilon$ such that $\lambda_a \geq 1-\varepsilon$, aligning with assumptions taken in the literature of benchmarking~\cite{Flammia_2020, chen2024efficient, Chen_2023_learnability, chen2025disambiguatingpaulinoisequantum}. For obvious reasons and slight notational abuse, we refer to $\varepsilon$ as \emph{physical error probability}. It is equivalent to view the Pauli operator as an element in the Boolean group $\mathcal{A} \cong \mathbb{F}^{2n}_2$ equipped with the $\mathbb{F}_2$-valued symplectic inner product. In what follows, readers may think $\mathcal{A}$ as Pauli group or spacetime Pauli group defined in~\cite{delfosse2023spacetimecodescliffordcircuits} and see in Section~\ref{main-subsec: circuit-to-code-mapping}. One can formally represent $\Lambda \in \mathbb{R}[\mathcal{A}]$, the space of real-valued functions on $\mathcal{A}$ subject to  normalization~\footnote{The use of real-valued functions $\mathbb{R}[\mathcal{A}]$ may sound strange since typically we extend to the complex-valued field. However, in the case of the Boolean group $\mathcal{A}$, its characters are all real-valued. Hence, we can safely restrict to the real domains.}. As such, we refer to $\Lambda$ as a \emph{fault distribution} on $\mathcal{A}$ and the Pauli error rates and the Pauli eigenvalues are related by Fourier transform over $\mathcal{A}$ (see details in Section~\ref{sec: learnability-general}). Hence, in what follows, we denote $p_a$ to be the fault probability of $a \in \mathcal{A}$ and $\lambda_a$ its Pauli eigenvalue. Viewing a Pauli fault model as a probability distribution, its degrees of freedom can be roughly understood as independent probability events. For example, the joint two-qubit bit-flip must be a result of compositions of simultaneous single bit-flips, if the fault model is restricted to single bit-flips. As a result, this fault distribution can only have the number of degrees of freedom that scales linearly in the system size.

The benefit for the abstraction, as we will see in the following sections, is two-fold: $(i)$ it unifies our discussion on Pauli faults on the memory level and circuit level. $(ii)$ It naturally connects to the theory of compressed sensing via the Fourier transform over Boolean group. Indeed, for a quantum Clifford circuit, we could consider the Pauli faults taking place at various spacetime coordinates on the circuit, which form the spacetime Pauli group. Additionally, one could represent measurement faults as a classical bit-flip occurring at the ancilla initialized at $|0\rangle$ state, such as the framework of the quantum data-syndrome code~\cite{ashikhmin2020quantum}. Fault distribution in these settings can all be considered as a distribution from some Boolean group algebra $\mathbb{R}[\mathcal{A}]$. By adapting this viewpoint, we can parametrize a distribution $\Lambda$ on $\mathcal{A}$ as follows
\begin{align}\label{eq: main-character-fourier-basis}
    \Lambda = *_{a \in \mathcal{K}}\left( (1-q_a) + q_a \chi_a \right),
\end{align}
where $\mathcal{K} \subseteq \mathcal{A}$ denotes the set of independent parameterization with coefficients $q_a \in [-\infty, 1/2)$, excluding the identity. $\chi_a$ denotes the character function of the group $\mathcal{A}$, $\chi_a: \mathcal{A} \rightarrow \{ -1, 1 \}$ such that $\chi_a(b) = (-1)^{\langle a, b \rangle}$. This parameterization is also referred to as the \emph{Pauli-Lindblad} model~\cite{van_den_Berg_2024, van_den_Berg_2023, gupta2023probabilisticerrorcancellationdynamic, chen2025disambiguatingpaulinoisequantum, jaloveckas2023efficientlearningsparsepauli} and in our case we additionally require that $q_a$ can take negative values in order to represent generic Pauli channels. For instance, even for the single-qubit Pauli channel, the coefficients for the character basis could be negative, which do not satisfy the Pauli-Lindblad model in a strict sense (see an example in Section~\ref{subsec: parameterization-fault-distribution}). The inclusion of negative coefficients enables us to parameterize any Pauli channel into the form Eq.~\eqref{eq: main-character-fourier-basis} (See texts around Lemma~\ref{lemma: fourier-character-basis-sign} for complete and rigorous statements), so that we assume this generality in what follows.

The formalism of Boolean group algebra is also convenient when considering with (subsystem) stabilizer codes. A quantum (subsystem) stabilizer code $(\mathscr{G}, \mathscr{M})$ is uniquely specified by its gauge group $\mathscr{G} \subseteq \mathcal{A}$ and its center $\mathscr{M} = \mathscr{G}^\perp \cap \mathscr{G}$, the stabilizer subgroup (called later as the measurement subgroup) of the gauge group which commutes with every element in $\mathscr{G}$. A (plain) stabilizer code is given by the gauge group $\mathscr{M} = \mathscr{G}$. The smallest nontrivial and well-known subsystem code is the Bacon-Shor $[[4, 1, 2]]$ code~\cite{alam2025baconshorboardgames} defined on the $2 \times 2$ lattice, where the gauge group is generated by $\mathscr{G} = \langle X_1X_2, X_3X_4, Z_1Z_3, Z_2Z_4 \rangle$ and $\mathscr{M} = \langle Z_1Z_2Z_3Z_4, X_1X_2X_3X_4 \rangle$. At the core of our mathematical framework lies in characterizing the symmetry of an underlying error-correcting code. Famously, C.N.Yang said, "symmetry dictates interactions". The gauge group manifests the symmetry action by the notion of \emph{logical equivalence}. Namely, logical operators dressed by gauge elements are classified the same. In the case of the $[[4, 1, 2]]$ Bacon-Shor code above, the bare $Z$ logical operator is $Z_1Z_2$ and the bare $X$ logical operator $X =X_1X_3$, which all commute with elements in $\mathscr{G}$. The Bacon-Shor code also serves as an important example when discussing the circuit-to-code isomorphism, see in Figure~\ref{fig:BaconShorL3} in the Appendix. Mathematically, the above discussion motivates the following equivalence relation, for $a, b \in \mathcal{A}$, $a \sim_{\mathscr{G}} b$ if there exists $g \in \mathscr{G}$ such that $a = bg$~\footnote{Note that we denote the gauge group $\mathscr{G}$ modulo the phase, which is a restricted version of the conventional definition, so that $ab = ba$. This restriction, in our case, is harmless as only the commutation relation is encoded via inner product.}.

\subsection{Main results}\label{main-subsec: main-results}
\begin{figure*}[!pt]

\includegraphics[width=1.0\textwidth]{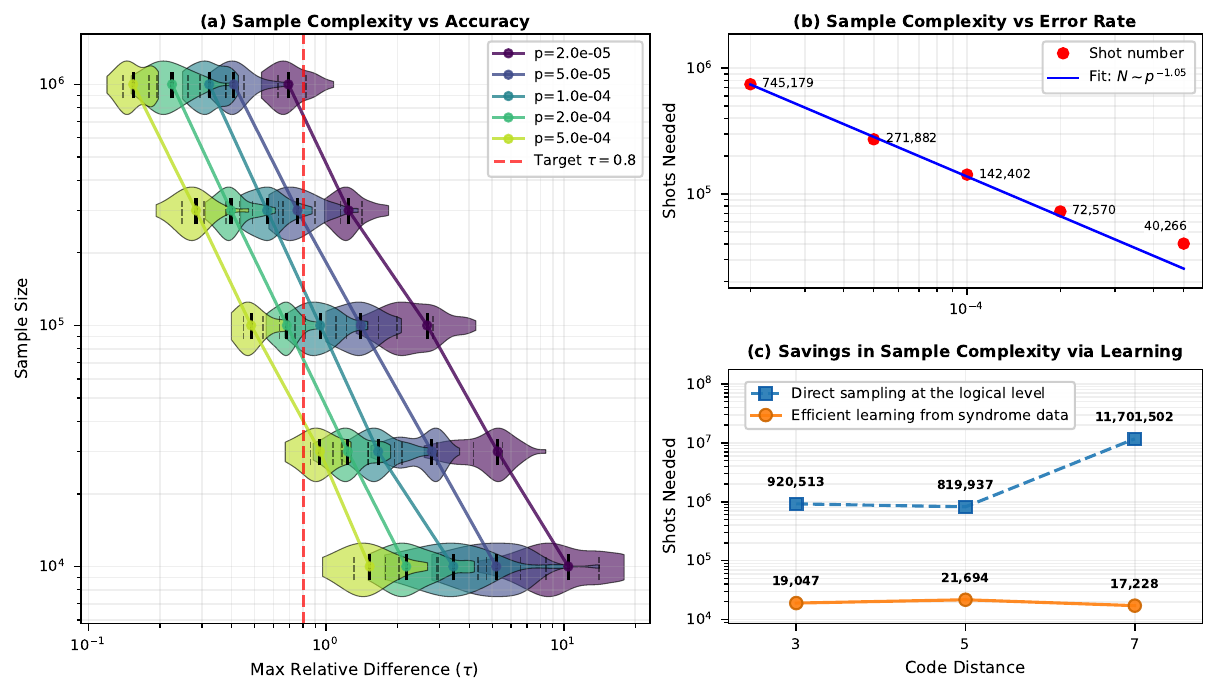}

\caption{\textbf{(a)} Sample complexity (number of shots) for learning priors (here, for the $d=5$ surface code) across several physical error probabilities, where measurement errors and data errors are modeled by the single-qubit depolarizing channel of fixed error rate $p$.  For fixed $p$, the sample size ($N$) required to reduce the maximum relative deviation between the true and predicted priors below a tolerance $\tau$ exhibits the expected scaling $N \propto \tau^{-2}$. \textbf{(b)} Fixing a target accuracy (here $\tau = 0.8$), the sample size scales approximately inversely with the error rate, $N \propto p^{-1}$ (fit shown in blue). \textbf{(c)} Sample complexity comparison at $p = 5\times 10^{-4}$ for achieving $10\%$ relative error: the predicted logical error probability requires $2$–$3$ orders of magnitude smaller sample size than the (naively) sampled logical error probability estimation for the rotated surface code at distances $d=3,5,7$.}
\label{Fig: main-simulation}
\end{figure*}

Given a fault distribution $\Lambda \in \mathbb{R}[\mathcal{A}]$, we define the \emph{effective} distribution $\Lambda^{\operatorname{eff}}$, explicitly removing the degrees of freedom associated with the gauge group 
\begin{align}\label{eq: main-eff-distribution}
    \begin{aligned}
        p^{\operatorname{eff}}_a &= \frac{1}{|\mathscr{G}|}\sum_{g \in \mathscr{G}} p_{ag}. 
    \end{aligned}
\end{align} 

Note that $p^{\operatorname{eff}}_a$ can be expressed according to Eq.~\eqref{eq: main-character-fourier-basis}, in a more complicated form. We pay special interest in the following three types of distribution: $(i)$ the physical Pauli fault distribution $\Lambda \in \mathbb{R}[\mathcal{A}]$, $(ii)$ the effective distribution $\Lambda^{\operatorname{eff}}$ modulo the logical equivalence relation $\sim_{\mathscr{G}}$ and--its closely related concept--the \emph{prior} distribution (See Eq.~\eqref{eq: main-prior-distribution}). $(iii)$ The logical error probability conditioned on an application of a fixed decoder. By fixing the decoder correction, the logical fault distribution $\Lambda^{L}$ is determined by the effective distribution (See Appendix~\ref{subsec: spacetime-logical-noise-channel}). We refer to the syndrome data by the Pauli eigenvalues associated with the measurement subgroup $\mathscr{M}$, which can be experimentally obtained via sampling the measurement outcomes of the stabilizer checks. Note that $\mathscr{M}$ can also denote the spacetime stabilizer checks, via the circuit-to-spacetime code mapping presented in Section~\ref{main-subsec: circuit-to-code-mapping}.

 Hence, it is equivalent to restricting our attention on $\mathscr{G}^\perp$ and it is often more convenient to formally state our questions taking the Fourier transform. 
\begin{enumerate}
    \item (\emph{learnability}). Given access of the syndrome data ${\mathscr{M}}$, what is the necessary and sufficient condition for which $\lambda_a$ is learnable for any $a \in \mathcal{A}$. 
    \item (\emph{learnability up to logical equivalence}). Given the access of the syndrome data  $\mathscr{M}$, what is the necessary and sufficient condition for which $\lambda^{\operatorname{eff}}_b$ is learnable for any $b \in \mathcal{A}$?
\end{enumerate}

Note that $\lambda^{\operatorname{eff}}_b = \lambda_b$ if $b \in \mathscr{G}^\perp$ and zero otherwise~\cite{Wagner_2023}. Exploiting the logical equivalence is fundamental in learning the underlying Pauli fault distribution from the syndrome measurement outcomes. To see this intuitively, for $a \in \mathcal{A}$, let $\sigma(a) \in \mathbb{F}^M_2$ denote the \emph{syndrome vector} such that $\sigma(a)_i =1$ if $a$ anticommutes with the $i$th stabilizer generator of $\mathscr{M}$ and $0$ otherwise. If $a \sim_{\mathscr{G}} b$, then it must follow that $\sigma(a) = \sigma(b) $ so that the syndrome measurement is unable to distinguish these two elements. It turns out that this intuition can be made precise in the following. 

\begin{theorem}[Learnability]\label{thm: main-physical-learnability}
    Let $\Lambda \in \mathbb{R}[\mathcal{A}]$ be a distribution such that $\lambda_b > 0$ for all $b \in \mathcal{A}$. Then any $\lambda_b$, $b \in \mathcal{A}$, can be learned from $\lambda_m, m \in \mathscr{M}$ if and only if every $a \in \mathcal{K}$ corresponds to a unique syndrome. 
    \begin{proofsketch}
        Parameterizing $\Lambda$ in Eq.~\eqref{eq: main-character-fourier-basis} and taking logarithm for the measured syndrome eigenvalues, 
        \begin{align}
            -\log \lambda_m = \sum_{a \in \mathcal{K}_{\mathscr{M}}} A_{\mathscr{M}}[m, a] (-\log (1-2q_a)).
        \end{align}
        By assumption all the Pauli eigenvalues $\lambda_b > 0$, for all $b \in \mathcal{A}$. The parameterization according to Eq.~\eqref{eq: main-character-fourier-basis} implies that the parametrization coefficient $q_a < 0.5$ so that the logarithm is well-defined and every Pauli operator $a \in \mathcal{K} $ corresponds to distinct nonzero syndrome. Hence, $\mathcal{K} = \mathcal{K}_{\mathscr{M}}$.  Note that the matrix $A_{\mathscr{M}} \in \mathbb{R}^{|\mathscr{M}| \times K}$ where $A_{\mathscr{M}}[m, a] = 1$ if they anticommute and $0$ if they commute. It is easy to see that $A_{\mathscr{M}} = (1/2) (1_{q \times K } - H_{\mathscr{M}})$ where $H_{\mathscr{M}}$ denotes the Hadamard matrix with columns restricted on $\mathcal{K}_{\mathscr{M}}$, from which one can show that $A_{\mathscr{M}}$ achieves full-column rank. Hence, the above linear system is injective and admits a unique solution. 
        
    \end{proofsketch}
\end{theorem}

\begin{theorem}[Learnability up to logical equivalence]\label{thm: main-logical-learnability}
     Let $\Lambda \in \mathbb{R}[\mathcal{A}]$ be a distribution such that $\lambda_b > 0$ for all $b \in \mathscr{G}^\perp$. Then any $\lambda_b$, $b \in \mathscr{G}^\perp$ (or equivalently, $\lambda^{\operatorname{eff}}_b$, $b \in \mathcal{A}$), can be learned from $\lambda_m, m \in \mathscr{M}$ if and only if for every pair $a, b \in \mathcal{K}$ we have that $\sigma(a) = \sigma(b) \iff a \sim_{\mathscr{G}} b$.  
     \begin{proofsketch}
     We now supplement a proof sketch, and the full proof is presented in Appendix~\ref{subsec: learnability-general-fault}. Taking the Fourier transform using the parameterization of a distribution on the character basis Eq.~\eqref{eq: main-character-fourier-basis}, one arrives at the following system of linear equations by taking the logarithm
         \begin{align}
             - \log \lambda_m = \sum_{a \in \mathcal{K}} D_{\mathscr{M}}[m, a] (-\log (1-2q_a)).
         \end{align}
         The matrix $D_{\mathscr{M}} \in \mathbb{R}^{|\mathscr{M}| \times |\mathcal{K}|}$ with $D_{\mathscr{M}}[m, a] = 1$ if they anticommute and $0$ if they commute. Let us further partition the columns of $D_{\mathscr{M}}$ as follows, 
         \begin{align}
             D_{\mathscr{M}} = \left(A_{\mathscr{M}} | B_{\mathscr{M}} \right).
         \end{align}
         The columns of $A_{\mathscr{M}} \in \mathbb{R}^{2^m \times K}$ denote independent parameters in $\mathcal{K}$ correspond to the unique syndrome associated with $\mathscr{M}$. In other words, the columns of $A$ give the subset $\mathcal{K}_{\mathscr{M}}$ with $|\mathcal{K}_{\mathscr{M}}| = K$. The columns of $B_{\mathscr{M}}$ consist of other parameters in $\mathcal{K}$ with duplicate syndromes. Similarly, we denote the row-augmented matrix $D_{\mathscr{G}^\perp}$ consisting of linear equation with $-\log \lambda_l$ for $l \in \mathscr{G}^\perp$, where under the rank-preserving column operation
        
         \begin{align*}
               D_{ \mathscr{G}^{\perp}} = \left(\begin{array}{c|c} 
  A_{\mathscr{M}} & B_{\mathscr{M}} \\ 
  \hline 
  A_{\mathscr{G}^\perp / \mathscr{M}} & B_{\mathscr{G}^\perp / \mathscr{M}} 
\end{array}\right) \mapsto \left(\begin{array}{c|c} 
  A_{\mathscr{M}} & \bvec{0} \\ 
  \hline 
  A_{\mathscr{G}^\perp / \mathscr{M}} & B'_{\mathscr{G}^\perp/ \mathscr{M}} 
\end{array}\right).
         \end{align*}
         Note that $B'_{\mathscr{G}^\perp / \mathscr{M}}$ is zero if and only if the parameters with duplicate syndromes in $B$ are related to those in $A$ via elements in the gauge group $\mathscr{G}$. From Theorem~\ref{thm: main-physical-learnability}, $A_{\mathscr{M}}$ is full-column rank. As a direct result, any $-\log \lambda_l$ for $l \in  \mathscr{G}^\perp$ lies in the row span of the $A_{\mathscr{M}}$, which concludes the theorem. 
     \end{proofsketch}
\end{theorem}

Our learnability framework strengthens that from~\cite{Wagner_2022, Wagner_2023} by providing necessary and sufficient conditions for learnability (up to a logical equivalence). Though at first glance, the assumption enabling the learnability in the two frameworks takes a noticeable difference, the present framework strictly encompasses the latter. For instance, for a Pauli channel which can be written as a convolution of local ones, the sufficient condition for learnability up to logical equivalence in~\cite{Wagner_2023} is that the union of supports of two component local Pauli channels is smaller than the code distance. Hence, for any such two local Pauli operators $a$ and $b$, if they correspond to the same syndrome, then they must be related by elements in the gauge group. In Section~\ref{subsec: parameterization-fault-distribution} of the Appendix, we give a rigorous mathematical connection between these two formalisms through the lens of M\"{o}bius transformation and Zeta functions on the Boolean group. A closely related concept, especially in the decoder community, is referred to as the prior distribution $\Lambda^{\operatorname{prior}}$. The prior distribution removes the degrees of freedom of the physical fault processes by aggregating all independent fault processes with the same syndromes. In this regard, prior distribution is always learnable from syndrome data as a simple application of Theorem~\ref{thm: main-physical-learnability}, which is recently captured in Ref.~\cite{blumekohout2025estimatingdetectorerrormodels, arms2025estimatingdetectorerrormodels}. The prior distribution $\Lambda^{\operatorname{prior}}$ is given by 
\begin{align}\label{eq: main-prior-distribution}
    \Lambda^{\operatorname{prior}} = *_{c \in \mathcal{K}_{\mathscr{M}}}((1-q_{[c]}) + q_{[c]} \chi_c ),
\end{align}
which can always be constructed from $\Lambda$, with coefficients
\begin{align}\label{eq: main-prior-distribution-coeff}
    q_{[c]} = \frac{1-\prod_{a \in \mathcal{K}: \sigma(a) = \sigma(c)} (1-2q_a)}{2}, 
\end{align}
given by aggregating all the coefficients $q_a$ of the same syndrome from the parameterization Eq.~\eqref{eq: main-character-fourier-basis}.  If the coefficients of the parametrization Eq.~\eqref{eq: main-character-fourier-basis} are probabilities such as the case of Pauli-Lindblad form, We could simply view $q_{[c]}$ in Eq.~\eqref{eq: main-prior-distribution-coeff} as aggregating the probability of all the independent fault processes sharing the same syndrome as $c$. For example, $(1- (1-2q_c)(1-2q_{c'}))/2 = q_c + q_{c'} - 2 q_c q_{c'}$, which precisely indicates the probability of fault of syndrome $\sigma(c) = \sigma(c') $ that is triggered by $c$ or $c'$. This explains our notation $q_{[c]}$ which represents the equivalence class of fault that shares the same syndrome as $c$. Since the prior distribution compresses degrees of freedom of the physical fault processes within $\mathcal{K}_{\mathscr{M}}$, one can simply regard it as a distribution over syndrome errors. Intuitively speaking, the prior distribution treats faults that share the same syndrome as the same fault. If fault $b$ differs from fault $a$ by some logical action, these are logically inequivalent. Hence, prior distribution in general cannot predict true logical error probability or the effective error probability. A necessary and sufficient condition when it can is precisely given by $\sigma(a) = \sigma(b) \iff a \sim_{\mathscr{G}} b$ for $a, b \in \mathcal{K}$. That is, independent fault processes that share the same syndrome must be logically equivalent. In section~\ref{main-subsec: lep-sampling-complexity} we rigorously prove this intuitive statement.

Besides learnability, another fundamental question is the sample complexity. There are exponentially many possible stabilizer measurement outcomes associated with the dimension of the stabilizer group. On the other hand, the fault process patterns are typically sparse: if one assumes local errors or a sparse Pauli-Lindblad model \cite{chen2025disambiguatingpaulinoisequantum, van_den_Berg_2023, van_den_Berg_2024}. Hence, a natural question can be asked whether we could learn any sparse fault process computationally efficiently--with only a similar scaling of degrees of freedom to that of the physical noise process. Since the stabilizer group forms a (Abelian) group, one could approach this question again in the same representation-theoretical setting. For a typical Pauli fault channel with limited correlations, the parameterization set $\mathcal{K}$ typically has size polynomial scaling with the system size. A central question to this end is whether or not there exists an efficient protocol that avoids utilizing all the syndrome expectations, which would quickly become intractable. A simple question can be framed that if subsampling also polynomially many checks would ensure the matrix $A_{\mathscr{M}}$ to be of full-column rank. More precisely, denote a set $Q \subset \mathscr{M}$ with $|Q| = q$ and denote the row-subsampled restricted matrix $A_{\mathscr{M}, \operatorname{res}}$, we wish to show that $A^T_{\mathscr{M}, \operatorname{res}} A_{\mathscr{M}, \operatorname{res}}$ has full rank. Mathematically, it relates to a question of low-rank reconstruction, which typically requires the underlying matrix to be structured. Since our distribution is parameterized in the character basis, $A_{\mathscr{M}}$ bears a close relation to the Fourier (Hadamard) matrix. Indeed, note that $K = |\mathcal{K}_{\mathscr{M}}|$

\begin{theorem}[Informal]\label{thm: main-row-subsampled-rip-A}
Let $Q \subset \mathscr{M}$ be a uniform subsampled set with replacement such that it contains $q = O\left(\delta^{-2}_K K \max( \log^2(K/\delta_K) \log^2(1/\delta_K) \log(K)), \log(1/\delta))\right)$ many samples. Then with probability at least $1- \delta$, $A_{\mathscr{M}, \operatorname{res}}$ achieves full column rank or equivalently $A^T_{\mathscr{M}, \operatorname{res}}A_{\mathscr{M}, \operatorname{res}}$ has full rank. 
\begin{proofsketch}
    It is easy to see that $A_{\mathscr{M}, \operatorname{res}} = (1/2) (1_{q \times K } - H_{\mathscr{M}, \operatorname{res}})$, where $H_{\mathscr{M}, \operatorname{res}}$ denotes the row- and column-restricted Walsh-Hadamard matrix $H$. In this case, it suffices to show that $H_{\mathscr{M}, \operatorname{res}}$ is full-column rank or even more strongly, the so-called restricted isometry property, for any $x$
\begin{align}
    q(1 - \delta_K) \|x\|_2 \leq  \|Hx\|_2 \leq q(1 + \delta_K) \|x\|_2. 
\end{align}
The constant $\delta_K$ is called the ($K+1$th-order) restricted isometric constant, and $A_{\mathscr{M}, \operatorname{res}}$ achieves full-column rank whenever $\delta_K$ is less than $1$. Using the techniques from the compressed sensing~\cite{haviv2015restrictedisometrypropertysubsampled} with minimal adaptation, one can show that with high probability, $\delta_K$ can be chosen (much) smaller than one. The full proof is presented in Section~\ref{subsec: compressed-sensing-theory}. 
\end{proofsketch}
\end{theorem}

Thanks to our formalism with the representation theory of finite Abelian groups, Theorem~\ref{thm: main-physical-learnability}, Theorem~\ref{thm: main-logical-learnability}, and Theorem~\ref{thm: main-row-subsampled-rip-A} apply indiscriminately to both the phenomenological Pauli faults and circuit-level Pauli faults. In practice, the syndrome expectation is only estimated up to some specified precision. Hence, it is crucial to ensure that the sampling error (imprecision) would be reasonably controlled through propagation along the estimation protocol, while we still ensure an efficient sample complexity overhead. A naive sampling protocol using the Hoeffding bound to a relative precision $\tau$ (which could be a constant term) utilizes $O(\varepsilon^{-2})$ number of samples. This sample complexity, however, is far from  optimal. Since each measurement outcome records the commutation parity of a given check, it can be regarded as a Bernoulli random variable drawn from $\operatorname{Bern}(\varepsilon)$. This simple observation provides an optimal, variance-aware, sample-complexity scaling, motivated from Ref.~\cite{lee2025efficientbenchmarkinglogicalmagic}. Utilizing the protocol and proving the well-conditionedness of the learning protocol in Section~\ref{subsec: error-propagation-compressed-sensing}, one could efficiently learn the \emph{prior} distribution~\cite{Sivak2023RealTimeQEC, Sivak_2024} (See Eq.~\eqref{eq: main-prior-distribution}) to compute the logical error probability, with optimal sample complexity (See Figure~\ref{Fig: main-simulation} \textbf{(a)} \& \textbf{(b)}). 

\begin{theorem}[Informal]\label{thm: background-free-estimation-prior}
  Suppose $\Lambda$ is learnable up to the logical equivalence. Then any coefficient of the prior distribution $\Lambda^{\operatorname{prior}}$ can be estimated to an additive precision $O(\varepsilon)$, consuming $\Theta(1/\varepsilon)$ many samples of syndrome measurements. 
\end{theorem}

The proof is given formally in Section~\ref{subsec: error-propagation-compressed-sensing} and cultivates in Theorem~\ref{thm: sample-complexity-prior-distribution}. All in all, combining the above techniques, we give an end-to-end protocol on estimating logical error probability (LEP), which is summarized in Algorithm~\ref{algo: end-2-end-framework-from-syndrome}.

\setlength{\algomargin}{0.8em}
\SetNlSkip{0.8em}
\begin{protocol}{Estimating logical errors from syndrome}{alg:estimation-from-syndrome}\label{algo: end-2-end-framework-from-syndrome}
  \KwIn{The syndrome extraction circuit $C_T(\mathcal{S})$ that implements a base stabilizer group $\mathcal{S}$. The set of independent parameters in $\mathcal{K}_\mathscr{M}$ from the unknown Pauli fault distribution $\Lambda$. The subsampled $Q \subset \mathscr{M}$ with $A_{\mathscr{M}, \operatorname{res}}$. A fixed decoder records the corrections to apply conditions on the syndromes as a dictionary $\{z \in \mathbb{F}^M_2 = \sigma(\mathscr{M}): a_{z}\}$. }
  \KwResources{A syndrome extraction circuit in Definition~\ref{def: main-syndrome-extraction-circuit}. }
  \KwPre{Check $\mathcal{K}$ and $\mathscr{M}$ to ensure the logical learnability condition is met in Theorem~\ref{thm: main-logical-learnability}. A pre-determining scheduling in the syndrome extraction circuit and measurement sequence $\mathcal{M}_{T-0.5}, \cdots, \mathcal{M}_{0.5}$.}
  \KwGoal{Logical error probability $p^{L}_{\overleftarrow{l}}$ associated with the logical operator $\overleftarrow{l}$.}

  \Round{1}
  \Sample{
  \tcp*[r]{Suppose that the physical error rate $\varepsilon$ and we repeat the circuits $S = O( \varepsilon^{-1})$ shots.}
   For each shot $\ell=1, \cdots, S$ \; 
   \vspace{1 mm}
   \quad For $j=1, \cdots, M$ and the tuple $(t_j, i_j)$ such that $t_j =0, 1, \cdots, T-1$ and $i_j =1, \cdots, |\mathcal{M}_{t_j +0.5}|$\;
   \vspace{1 mm} 
   \quad \quad \Measure parity $\{0, 1 \}$ outcome $s^{(\ell)}_j$ \; 
    \quad For each element in $m \in \mathcal{Q} \subseteq \mathscr{M}$, \; 
    \quad \quad postprocess with the parity $\sum^M_{r=1} (s^{(\ell)}_j)^{\mu_r}$ where $\mu \in \mathbb{F}^M_2$ indicates the decomposition $m$ in terms of the generators (checks). \; 
    \vspace{1 mm}
    Compute the estimator
    \begin{align*}
        \bar{\lambda}_{m} = (1/S) \sum^S_{\ell =1} (-1)^{\sum^M_{r=1}(s^{(\ell)})^{\mu_r}}
    \end{align*}
    and store the values in a vector $\bar{y}_{m} = - \log \bar{\lambda}_{m} \in (0, 1]^{q}$.
 }
  \Round{2}
   \Sample{\tcp*[r]{Proceed with offline calculation with the learnability framework. Estimate the coefficients of the prior distribution. }
   
   Solve the recovery problem from Lemma~\ref{lemma: main-noisy-recovery-all-distribution} (Lemma~\ref{lemma: noisy-recovery-all-distribution})
   $$
    \bar{y} = A_{\mathscr{M}, \operatorname{res}} \bar{x},  
   $$
   from techniques introduced in Lemma~\ref{lemma: main-noisy-recovery-prior} (Lemma~\ref{lemma: noisy-recovery-prior}). \; 
   \vspace{1 mm}
   Update the estimated coefficients from $\Lambda^{\operatorname{prior}}$ by $q_{[a]} = \frac{1}{2}(1 - e^{-\bar{x}_a})$. \;
   }

   \Round{3}
   \tcp*[r]{Compute the logical probability $p^{L}_{\bar{l}}$ with the fixed decoder correction terms}
   
   Pick any logical operator $\overleftarrow{l}$\; 
   \KwOut{$ p^{L}_{\bar{l}}$ computed from Eqs.~\eqref{eq: main-prior-distribution-logical-error} \&~\eqref{eq: main-logical-error-probability}.}
   
\end{protocol}

\section{Learnability of syndrome circuit faults}\label{main-sec: circuit-learnability}

The above section pinpoints the mathematical formalism from Fourier analysis on some Abelian group $\mathcal{A}$. This abstraction is provided in Section~\ref{sec: learnability-general}, where it unifies the discussion with various quantum/classical codes and quantum circuits alike. We now make it concrete in the setting of the Clifford circuits. A \emph{spacetime} Pauli group for a Clifford circuit is given by $\mathcal{A} : = \mathcal{A} \cong  \otimes^T_{i=0} \mathrm{P}^{n}$, where $a \in \mathcal{A}$ we can represent it by $a = \eta_T(a) \eta_{T-1}(a) \cdots \eta_0(a)$, modulo the phase. 

\subsection{Pauli propagation and invariance}\label{main-subsec: pauli-porpagation}

For a quantum circuit consisting of logical Pauli measurements and Clifford unitaries, the action of these instruments would deform and propagate the stabilizers and gauge groups. In other words, for any Pauli operator at the circuit initialization, one can trace the trajectory of such Pauli operator as the worldlines, as an element of the spacetime Pauli group $\mathcal{A}$. This encompasses the main model of interest: the syndrome extraction circuits.

\begin{definition}[Syndrome extraction circuit]\label{def: main-syndrome-extraction-circuit}
    The syndrome extraction circuit $\mathcal{C}_T(\mathcal{S})$ for an even number $T$ encodes a base stabilizer code with the stabilizer group $\mathcal{S}$ at initialization $t=0$, with the following rules. 
    \begin{itemize}[label=--]
        \item For each integer time slice, we consider an $n$-qubit system (or an $(n+r)$-qubit system that includes $r$ ancilla qubits, one for each stabilizer generator).  
        \item For each half-integer time slice $t+0.5=0.5, \cdots, T-0.5$, we apply logical Clifford unitaries, logical Pauli measurements, stabilizer measurements, which act on qubits at time $t$. Furthermore, we assume that for each time $t$, these operations act on disjoint sets of qubits~\cite{delfosse2023spacetimecodescliffordcircuits}.
    \end{itemize}
\end{definition}
The formalism of the syndrome extraction circuit in Definition~\ref{def: main-syndrome-extraction-circuit} is directly motivated by the circuit-to-code mapping discussed from~\cite{delfosse2023spacetimecodescliffordcircuits, bacon2017sparse}. For the simplicity of discussion, we follow the treatment in~\cite{delfosse2023spacetimecodescliffordcircuits} where we discard the ancillary worldlines. We provide a systematic analysis in Appendix~\ref{subsec: code-circuit-iso} and recall the following notation analogous to those of~\cite{bacon2017sparse, delfosse2023spacetimecodescliffordcircuits}. In the case where the dominant types of faults consist of data errors and ancilla measurement errors, this simplification is done without loss of generality. In the syndrome extraction circuit, we have assumed that the stabilizer group is preserved at all instantaneous time steps and the output code at time $t=T$ must also return to the base stabilizer code with stabilizer group $\mathcal{S}$. A key notion of the Clifford spacetime circuit is the propagation of Pauli frames: since gates are assumed to be Clifford, one can propagate Pauli channels freely. We also give another version of the spacetime code mapping, appending the $r$ ancilla qubits with a similar circuit-to-code mapping (See Theorem~\ref{thm: circuit-subsystem-code-ancilla}). For $a \in \mathcal{A}$ we can write it more explicitly, 
\begin{align*}
    a = \eta_T(a) \cdots \eta_1(a) \eta_0(a),
\end{align*}
where $\eta_t(a)$ extracts the components in the tensor product corresponding to the time slice $t=0, \cdots, T$. 

\begin{definition}[Forward and backward propagation]
    For any element $a \in \mathcal{A}$,  Clifford circuit $\mathcal{C}_T(\mathcal{S})$, and time $i$ with $0 \leq i \leq T$, $u_{(i + 1, i)}$ denotes the Clifford unitary acting from time $i$ to $i+1$.  Denote the Clifford unitaries acting from time $i$ to time $j$ by $u_{(j, i)}=u_{(j, j-1)} \circ \cdots \circ u_{(i+1, i)}$ for $j \geq i+1$ (otherwise, it is an identity). Then the forward propagation and the backward propagation are respectively defined by

$$
\eta_t(\overrightarrow{a})=\prod_{i=0}^t u_{(t, i)}\left(\eta_i(a)\right); \quad \eta_t(\overleftarrow{a})=\prod_{i=t}^T u_{(T, i)}^{-1}\left(\eta_i(a)\right),
$$
for $T\geq t \geq 0$.
\end{definition}
We denote $\mathcal{M}_{t+0.5}$ by the subgroup generated by the stabilizer checks (generators) measured at time $t$. Furthermore, we can define a notion of inner product over the spacetime Pauli group $\mathcal{A}$ induced by that from each time slice, namely $\langle a, b \rangle = \sum^T_{i=0}\langle \eta_i(a), \eta_i(b) \rangle \mod 2$. With slight notational abuse where $\langle \cdot, \cdot \rangle$ also denotes (symplectic) inner product over the spatial Pauli group. The backward propagation and forward-propagation are intrinsically related by the inner products, 

\begin{proposition}[Corollary~1 and Proposition~3 in Ref.~\cite{delfosse2023spacetimecodescliffordcircuits}]
    The propagation operators have the following properties, 
    \begin{itemize}[label=--]
        \item The cumulant operations are automorphism within $\mathcal{A}$ which satisfies
        \begin{align}\label{eq: cumulant-automorphism-spacetime}
            \overleftarrow{ab} = \overleftarrow{a}\overleftarrow{b}; \quad  \overrightarrow{ab} = \overrightarrow{a}\overrightarrow{b},
        \end{align}
        for  $a, b \in \mathcal{A}$. 

        \item The forward and backward propagation operators are related by, 
        \begin{align}
            \langle \overleftarrow{a}, b \rangle  = \langle a, \overrightarrow{b}\rangle. 
        \end{align}
    \end{itemize}
\end{proposition}
The proof is given in Ref.~\cite{delfosse2023spacetimecodescliffordcircuits} and restated in Proposition~\ref{prop: spacetime-cumulant-automorphism}. We may understand the inner product relation by analogy with $4$-momentum conservation in particle physics. At some time slice $t$, we perform some measurement for $\eta_t(m)$ with $m \in \mathcal{M}_{t+0.5}$. It would detect any fault $a$ that has been propagated to at time $t$, if $\langle \eta_t(m), \eta_t(\overrightarrow{a}) \rangle = 1$. Such a fault $a$ may be induced earlier in time, say $t=0$. In this case we have that $\langle \eta_t(m), \eta_t(\overrightarrow{a}) \rangle = \langle \eta_t(m), u_{(t, 0)}(a)\rangle = \langle u^{-1}_{(t, 0)}(\eta_t(m)), u^{-1}_{(t, 0)} \circ u_{(t, 0)}(a) \rangle = \langle \eta_0(\overleftarrow{m}), \eta_0(a) \rangle$. Finally this implies $\langle \overleftarrow{m}, \eta_0(a)\rangle=1$, taking the backward propagation on the measurement operator. We may think  $\langle \eta_t(m), \eta_t(\overrightarrow{a}) \rangle$ is evaluated at the laboratory frame, and $\langle \overleftarrow{m}, \eta_0(a)\rangle=1$ is evaluated at the particle rest frame, where the fault is instantaneously created. The invariance $\langle \overleftarrow{m}, \eta_0(a)\rangle=1 = \langle \eta_t(m), \eta_t(\overrightarrow{a}) \rangle$ takes the analogy of momentum conservation regardless of which frames one observes. We carefully examine the change-of-frames in relation to the logical error channel in Section~\ref{subsec: spacetime-logical-noise-channel}, and we choose to always work in the rest frame throughout.

\subsection{Circuit-to-code mapping}\label{main-subsec: circuit-to-code-mapping}

We now apply the main techniques from Section~\ref{main-sec: main-results} to a quantum circuit $\mathcal{C}_T(\mathcal{S})$ implementing a base stabilizer code with stabilizer group $\mathcal{S}$.  In the setting of the syndrome extraction circuit, one can formally view it as a subsystem code. 

\begin{theorem}[Circuit-to-code isomorphism]\label{thm: main-circuit-to-code-iso}
Let $\mathcal{S}$ be a base stabilizer group that encodes $k$ logical qubits. Let $\mathcal{C}_T(\mathcal{S})$ be a syndrome extraction circuit that implements $\mathcal{S}$, terminating at time $T$. Then there exists a constructive code isomorphism $\mathcal{C}_T(\mathcal{S}) \mapsto \mathscr{Q} = (\mathscr{G}, \mathscr{M})$, where $\mathscr{Q} = (\mathscr{G}, \mathscr{M})$ denotes the subsystem code defined on the spacetime Pauli group $\mathcal{A}$, with code parameters encoding precisely $k$ logical qubits and $T$ is chosen to be even integer. 
\end{theorem}

The spacetime stabilizer group or the measurement group $\mathscr{M}$ is generated by 

\begin{align}\label{eq: main-sapcetime-measuresment-group}
    \mathscr{M} = \langle \overleftarrow{\eta_t(m)} \in \mathcal{M}_{t+0.5}; \: t=0, \cdots T-1 \rangle \cdot \langle \overleftarrow{\eta_T(s)}, \: s \in \mathcal{S} \rangle,  
\end{align}
where $\{ \overleftarrow{\eta_T(s)}, \: s \in \mathcal{S} \}$ assumes a layer of perfect syndrome measurements at the terminal time. The proof detail is presented in Appendix~\ref{subsec: code-circuit-iso}. The intuition for the above construction of isomorphism relies on explicitly constructing the symmetry of $\mathcal{C}_T(\mathcal{S})$ so that it matches with that of the base stabilizer code with stabilizer group $\mathcal{S}$. In particular, the gauge group is constructed by finding all elements (of spacetime Pauli group) whose forward propagation lies in the stabilizer group $\mathcal{S}$ at the terminal time and also commute (undetected) from every intermediate stabilizer checks. Evidently, this means that for any $g \in \mathscr{G}$, $\eta_T(\overrightarrow{g}) \in \mathcal{S}$.

A key distinction to the (more versatile) Clifford circuit-to-code mapping presented in \cite{delfosse2023spacetimecodescliffordcircuits} is the omission of logical Pauli measurements or, more generally, non-commuting Pauli measurements as spacetime checks. This can be incorporated in our setting by modifying the gauge group construction to Theorem~\ref{thm: main-circuit-to-code-iso} (see in Section~\ref{subsec: code-circuit-iso}). The main rationale behind treating logical measurement as a spacetime checks is that by performing a logical Pauli measurement at some time $t+0.5$, the instantaneous stabilizer group at $t$ would contain the logical operator for which it measured. However, since the logical measurements preserve the stabilizer group by construction, if one only intends to record the stabilizer measurement outcomes, there is no harm in including the logical Pauli measurement (or any logical measurements) in Definition~\ref{def: main-syndrome-extraction-circuit} without treating them as spacetime checks.

\begin{remark}
In many practical settings, one may wish to construct detector error models~\cite{derks2024designing, delfosse2023spacetimecodescliffordcircuits, blumekohout2025estimatingdetectorerrormodels}. In our case, the detectors can be defined by defining a set of $\mathbb{F}_2$-linear equations from the measurement subgroup $\mathscr{M}$, in addition to spacetime checks forming from--in this case--the logical Pauli measurements. This observation is similarly captured by the recent work~\cite{pesah2025faulttoleranttransformationsspacetimecodes}. In Refs.~\cite{delfosse2023spacetimecodescliffordcircuits, derks2024designing} one could think to construct spacetime code directly from these detectors. These constructions differ here due to different purposes: since we only care about the syndrome information, we simply do not treat any logical Pauli measurements as spacetime checks. One may wish to generalize the circuit-to-code mapping by treating the mid-circuit Pauli measurements or any mid-circuit measurements as spacetime checks, if we wish to leverage more information than the syndromes. We describe a way to incorporate this into the circuit-to-code mapping in Appendix~\ref{subsec: code-circuit-iso} and leave a systematic study for future work. 
\end{remark}

\begin{figure*}[pt]
\centering

\includegraphics[width=0.9\textwidth]{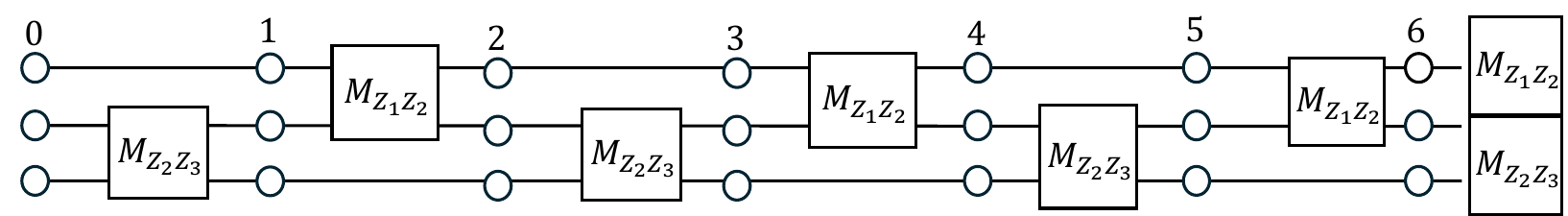}

\caption{A schematic illustration of the repetition code with three rounds of syndrome extraction. As in Eq.~\eqref{eq: rep-circuit-stabilizer-no-ancilla}, there are 8 independent stabilizer generators given in backward propagation of $6$ mid-circuit stabilizers and $2$ output stabilizers. A schematic plot featuring ancillas for the stabilizer measurements is given in Figure~\ref{fig: rep-circuit-ancilla}.}

\label{Fig: main-rep-circuit-spacetime-example}
\end{figure*}

\begin{example}[3 rounds of syndrome extraction with repetition codes]
     Let us consider we have $3$ rounds of syndrome measurements and our measurement group
    \begin{align}
        \mathcal{M}^{\text{even}}_{i+0.5}  = \langle Z_1Z_2 \rangle; \quad  \mathcal{M}^{\text{odd}}_{i + 0.5}  = \langle Z_2Z_3 \rangle.
    \end{align}
    For $i=0, 1, \cdots, 5$. We now construct our gauge group. First, we consider the following subgroup  
    \begin{align}\label{eq: main-rep-code-examples-pauli-transports}
        \begin{aligned}
            &\mathcal{G}^{\text{even}}_{(i+1, i)}  = \langle g(X_3, i), g(X_1X_2X_3, i) , g(Z_1, i), g(Z_2, i), g(Z_3, i) \rangle, \\
        &\mathcal{G}^{\text{odd}}_{(i+1, i)}  = \langle g(X_1, i), g(X_1X_2X_3, i) , g(Z_1, i), g(Z_2, i), g(Z_3, i) \rangle. 
        \end{aligned}
    \end{align}
    Then the gauge group $\mathscr{G} \subseteq \mathrm{P}^{21}$ can be constructed 
        \begin{align}\label{eq: main-rep-code-spacetime-gauge-group}
           \mathscr{G} = \langle \mathcal{G}_{(i+1, i)}, \eta_i(\mathcal{M}_{i+0.5}) : i=0, \cdots 5 \rangle,
        \end{align}
        where we denote $\eta_i(\mathcal{M}_{i+0.5})$ a shorthand notation as $\{\eta_i(m), \forall m \in \mathcal{M}_{i+0.5}\}$. The center of the gauge group $\mathscr{M} = \mathscr{G}^\perp \cap \mathscr{G}$ can be given by 
            \begin{align}\label{eq: rep-circuit-stabilizer-no-ancilla}
            \left\{ 
            \begin{aligned}
            &\overleftarrow{\eta_6(Z_1Z_2)}, \overleftarrow{\eta_6(Z_2Z_3)},  \overleftarrow{\eta_5(Z_2Z_3)}, \overleftarrow{\eta_4(Z_1Z_2)}, \\
            &\overleftarrow{\eta_3(Z_2Z_3)}, \overleftarrow{\eta_2(Z_1Z_2)},
            \overleftarrow{\eta_1(Z_2Z_3)}, \overleftarrow{\eta_0(Z_1Z_2)}
            \end{aligned}  \right\},
        \end{align}
which gives $8$ desired generators. Note that by the dimension counting over symplectic space, the gauge group has $ 5 \times 6 + 2 = 32$ linearly independent operators so that $\dim \mathscr{G}^\perp = 42 - 32 = 10$. Then it follows that $\dim \mathscr{G}^\perp / \mathscr{M} = 10 - 8 = 2$, which encodes precisely one logical qubit. It is easy to see that the choice of such a logical representative is given by 
\begin{align}
    \overleftarrow{\eta_6(X_1X_2X_3)}, \quad \overleftarrow{\eta_6(Z_1)}. 
\end{align}
\end{example}

 To apply the learnability statements, Theorem~\ref{thm: main-physical-learnability} and Theorem~\ref{thm: main-logical-learnability}, let $\mathcal{C}_{T}(\mathcal{S})$ be represented by a subsystem code $(\mathscr{G}, \mathscr{M})$ over the spacetime group $\mathcal{A}$ and $\Lambda$ a fault distribution on $\mathcal{A}$. The gauge transformation would reduce to the stabilizer deformation on the final output stabilizer code, given by the final round of (perfect) stabilizer measurements. 

\begin{corollary}
    Let $(\mathscr{G}, \mathscr{M})$ be the circuit-to-code isomorphism in Theorem~\ref{thm: main-circuit-to-code-iso}. For two circuit Pauli faults $a, b \in \mathcal{A}$, their final accumulated errors at $T$ are given respectively by $\eta_T(\overrightarrow{a})$ and $\eta_T(\overrightarrow{b})$. Then, 
     \begin{enumerate}
         \item If $\sigma(a) = \sigma(b)$, then $\eta_T(\overrightarrow{a})$ and $\eta_T(\overrightarrow{b})$ must have the same syndrome associated with the output stabilizer group $\mathcal{M}_{T+0.5} = \mathcal{S}$.

         \item For any $a \sim_{\mathscr{G}} b$, that is, there exists some element in the gauge group $g \in \mathscr{G}$ such that $a = b  g$. Then there always exists some stabilizers $s \in \mathcal{S} = \mathcal{M}_{T+0.5}$ such that $\eta_T(\overrightarrow{a}) = \eta_T(\overrightarrow{b} )\eta_T(s)$. 
     \end{enumerate}
    
\end{corollary}

Since the learnability dictates that for any $a, b \in \mathcal{K}$, $\sigma(a) = \sigma(b) \iff a \sim_{\mathscr{G}} b$. Then it is easy to see that from the above $\sigma(a) = \sigma(b) \iff \eta_T(\overrightarrow{a}) \sim_{\mathcal{S}} \eta_T(\overrightarrow{b})$ at the output time $T$. Hence, if the propagation of all $\Lambda$ is learnable up to logical equivalence at terminal time $T$ with an output (base) stabilizer code, then the circuit fault is strictly learnable up to a logical equivalence. Note that the converse in general does not hold. 

Fault tolerance is typically associated with the notion of distance, describing how many and how large (as in weight) faults it can in principle correct. For a base stabilizer code $\mathcal{S}$ encoding $k$ logical operators with distance $d$. It is natural to ask what is the notion of distance of the circuit-to-code mapping $C_T(\mathcal{S}) \mapsto (\mathscr{G}, \mathscr{M})$. For any subgroup (subspace) $\mathcal{H} \leq \mathcal{A}$, we denote $\mathcal{H}^\perp$ to its orthogonal complement space with respect to the $\langle \cdot, \cdot \rangle$, i.e., elements that commute with $\mathcal{A}$. Hence, the \emph{bare logical subspace} of $(\mathscr{G}, \mathscr{M})$ is given $ \mathscr{L}_{\operatorname{bare}} := \mathscr{G}^\perp / \mathscr{M} $ and the \emph{dressed logical subspace} $\mathscr{L}_{\operatorname{dressed}} := \mathscr{M}^\perp / \mathscr{G}$. It is noted that we have canonical isomorphism $\mathscr{M}^\perp / \mathscr{G} \cong \mathscr{G}^\perp / \mathscr{M}$. The bare logical distance refers to the minimal weight representation of the bare logical subspace. From Theorem~\ref{thm: main-circuit-to-code-iso} (see the full proof in Theorem~\ref{thm: circuit-subsystem-code-1} in Section~\ref{subsec: code-circuit-iso}), $ \mathscr{L}_{\operatorname{bare}} = \{ \overleftarrow{\eta_T(l)}\}$ for logical operators defined from the base stabilizer group $\mathcal{S}$. Hence, it is clear that the bare distance is always lower-bounded by the distance of the base stabilizer group. It then remains to bound the dressed distance, as the minimal weight element from $\mathscr{L}_{\operatorname{dressed}}$. We leave the systematic investigation of bounding the dressed distance to future work, where it involves detailed analysis on the fault tolerance of individual circuit components. In what follows, we consider a simple example, assuming the circuit consists of strictly transversal implementation of Clifford unitaries and Pauli measurements. In this case, we have that $|\eta_{t+1}(\overrightarrow{a}) | \leq |\eta_{t }(a)| $ for $t=0, \cdots, T-1$, for any $a \in \mathcal{A}$. Let us denote $\Phi: \mathscr{M}^\perp \rightarrow \mathcal{S}^{\perp} / \mathcal{S}$ by $\Phi(a) = \eta_T(\overrightarrow{a}) \eta_T(s)$ for any $s \in \mathcal{S}$. We claim that $\Phi$ is constant over $\mathscr{G}$, which follows from our intuitive statement that the terminal contribution of elements from the spacetime gauge group lies in the stabilizer group (Formally see Eq.~\eqref{eq: spacetime-gauge-group} and see in Eq.~\eqref{eq: main-rep-code-spacetime-gauge-group} for an example). Hence modulo over $\mathcal{S}$, it must evaluate to identity. Recall that the dressed logical space is $\mathscr{M}^\perp/\mathscr{G}$ and note that $\mathscr{M}^\perp = \mathscr{G} \cdot \mathscr{G}^\perp = \mathscr{G} \cdot \mathscr{M} \cdot \mathscr{L}_{\operatorname{bare}} = \mathscr{G} \cdot \mathscr{L}_{\operatorname{bare}}$. Since $\Phi$ is a well-defined homomorphism, we conclude that for any $a \in \mathscr{M}^\perp / \mathscr{G}$, $|\Phi(a)| \geq d$. Finally by the above assumption, $ d\leq  |\Phi(a)| \leq |\eta_T(\overrightarrow{a})| \leq  |a|$, whenever $a \notin \mathscr{M}$. Hence, the dressed distance is at least the code distance of the base stabilizer group.

\begin{example}
[Fault-tolerant circuit]
A sufficient condition analogous to that from Ref.~\cite{Wagner_2023} can be given: the circuit Pauli faults are learnable up to the logical equivalence if the Pauli fault distributions are local in space and time coordinates such that the support of any $a\in \mathcal{K}$ is less than the $\lfloor d/2\rfloor$ where $d$ denotes the dressed distance of the spacetime subsystem code $(\mathscr{G}, \mathscr{M})$. A typical example assumes a time-stationary fault distribution that is also locally correlated. Then one could learn $\Lambda$ up to logical equivalence and estimate the logical error probability (LEP) from syndrome data alone. This is numerically demonstrated in Figure~\ref{Fig: main-simulation}. One might worry that the terminal (perfect) syndrome measurement is unrealistic in such cases. This can be addressed by assuming final traversal measurements typical in memory experiments. Note that since the terminal logical measurements and the terminal syndrome measurements are similarly extracted from faulty transversal measurements, such terminal measurement faults do not pose an obstacle. Hence, if our goal is to utilize the syndrome data to infer terminal logical measurements obtained from terminal transversal measurements, the same learnability condition as in Theorem~\ref{thm: main-circuit-to-code-iso} could still apply. 
\end{example}

\subsection{Estimating the unlearnable degrees of freedom}\label{main-subsec: main-unlearnable-dof}

 Suppose that we perform the final transversal measurements from which we infer the logical measurements and terminal syndrome measurements. Then we only arrive at the case of unlearnability with independent fault processes $a \in \mathcal{K}$ such that $a = \Omega(d)$ where $d$ denotes the dressed distance associated with $(\mathscr{G}, \mathscr{M})$. These faults (circuit catastrophic error) are typically assumed to be extremely rare to occur and it is therefore desirable to give a estimate on how such a fault source could deviate prediction of true logical error probability using syndrome data.

 We can explicitly characterize the learnable degrees of freedom similar to writing the matrix $D_{\mathscr{G}^\perp}$ in Theorem~\ref{thm: main-logical-learnability}. If we are unable to learn up to logical equivalence, then there would exist $a, b \in \mathcal{K}$ with the same syndrome $\sigma(a) = \sigma(b)$. Taking the logarithm as standard, we can partition the matrix $D_{\mathscr{G}^\perp}$ as follows, 
\begin{align}\label{eq: main-detector-matrix-full}
    D_{\mathscr{G}^\perp} = \left(\begin{array}{c|c|c} 
  A_{\mathscr{M}} & B_{\mathscr{M}} & C_{\mathscr{M}} \\ 
  \hline 
  A_{\mathscr{G}^\perp / \mathscr{M}} & B_{\mathscr{G}^\perp / \mathscr{M}}  & C_{\mathscr{G}^\perp / \mathscr{M}}
\end{array}\right).
\end{align}
Let $A$ denote these with unique syndromes and $B$ denote these with duplicate syndromes but related to gauge group $\mathscr{G}$. Let $C$ denote the duplicate syndromes that are not learnable (not related by the gauge group). In this case, one can always find a basis such that, 
\begin{align}\label{eq: main-detector-matrix-basis-change}
    D_{\mathscr{G}^\perp} = \left(\begin{array}{c|c|c} 
  A_{\mathscr{M}} & 0 & 0\\ 
  \hline 
  A_{\mathscr{G}^\perp / \mathscr{M}} & 0  & C'_{\mathscr{G}^\perp / \mathscr{M}}
\end{array}\right), 
\end{align}
where each (nonzero) column within the partition $C$ undergoes a column reduction with the column of $A$ sharing the same syndrome. By definition of $C$, $C'_{\mathscr{G}^\perp / \mathscr{M}}$ cannot be all zero. 

\begin{remark}
    An austere reader might notice that the partition of columns $A, B, C$ are not unique. Hence, suppose that $a, b, c \in \mathcal{K}$ such that $a \sim_{\mathscr{G}} b $ and $a \sim_{\mathscr{L}_{\operatorname{dressed}}} c$. If we partition $a, b,c $ to the columns of $A, B, C$ respectively, then we have precisely $1$ unlearnable degree of freedom. However, if we identify the column of $A$ with $c$, then we have two unlearnable degrees of freedom, which could be counter-intuitive. To explain this, one may best understand that these degrees of freedom relative to the choice of $\mathcal{K}_{\mathscr{M}}$ as some correction terms to the logical error probability. That is, different way of partitioning would result in different estimations of logical error probabilities. We examine this aspect in more detail in Section~\ref{main-subsec: lep-sampling-complexity}.
\end{remark}

In what follows, we consider how to proceed with bounding the error resulting from cases where there exists an error that is not learnable up to logical equivalence. Suppose that the column $C$ is given by the rare events with probability $\varepsilon_c$ to occur. For $b \in \mathscr{G}^\perp$ and $\lambda_b = \lambda^{\operatorname{eff}}_b$
\begin{align}
   \begin{aligned}
        \log \lambda_b &= \sum_{j} e_j \log \lambda_{s_j} + \sum_{c \in C} C'_{\mathscr{G}^\perp / \mathscr{M}}[b, c]\log(1-2q_c), \\
    &= \sum_{j} e_j \log \lambda_{s_j} + O(|C'_{\mathscr{G}^\perp / \mathscr{M}}[b, :]| \varepsilon_c),
   \end{aligned}
\end{align}
where $e_j$ is the coefficient obtained according to the procedure outlined in the proof of Theorem~\ref{thm: main-logical-learnability} and see details in Section~\ref{subsec: error-propagation-compressed-sensing}. Roughly speaking, the first term indicates the part of information that can be learned from the syndrome data and the second term indicates the additional unlearnable (unknown) corrections from syndrome data.  In this case $\lambda_b$ ($\log \lambda_b$) is approximated within $O(\varepsilon_c)$ relative (and additive) precision, as indicated from the second line (note that $\lambda_b$ is sufficiently close to $1$). Relating to the effective probability for $a \notin \mathscr{G}$, we can perturbatively expand (say that if $\Lambda$ is given as the Pauli-Lindblad form) for the first-order approximation~\cite{Chen_2023_learnability}

\begin{align}
   \begin{aligned}
        &p^{\operatorname{eff}}_a = \frac{1}{|\mathcal{A}|}\sum_{b \in \mathscr{G}^\perp} \chi_b(a) \lambda_b \approx \frac{1}{|\mathcal{A}|} \sum_{b \in \mathscr{G}^\perp} \chi_b(a) \left( 1 + \log \lambda_b \right), \\
        &\approx    \frac{1}{|\mathcal{A}|} \sum_{b \in \mathscr{G}^\perp} \chi_b(a) ( \sum_j e^{(b)}_j \log \lambda_{s_j} + O(|C'_{\mathscr{G}^\perp / \mathscr{M}}[b, :]| \varepsilon_c) ). 
   \end{aligned}
\end{align}
The first equation is given by performing Fourier (Walsh-Hadamard) transform over the effective probability which averages over $\mathscr{G}$ from Eq.~\eqref{eq: main-eff-distribution}. This leads to an additive error $O(\varepsilon_c)$.

\section{Sampling complexity bounds in estimating logical error probability}\label{main-sec: sampling-complexity}

We discuss how to rigorously bound the effect of sampling imprecision and the sample complexity questions in our estimation frameworks. This would take two parts: $(i)$ we derive an optimal sample complexity bound for estimating the coefficients of the prior distribution. $(ii)$ The control of error propagation of the learning algorithm through compressed sensing. Finally, we comment on how to estimate the logical error probability by combining all ingredients.

\subsection{Sample complexity in syndrome extraction}\label{main-subsec: background-free-sampling}

Our goal is to provide a simple observation on the sample complexity required for estimating the Pauli eigenvalues associated with stabilizers from $\mathscr{M}$. Let $\Lambda$ be a fault distribution $\Lambda \in \mathbb{R}[\mathcal{A}]$ and some $m \in \mathscr{M}$. If $\varepsilon$ (physical error probability) is sufficiently small, then we typically require to estimate the $\bar{\lambda}_{m}$ to $O(\varepsilon)$ additive precision. Let us define a binary random variable $x$ associated with $m$ if a fault occurs that anti-commutes with $m$. Then $x$ is a Bernoulli random variable drawn from $\operatorname{Bern}(\sum_{a\in \mathcal{A}: \langle a, m\rangle = 1} p_a) \approx \operatorname{Bern}(\varepsilon)$. Note that we assume that $\sum_{a\in \mathcal{A}: \langle a, m\rangle = 1} p_a \leq \varepsilon$ and in this case we simply denote $x \sim \operatorname{Bern}(\varepsilon)$. 
\begin{align}
    \mathbb{E}_{a \sim \Lambda}[(-1)^{\langle m, a \rangle}] = \mathbb{E}_{x \sim \operatorname{Bern}}[(-1)^x] = 1 - 2 \mathbb{E}_{x \sim \operatorname{Bern}(\varepsilon)}[x].
\end{align}
 Utilizing the Chernoff bound and Le Cam two-point method~\cite{gong2024samplecomplexitypurityinner, lecam1972limits, lecam1986amsdt}, one can show a similar result to~\cite{lee2025efficientbenchmarkinglogicalmagic}, it is known that with sample complexity $S = \Theta(\varepsilon^{-1} \tau^{-2})$, we can construct an estimator $\bar{x} = (1/S)\sum^S_{i=1} x^{(i)}$ such that $\bar{x} \approx_{\tau, 0} x$. This optimal scaling is special to the fact that the variance of the $\operatorname{Bern}(\varepsilon)$ is given by $\varepsilon(1-\varepsilon) = O(\varepsilon)$, which is an instance of a variance-aware method commonly used in computer science and machine learning literature. Converting to the estimated syndrome expectations,

\begin{theorem}\label{thm: main-1/epsilon-scaling}
     Denote $\bar{\varepsilon} = \frac{1}{S} \sum^S_{i=1} x^{(i)}$, $\bar{\mu} = 1 - 2 \bar{\varepsilon}$, and $\mu = 1- 2 \varepsilon$, for a sufficiently small $\varepsilon > 0$ and any $1 > \tau > 0$, with $S = O(\varepsilon^{-1} \tau^{-2} \log(1/\delta))$. Then the following holds, with probability at least $1 - \delta$, for any $\delta > 0$
    \begin{enumerate}
        \item (Additive precision) we have $\bar{\mu} \approx_{ 0, \tau \varepsilon} \mu$. 
        \item (Multiplicative precision) We have $\bar{\mu} \approx_{\tau \varepsilon / (1- 2\varepsilon), 0} \mu$.
    \end{enumerate}
    
\end{theorem}

\subsection{Error propagation from estimation via compressed sensing}\label{main-subsec: error-propagation-rip}

Since in practice we could only estimate the syndrome data up to some prescribed precision, we need to bound how the error could propagate in the learning process. 
From the preceding section, we assume $\bar{\lambda}_a, a \in \mathscr{M}$ is estimated up to relative precision $\tau \varepsilon/ (1-2\varepsilon)$ for some small parameter $\varepsilon > 0$ and $1 > \tau > 0$.  Let $Q \subset  \mathscr{M}$ from which we sample according to Theorem~\ref{thm: rip-bos-rowsampled-refined} (Informally stated in Theorem~\ref{thm: main-row-subsampled-rip-A}), and we denote the vector $y_a := - \log \lambda_a$, $y \in \mathbb{R}^{q}$, so that the estimated $|\bar{y}_a - y_a | = |\log(\bar{\lambda}_a/ \lambda_a)| = |\log (1 \pm  \tau \varepsilon /(1 - 2\varepsilon))| =  \tau \varepsilon + O(\varepsilon^2)$, where we use the multiplicative precision statement from theorem~\ref{thm: main-1/epsilon-scaling} (See also in full details around Theorem~\ref{thm: 1/epsilon-scaling-upper-bound} in Section~\ref{subsec: error-propagation-compressed-sensing}). To the first order, we simply denote $\bar{y}_a \approx_{0, \tau \varepsilon} y_a$ for $a \in Q \subset \mathscr{M}$. Denote the noise vector $\omega$ such that $\bar{y} = y + \omega$, we have $\omega \in \mathbb{R}^q$ with $\|\omega \|_1 \leq q \tau  \varepsilon$ and $\|\omega\|_2 \leq \sqrt{q} \tau \varepsilon$.  Since $\tau$ is assumed to be a small constant, we drop it for the following analysis. Our goal is to examine how the sampling imprecision could propagate into learning the prior distributions Eq.~\eqref{eq: main-prior-distribution}, while the full detailed proof is presented in Section~\ref{subsec: error-propagation-compressed-sensing}.

Recall that $A_{\mathscr{M}, \operatorname{res}}[a, b] = \langle a, b \rangle $ for $a \in Q \subset \mathscr{M}$ and $b \in \mathcal{K}_{\mathscr{M}} \setminus \{ I \}$. We denote $K+1$th-order restricted isometric constant to be $\delta_K$, which is assumed to be small and a constant. 

\begin{lemma}\label{lemma: main-noisy-recovery-all-distribution}
    Consider the following recovery problem, 
    \begin{align}\label{eq: noisy-estimation-problem}
        y = A_{\mathscr{M}, \operatorname{res}} x + \omega, 
    \end{align}
    subject to the noise $\omega$, where we have that $x_a = - \log (1-2q_{[a]}), \: a \in \mathcal{K}_{\mathscr{M}} $. Then there exists a unique choice of noisy recovery $\bar{x}$ such that 
    \begin{align}
        \|\bar{x} - x \|_2 \leq  2  \frac{\sqrt{1 + \delta_K}}{1 - \delta_K}  \varepsilon.
    \end{align}
\end{lemma}

At this stage, we could also discuss how to directly infer the (probability) coefficients of the prior distribution in Definition~\ref{def: prior-distribution} and see Eq.~\eqref{eq: main-prior-distribution}. 

\begin{lemma}\label{lemma: main-noisy-recovery-prior}
     Under the same setting as above, there exists a unique noisy recovery such that for every $a \in \mathcal{K}_{\mathscr{M}}$, the coefficient of the prior distribution $\Lambda^{\operatorname{prior}}$ 
    \begin{align}
        \bar{q}_{[a]} \approx_{0,O(\varepsilon) } q_{[a]},
    \end{align} 
    where the term in the big-$O$ notation is given by $\frac{\sqrt{1 + \delta_K}}{1-\delta_K} \varepsilon$. 
    \begin{proof}
        From Lemma~\ref{lemma: main-noisy-recovery-all-distribution}, note that for $a \in \mathcal{K}$, we have that 
        \begin{align}
            q_{[a]} = \frac{1}{2}(e^{-x_a} - 1),
        \end{align}
        so that we have that 
        \begin{align}
           \begin{aligned}
                |\bar{q}_{[a]} - q_{[a]}| &= \frac{1}{2} | e^{-x_a} - e^{-\bar{x}_a} | \leq \frac{e^{-\min(x_a, \bar{x}_a)}}{2} |x_a - \bar{x}_a|, \\
                &\leq \frac{1}{2} \|\bar{x} - x \|_2 \leq \frac{\sqrt{1 + \delta_K}}{1-\delta_K} \varepsilon.
           \end{aligned}
        \end{align}
        The first inequality is given by the mean-value theorem and the rest follows as standard. 
    \end{proof}
\end{lemma}

 This gives rise to an additive error for estimating the coefficients from the prior distribution. If one wishes to arrive at relative precision for every coefficient, an additional assumption is needed: coefficients from prior distribution $\Lambda^{\operatorname{prior}}$ are sufficiently uniform. Combining with the sample complexity in section~\ref{main-subsec: background-free-sampling}, we conclude with an optimal sample bound for estimating coefficients of the prior distribution to $O(\varepsilon)$ precision, which is summarized in Theorem~\ref{thm: background-free-estimation-prior} and Theorem~\ref{thm: sample-complexity-prior-distribution}. Using similar ideas with the compressed-sensing property, we show in Section~\ref{subsec: error-propagation-compressed-sensing} the error propagation from the effective distribution can be bounded, see Theorem~\ref{thm: error-propagation-logical}.

\subsection{End-to-end estimation framework for logical error probability}\label{main-subsec: lep-sampling-complexity}

We now present the end-to-end estimation protocol for estimating the logical error probability by explaining and analyzing the Algorithm~\ref{algo: end-2-end-framework-from-syndrome}. So far, our emphasis has been limited to discussing the prior distribution $\Lambda^{\operatorname{prior}}$ and the effective distribution $\Lambda^{\operatorname{eff}}$. We now discuss how estimating this would lead to estimating the logical error probability. The detailed discussion is presented in Section~\ref{subsec: spacetime-logical-noise-channel} and Section~\ref{subsec: sampling-LEP}. Under the circuit-to-code mapping presented in Section~\ref{main-sec: circuit-learnability}, we can associate the (bare) logical operators, including the stabilizers, from $\mathscr{G}^\perp$. Under this setting, the logical error probability can always be derived from the effective distribution, specified by the choice of decoder correction. For the simplicity of discussion, we assume that we utilize all $\mathscr{M}$ for syndrome measurements. Note that in practice, one might use a subgroup $\mathscr{M}' \subset \mathscr{M}$, such as the detector error model. This would not affect our discussion in what follows (The learnability condition would change accordingly). A fixed decoder records the corrections to apply conditions on the syndromes as a dictionary $\{z \in \mathbb{F}^M_2 = \sigma(\mathscr{M}): a_{z}\}$, with $2^M = |\mathscr{M}|$ and $a_z \in \mathcal{A}$ such that $\sigma(a_z) = z$. Let $\overleftarrow{l} \in \mathscr{G}^\perp / \mathscr{M}$ be a (bare) logical operator (here we note that the bare logical operators are generated by the backward propagations from the terminal time for the circuit-to-code isomorphism in Theorem~\ref{thm: subsystem-circuit-to-code-isomorphism}),

\begin{align}\label{eq: main-logical-error-probability}
    \begin{aligned}
       p^{L}_{\bar{l}}  &= |\mathscr{G}| \sum_{z \in \mathbb{F}^M_2} p^{\operatorname{eff}}_{a_z \overleftarrow{l}}, \\
       & = \frac{1}{|\mathscr{G}^\perp|} \sum_{z \in \mathbb{F}^M_2}  \sum_{b \in \mathscr{G}^\perp} \chi_{b} \lambda_{b}(a_z \overleftarrow{l}), \\
       &= \frac{1}{|\mathscr{G}^\perp|} \sum_{z \in \mathbb{F}^M_2} \sum_{\overleftarrow{l'} \in \mathscr{G}^{\perp} / \mathscr{M}}\sum_{m \in \mathscr{M}} \lambda_{\overleftarrow{l'} m} \chi_{\overleftarrow{l}' m}(a_z \overleftarrow{l}), \\
       &= \frac{1}{|\mathscr{G}^\perp|} \sum_{z \in \mathbb{F}^M_2} \sum_{\overleftarrow{l'} \in \mathscr{G}^{\perp} / \mathscr{M}}\chi_{\overleftarrow{l}'}(\overleftarrow{l})  \sum_{m \in \mathscr{M}} \lambda_{\overleftarrow{l'} m} \chi_{\overleftarrow{l}' m}(a_z).
    \end{aligned}
\end{align}
This sum is combinatorial and intractable for computation. In many cases, a simplified way exists for evaluating the logical probability, utilizing the prior distribution. 

\begin{corollary}\label{cor: main-prior-distribution-learn-logical-error}
    The prior distribution in Eq.~\eqref{eq: main-prior-distribution} is capable of representing the logical error probabilities if and only if $C'_{\mathscr{G}^{\perp}/\mathscr{M}}=0$ in Eq.~\eqref{eq: main-detector-matrix-basis-change}. In other words, for any $a, a' \in \mathcal{K}$, $\sigma(a) = \sigma(a') \iff a \sim_{\mathscr{G}} a'$. 
    \begin{proof}
        The forward direction is straightforward. If $C'_{\mathscr{G}^{\perp}/\mathscr{M}}=0$, then we can represent $\lambda_b$ for any $b \in \mathscr{G}^\perp$ as a product of Pauli eigenvalues. The basis transformation leaving the columns in $B, C$ with zero entry in Eq.~\eqref{eq: main-detector-matrix-basis-change} precisely means that we have associated $1 -2 q_{[c]} = \prod_{a \in \mathcal{K}: \sigma(a) = \sigma(c)} (1-2q_a)$ for the coefficient $q_a$ under the parameterization from Eq.~\eqref{eq: main-character-fourier-basis}, which is learnable from Theorem~\ref{thm: main-physical-learnability}. To show that converse, note that if $C'_{\mathscr{G}^{\perp}/\mathscr{M}}\neq 0$ (recall notation from Eq.~\eqref{eq: main-detector-matrix-basis-change}), then there always exists some $b \in \mathscr{G}^\perp \setminus  \mathscr{M}$ such that 
        \begin{align}
    \lambda_b = ( \prod_{a \in \mathcal{K}_{\mathscr{M}}: \langle a, b \rangle = 1} (1-2q_{[a]})  ) (\prod_{c \in C} (1-2q_c)^{C'_{\mathscr{G}^\perp/ \mathscr{M}}[b,c]}).
\end{align}
The expression within the first parenthesis gives the product of coefficients of the prior distribution, which by definition (Eq.~\eqref{eq: main-prior-distribution}) equals $\prod_{a \in \mathcal{K}_{\mathscr{M}}: \langle a, b \rangle = 1} (1-2q_{[a]}) = \prod_{a \in \mathcal{K}_{\mathscr{M}}: \langle a, b \rangle = 1} \prod_{a' \in \mathcal{K}: \sigma(a') = \sigma(a)} (1-2q_{a'})$. The terms within the second parenthesis are independent from coefficients of the prior distribution; in other words, varying $q_c$ from $c \in C'_{\mathscr{G}^\perp / \mathscr{M}}[b, :]$ would provide different values of $\lambda_b$ for any given prior distribution, which concludes the result. 
    \end{proof}
\end{corollary}
When the fault distribution $\Lambda$ can be learned up to logical equivalence, we can express the logical error probability in a cleaner way, 
\begin{align}\label{eq: main-prior-distribution-logical-error}
     p^{L}_{\bar{l}} = \sum_{J \subset \mathcal{K}_{\mathscr{M}}} \prod_{c \in J}  q_{[c]} \prod_{c' \notin J}(1-q_{[c']}) \mathbf{1}\{\prod_{c \in J} ca_{\sigma(\prod_{c \in J} c)} \overleftarrow{l} \in \mathscr{G}\},
\end{align}
where $J \subseteq \mathcal{K}$ denotes a subset of independent fault processes $\mathcal{K}$. The full derivation is seen in Eq.~\eqref{eq: logical-error-probability}. The sum is still combinatorial and remains inefficient to compute exactly. In Section~\ref{subsec: sampling-LEP}, we give some analysis on heuristic approach in approximating the logical error probability to a relative precision, which we utilize in the numerical simulations presented in Figure~\ref{Fig: main-simulation}. 

\begin{remark}
    We finally are in the position to explain the remarks left in Section~\ref{main-subsec: main-unlearnable-dof}. Let $\mathcal{K}$ denote all the independent fault processes and we define the prior distribution by giving a partition with $\mathcal{K}_{\mathscr{M}}$. Note that choices of different $\mathcal{K}_{\mathscr{M}}$ would leave the coefficients of $\Lambda^{\operatorname{prior}}$ invariant. However, from Eq.~\eqref{eq: main-prior-distribution-logical-error}, the conditional evaluation $ \mathbf{1}\{\prod_{c \in J} ca_{\sigma(\prod_{c \in J} c)} \overleftarrow{l} \in \mathscr{G}\}$ might be different. Suppose we choose some $c' = c \overleftarrow{l'}$ for some $\overleftarrow{l'} \in \mathscr{L}_{\operatorname{bare}}$, then we might arrive at a different interpretation on when a logical error occurs. Interpretations when a logical error occurs only become consistent if $\sigma(c) = \sigma(c') \iff c \sim_{\mathscr{G}} c'$ for all $c, c' \in \mathcal{K}$. In this case, Eq.~\eqref{eq: main-prior-distribution-logical-error} is defined without ambiguity, which gives another viewpoint from Corollary~\ref{cor: main-prior-distribution-learn-logical-error}. 
\end{remark}

\section*{Discussion}
In this work, we systematically address the prospect of learning (logical) noise from syndrome data in various practical settings and give an end-to-end estimation protocol and optimal, rigorous sample complexity bounds. In this regard, our work can be seen as a systematic characterization of the learnability of Pauli faults from a restricted set of measurement data, built primarily on closely related lines of work~\cite {Wagner_2021, Wagner_2022, Wagner_2023}, while we generalize to the circuit-level fault models. In many current experimental settings with fault-tolerant circuits, we expect the techniques presented in the current work would be directly implementable. 

There are several future directions worth considering. First, it is of great interest to generalize the formalism to universal logical circuits with non-Clifford unitaries. Even in the case of Pauli faults, the propagation of non-Clifford circuits is not an automorphism on the spacetime Pauli group. The same can be extended to the case of coherent error~\cite{Bravyi_2018} or erasure error. Hence, it is of interest to study how to capture this type of fault from syndrome information. A sensible yet high-level intuition can be given as follows. Instead of asking how to establish an isometry between syndrome data and concerned unknown error parameterizations within $\mathcal{A}$ (or the group algebra of $\mathcal{A}$), we ask if such an isometry can be established beyond the characterization of the (spacetime) Boolean group $\mathcal{A}$. Second, one may wish to systematically address the fault-tolerance (dressed distance) of the circuit-to-code isomorphism. This naturally falls into the scope of dynamical codes~\cite{Hastings_2021, xu2025faulttolerantprotocolsspacetimeconcatenation, blackwell2025codedistancefloquetcodes, Fu_2025, FuGottesman:SubsystemSpacetimeCode}. This generalization would be useful in a benchmarking setting for logical components, utilizing measurements that are dynamically generated. Finally, one could consider the effect of full state preparation and measurement (SPAM) noise. This is particularly useful if one wishes to utilize only the mid-circuit measurements and benchmark individual circuit components~\cite{girling2025characterizationsyndromedependentlogicalnoise}, as the learnability condition derived could easily fail. This direction could potentially be more in line with a series of gate set benchmarking protocols~\cite{Nielsen_2021, chen2024efficient, Chen_2023_learnability, Zhang_2025}, yet potentially used for benchmarking at the logical level.

\emph{--Note added.} During the preparation of this manuscript, we became aware of~\cite{xiao2026insitubenchmarkingfaulttolerantquantum} with similar content. We thank insightful exchanges and discussions with the authors of~\cite{xiao2026insitubenchmarkingfaulttolerantquantum} during the Quantum Error Correction (QEC) conference 2025.

\begin{acknowledgments}
    We thank Pei Zeng for helpful discussions and the valuable feedback on the initial manuscript. 
    S.C. acknowledges funding provided by the Institute for Quantum Information and Matter, an NSF Physics Frontiers Center (NSF Grant PHY-2317110).
    S.Z. acknowledges support from National Research Council of Canada (Grant No.~AQC-217-1) and Perimeter Institute for Theoretical Physics, a research institute supported in part by the Government of Canada through the Department of Innovation, Science and Economic Development Canada and by the Province of Ontario through the Ministry of Colleges and Universities. H.Z., C.-T.C., S.C., A.G.M., S.L., and L.J. acknowledge support from the ARO (W911NF-23-1-0077), ARO MURI (W911NF-21-1-0325), AFOSR MURI (FA9550-21-1-0209, FA9550-23-1-0338), DARPA (HR0011-24-9-0359, HR0011-24-9-0361), NSF (ERC-1941583, OMA-2137642, OSI-2326767, CCF-2312755, OSI-2426975), and the Packard Foundation (2020-71479).
    
\end{acknowledgments}

\bibliography{references}

\clearpage
\widetext
\appendix
\begingroup
\section*{Appendix}

\titleformat{\section}[block]{\large\bfseries\filcenter}{}{0pt}{}
\setcounter{tocdepth}{2} 

\startcontents[sections]
\titlecontents{section}[0pt]{\vspace{5mm}}{\thecontentslabel\hspace{1em}}{}{\titlerule*[1pc]{.}\contentspage}
\titlecontents{subsection}[1.5em]{}{\thecontentslabel\hspace{1em}}{}{\titlerule*[1pc]{.}\contentspage}
\printcontents[sections]{}{1}{\section*{}\vspace{-10mm}}

\endgroup

\section{Preliminary}

\newcolumntype{N}{>{\centering\arraybackslash$\displaystyle}p{0.35\textwidth}<{$}}
\newcolumntype{E}{>{\raggedright\arraybackslash}X}

\begin{table}[htbp]
  \centering
  \renewcommand{\arraystretch}{1.5}
  \begin{tabularx}{\textwidth}{N E}
    \toprule
    \multicolumn{1}{c}{\textbf{Notation}} & \textbf{Meaning} \\
    \midrule
    \mathcal{A} & The (Boolean) spacetime group. \\
    \mathbb{C}[\mathcal{A}], \mathbb{R}[\mathcal{A}] & Group algebra of $\mathcal{A}$, either complex-valued or real-valued. \\
    \beta: \mathcal{A} \times \mathcal{A} \rightarrow \mathbb{F}_2 & Bilinear form associated with $\mathcal{A}$, which is assumed to be symmetric and nondegenerate. \\
    \chi_{a} := (-1)^{\beta(a, \cdot)}, \,a \in \mathcal{A} & Characters of the Boolean group defined via the bilinear form $\beta.$ \\
    f \in \mathbb{R}[\mathcal{A}] & A distribution over $\mathbb{R}[\mathcal{A}]$.  \\
    \mathscr{F}: \mathbb{R}[\mathcal{A}] \rightarrow \mathbb{R}[\mathcal{A}] & The Fourier transform over the Boolean group, which is real-valued and a ring isomorphism.  \\
    p_a := f(a), p_a \in [0, 1] & Probability for which $a \in \mathcal{A}$ occurs.   \\
    \lambda_a := \mathscr{F}[f](a) = \widehat{f}(a) & Fourier coefficient (Pauli eigenvalues) associated with $a \in \mathcal{A}$. \\
    \mathscr{G} \subseteq \mathcal{A} &   Gauge group as a subgroup of $\mathcal{A}$. \\
    \mathscr{G}^\perp & The orthogonal complement of $\mathscr{G}$ with respect to $\beta$. \\
    \mathscr{M} := \mathscr{G} \cap \mathscr{G}^\perp & The center of the gauge group, which is given as the measurement subgroup. \\
    M & The number of independent generators of the measurement subgroup $\mathscr{M}$. \\
    \sigma: \mathcal{A} \rightarrow \mathbb{F}^M_2 & The syndrome vector records the commutation rules with $\mathscr{M}$ with respect to $\beta$. \\
    \mathscr{M}^\perp / \mathscr{G}, \, \mathscr{G}^\perp / \mathscr{M} & The dressed and bare logical subspaces modulo the phases. \\
    \mathcal{K} \subseteq \mathcal{A}, & A parametrization subset of $\mathcal{A}$.\\
    \mathcal{K}_{\mathscr{M}} \subseteq \mathcal{K}, \, |\mathcal{K}_{\mathscr{M}}| =K & A subset of $\mathcal{K}$ gives (non-uniquely) the elements of distinct syndromes. \\
    p^{\operatorname{eff}}_a & The effective error probability modulo the gauge equivalence. \\

    \delta_K &  The $K+1$th order restricted isometry constant. \\
     1- \varepsilon & The probability where there is no fault from $\Lambda$. \\
   
    p_a & Pauli error rate (probability that the fault $a$ occurs). \\

    q_a & The parameterization coefficient in Eq.~\eqref{eq: main-character-fourier-basis}. \\

q_{[a]} &  The parameterization coefficient of the prior distribution in Eq.~\eqref{eq: main-prior-distribution}. \\

       \bottomrule
  \end{tabularx}
  \caption{Elementary notations that we use in this manuscript. Other notations used in more technical analysis are derived in consistency with the above notational conventions}
  \label{tab: notation}
\end{table}

We now present a preliminary review of several technical tools that we use to derive the main results. 

\subsection{Characters of an Abelian group and Fourier transforms}

We now generalize our questions in  a more general setting. Let $G$ be a group and denote its \emph{group algebra} by $\mathbb{C}[G] = \{ f: G \rightarrow \mathbb{C}: \|f\|_2 < \infty \}$. The condition $L_2$-boundedness ensures that if $G$ is a continuous (Lie) group, the group algebra is well-defined. For simplicity, we consider the case that $G$ is a finite group and let $\mathbb{C}[G]$ endow a typical multiplication structure, called the convolution
\begin{align}
    (f * g)(a) =  \sum_{b \in G} f(ab^{-1}) g(b).
\end{align}
A classical characterization of the group algebra $\mathbb{C}[G]$ can be given by the following. First we define a formal vector space $V = \{ h \in V: h = \sum_{a} h(a) |a \rangle, a \in G\}$, where this vector space has dimension $|G|$, equipped with the inner product $\langle g | h \rangle = \frac{1}{|G|} \sum_{a \in G} g(a) \overline{h(a)}$. 

\begin{proposition}
    There exists a ring isomorphism between the group algebra $\mathbb{C}[G]$ and the (left)-regular representation on $V$. 
    \begin{proof}
        Recall that for a group $G$, the (left)-regular representation $G$ is given by the map, 
        \begin{align}
            G \rightarrow \operatorname{End}(V); \quad g \mapsto \pi(g),
        \end{align}
        where $\bra{a}\pi(g)\ket{b} = \delta_{ag, b}$. Then we claim that this map induces a ring isomorphism, 
        \begin{align}
            \mathbb{C}[G] \rightarrow \operatorname{End}_{\mathbb{C}[G]}(V); \quad f \mapsto \sum_{a \in G} f(a) \pi(a),
        \end{align}
        where $\operatorname{End}_{\mathbb{C}[G]}(V) = \{ M \in \operatorname{End}_{\mathbb{C}} (V): L_a M = ML_a, \forall a \in G\}$ and $L_a |b \rangle = |a b \rangle$ the left action. Note that the matrix form of $\pi(a)$ is given by the right action $\pi(a) = R_a$. It is easy to see that for any $h \in G$, $L_a R_b | h \rangle = |ahb^{-1}\rangle = R_b L_a |h \rangle$. First we note that this map is well-defined.  
        \begin{align}
            \begin{aligned}
                \pi(f * g) &= \sum_{a \in G} (f * g)(a) \pi(a) =  \sum_{a \in G} \sum_{b \in G} f(a b^{-1}) g(b) \pi(a), \\
                &=  \sum_{a \in G} \sum_{b \in G} f(a) g(b) \pi(ab) = \pi(f) \pi(g).
            \end{aligned}
        \end{align}
        To show the injectivity, we first note that for all coefficients $c_a$, 
        \begin{align}
            \sum_{a \in G} c_a \pi(a)  = 0 \iff \sum_{a \in G} c_a \pi(a) |1 \rangle = \sum_{a} c_a |a \rangle = 0 \iff   c_a =0; \forall a \in G.
        \end{align}
        This follows from the fact that we identify each $|a \rangle$ as orthogonal basis elements in the regular representation. The surjectivity follows from definition of $\operatorname{End}_{\mathbb{C}[G]}(V)$, where we show that it is spanned by the right regular actions. Let $M \in \operatorname{End}_{\mathbb{C}[G]}(V)$, we have that $M \ket{a} = M L_a |1 \rangle = L_a M |1\rangle$. Let's denote $M|1 \rangle = \sum_{h \in G} t_h |h\rangle $ for some coefficients $t_h$, so that $M|a\rangle = \sum_{h \in G} t_h |ah \rangle = \sum_{h \in G} t_h R_{h^{-1}} |a \rangle$ as desired. 
    \end{proof}
\end{proposition}

In this regard, one typically writes $f \in \mathbb{C}[G]$ using the matrix form given by the identification of regular representation. Namely, $f = \sum_{a}  f(a) \pi(g)$. It is well-known that for any compact group, we can decompose its regular representation to the irreducible representations. This duality transformation is defined as the Fourier transform. In what follows, we consider $G =\mathcal{A}$ as an Abelian group. In this case, there exists a close relation between $\mathcal{A}$ and its dual $\widehat{\mathcal{A}}$.

\begin{definition}[Dual of group $A$]
     The dual group of any Abelian group $A$, denoted as $\widehat{\mathcal{A}}$, consists of all homomorphisms 
    $\chi: \mathcal{A} \rightarrow \mathbb{T}^1$, where the group operation is given by multiplication. These homomorphisms are called the \emph{linear characters} of the group $A$. 
\end{definition}

Let's further restrict our attention to the case of finite Abelian groups, where we state the following result, while the proofs are standard (See, e.g., \cite{TaoVu2006AddComb, goodman_wallach_2009}).

\begin{proposition}[Orthogonal properties of linear characters]
     For any finite Abelian group $\mathcal{A}$ that has group order $n$. Denote the $\mu_n$ as the group of $n$th unity. 
    \begin{enumerate}
        \item $\chi(a) \in \mu_n$ for all $a \in \mathcal{A}$; hence, $|\chi(a)|=1$.
        \item For any $\theta, \chi \in \widehat{\mathcal{A}}$, $\langle \chi, \theta \rangle = \delta_{\chi, \theta}$.
    \end{enumerate}
\end{proposition}

Hence, the characters form an orthonormal basis and the next result indicates that they give an orthonormal basis to the group algebra $\mathbb{C}[\mathcal{A}]$. 

\begin{theorem}
    Let $\mathcal{A}$ be a finite Abelian group, then 
    \begin{enumerate}
        \item $\mathcal{A} \cong \widehat{A}$
        \item (Pontryagin duality) Define for each $a \in \mathcal{A}$, the evaluation character, $\widehat{a}: \widehat{\mathcal{A}} \rightarrow \mathbb{T}$, $\widehat{a}(\chi) = \chi(a)$. Then the induced group homomorphism, 
    \begin{align}
        \mathcal{A} \rightarrow \widehat{\widehat{\mathcal{A}}}; \quad a \mapsto \widehat{a}
    \end{align}
    is an isomorphism.
    \end{enumerate}
\end{theorem}

It is noted that the isomorphism $A \cong \widehat{\mathcal{A}}$ need not to be canonical. However, writing out this isomorphism is critical for deriving the Fourier transforms and later results. To motivate the study, we provide two examples in understanding the isomorphism.

\begin{example}[Cyclic group $C_n$]
 Let $\zeta$ be the $n$th root of unity in $\mathbb{C}$ and $\gamma$ be the generator. For the $\chi: C_n \rightarrow \mathbb{T}^1$, the homomorphism property constraints that $\chi(\gamma^n)= \chi(1) =  \chi(\gamma)^n $ so that there are at most $n$ linearly independent characters. We can write out these $n$ characters as $\chi_k(\gamma^j) = \zeta^{jk}$. To see that they are indeed linearly independent (orthonormal basis), we define the following \emph{bicharacters} 
\begin{align}
    e: C_n \times C_n \rightarrow \mathbb{Q}/\mathbb{Z}.
\end{align}
Given by $\beta(\gamma^j, \gamma^k) = \frac{jk}{n}$ and the identification $\chi_k = e^{i2 \pi e(\cdot, \gamma^k)}$. Then we can compute that 
\begin{align}
    \langle \chi_{l} | \chi_m \rangle = \frac{1}{n} \sum^{n-1}_{j = 0} \chi_l(\gamma^j) \overline{\chi}_m (\gamma^j) = \frac{1}{n} \sum^{n-1}_{j=0} e^{i2\pi e(\gamma^j, \gamma^{l} \gamma^{-m})}.
\end{align}
Then the orthonormality follows from the definition of $\beta$. 
\end{example}

\begin{example}[Direct product of cyclic group]
     Let $A$ be a direct product of some cyclic group. 
    \begin{align}
        A = C_{n_1} \times \cdots \times C_{n_r}; \quad \widehat{A} = \widehat{C}_{n_1} \times \cdots \times \widehat{C}_{n_r},
    \end{align}
    where loosely speaking the decomposition of the dual group $\widehat{A}$ is given by the isomorphism, which we now explain. Denote the group element $(\gamma^{j_1}, \cdots, \gamma^{j_r})$, then the group character is simply given by $\chi((\gamma^{j_1}, \cdots, \gamma^{j_r})) = \chi(\gamma^{j_1}) \cdots \chi(\gamma^{j_r})$, where $\chi(\gamma^{j_1}) = \chi((\gamma^{j_1}, 1, \cdots, 1))$.  
\end{example}

The above examples motivate the definition of bilinear form. 

\begin{definition}[Bilinear form]
     A \emph{bilinear form} or a \emph{bicharacter} on a finite Abelian group $\mathcal{A}$ is a map, 
    \begin{align}
        \beta: \mathcal{A} \times \mathcal{A} \rightarrow \mathbb{Q} / \mathbb{Z}; \quad (\xi, a) \mapsto \beta(\xi, a), 
    \end{align}
    which is a group homomorphism in each of the variable $\xi, a$. We say the bicharacter is \emph{non-degenerate} if for every non-zero $\xi$, the map $a \mapsto \beta(\xi, a)$ is not identically zero and similarly for every non-zero $a$ the map $\xi\mapsto \beta(\xi, a)$ is not identically zero. We say the form is symmetric if $\beta(\xi, a) = \beta(a, \xi)$. 
\end{definition}

\begin{lemma}[Existence of bilinear forms \cite{TaoVu2006AddComb}]\label{lemma: exist-nondegenerate-symmetric-bilinear-forms}
     Every finite Abelian group $A$ has at least one non-degenerate symmetric bilinear form which induces the isomorphism 
    \begin{align}
        \phi: \mathcal{A} \rightarrow \widehat{\mathcal{A}}; \quad \phi(a) \mapsto \exp(i2\pi \beta(\cdot, a)) \equiv \chi_a.
    \end{align}
\end{lemma}

\begin{definition}[Elementary-$p$ Abelian group]
     An \emph{elementary-$p$ Abelian group $\mathcal{A}$} is a finite Abelian group such that all its non-identity elements have the common order $p$. In the case of $p=2$, we also refer this to a \emph{Boolean group}. 
\end{definition}

It is known that any Boolean group $\mathcal{A}$ is isomorphic to some $\mathbb{F}^n_2$ for some integer $N$. In what follows we restrict our attention to the Boolean group.

\begin{example}
    A good example for such a non-degenerate bicharacter is given by the finite Abelian group $\mathcal{A}= \mathbb{F}^n_2$. Then the bicharacter $\beta(\xi, a) = \langle \xi, a \rangle $ over $\mathbb{F}_2$. 
\end{example}

\begin{example}[$\mathbb{F}^{2n}_2$ and the Pauli group $\mathrm{P}^n$]
     Let's define $\beta = \langle \cdot, \cdot \rangle$ to be the standard inner product over $\mathbb{F}^{2n}_2$. Write $a = (a_x, a_z)$ and consider the automorphism $f \in \operatorname{Aut}(\mathcal{A})$ such that $f((a_x, a_z)) = (a_z, a_x)$, then we can define the identification $\Phi_{\beta'} = \Phi_{\beta} \circ f$ so that $(a_x, a_z) \mapsto (-1)^{\langle \cdot, (a_z, a_x)\rangle}$. The identification $\Phi_{\beta'}$ gives the Walsh-Hadamard transform over the Pauli group. 
\end{example}



\begin{definition}[Fourier transform]
     Let $f \in\mathbb{C}[\mathcal{A}]$ for some finite Abelian group $\mathcal{A}$. Then its \emph{Fourier transform} $\widehat{f}$ is given by the formula, 
    \begin{align}
        \widehat{f}(\xi) := \sum_{a \in \mathcal{A}} f(a) \overline{\chi}_\xi(a) = \sum_{a \in \mathcal{A}} f(a) e^{-i2\pi \beta(a, \xi)}.
    \end{align}
    We refer to $\widehat{f}(\xi)$ as the \emph{Fourier coefficient} of $f$ at $\xi \in \mathcal{A}$. 
\end{definition}

There are certain properties of Fourier transforms which we now state without proof (See \cite{TaoVu2006AddComb} and in various standard mathematical textbooks). 

\begin{lemma}
    Let $f, g \in \mathbb{C}[\mathcal{A}]$, then   
    \begin{enumerate}
        \item (Parseval's identity) $\langle f, g \rangle = \sum_{a} \widehat{f}(a) \widehat{g}(a)$.
        \item $\hat{f}(0) = \frac{1}{2^n}\sum_{x}f(x)$.
        \item (convolution theorem) $\widehat{f * g} = \widehat{f} \: \widehat{g}$.
        \item Let $\mathcal{C} \subseteq \mathcal{A}$, a subgroup. Then $\widehat{1_\mathcal{C}} = \frac{|\mathcal{A}|}{|\mathcal{C}|} 1_{\mathcal{C}^\perp}$, where $\mathcal{C}^\perp$ is the subgroup of orthogonal complement of $\mathcal{C}$ with respect to $\beta$. 
    \end{enumerate}
\end{lemma}

\subsection{Review on basic probability theory}

First, we give definitions regarding the multiplicative and additive precisions. 

\begin{definition}
    For $x, y  \in \mathbb{C}$, we say that $x \approx_{\varepsilon, \eta} y $ if $x \in [ (1-\varepsilon)y - \eta, (1+\varepsilon)y + \eta]$ for two real parameters $\varepsilon, \eta > 0$. 
\end{definition}

\begin{lemma}[Compositions of Approximations]\label{lemma: approximation-compositions-multiplication}
     Let $x, y, z \in \mathbb{C}^N$ and $x \approx_{\epsilon_1, \eta_1} y $ and $y \approx_{\epsilon_2, \eta_2} z$ for parameters $\epsilon_1, \epsilon_2, \eta_1, \eta_2 > 0$, then we have that 
    \begin{align}
        x \approx_{\varepsilon', \eta'} z,
    \end{align}
    where we have $\varepsilon' = \epsilon_1 + \epsilon_2 + \epsilon_1 \epsilon_2$ and $\eta' = \eta_1 + (1 + \epsilon_1)\eta_2$. Furthermore, $x \approx_{\epsilon_1, \eta_1} y$ then $cx \approx_{\varepsilon, c\eta } y$ for some constant $c$. 
    \begin{proof}
        We can prove this via triangle inequality. Note that we have 
        \begin{align}
            |x - z| = |x -y | + |y -z| \leq \epsilon_1|y| + \eta_1 + \epsilon_2|z| + \eta_2.
        \end{align}
        Apply the triangle inequality again where we have $|y| \leq |z| + |y-z| \leq (1+ \epsilon_2)|z| + \eta_2$. Combining this we finally arrive at 
        \begin{align}
            |x-z| \leq \epsilon_1(1+\epsilon_2)|z| + \epsilon_1 \eta_2 + \eta_1 + \epsilon_2|z| + \eta_2 \implies x \approx_{\varepsilon', \eta'} z.
        \end{align}
        as desired. The second point is clear from observation. 
    \end{proof}
\end{lemma}

We start from the most basic probability inequality, known as the \emph{Markov inequality}. Let $X$ be a nonnegative random variable with $\mathbb{E}[X] < \infty$. For every $a>0$, 
\begin{align}
    \operatorname{Pr}[X \geq a] \leq \frac{\mathbb{E}[X]}{a}.
\end{align}

Let $X$ be a random variable. Its \emph{moment generating function} is given by  $M_X(t) = \mathbb{E}[e^{tX}]$ with a parameter $t > 0$. The moment generating function completely characterizes an underlying distribution. In some simple cases, we can arrive at an analytical form of the moment generating function. An important example is the Bernoulli distribution $\operatorname{Bern}(p)$, where we have 
\begin{align}\label{eq: moment-generating-bernoulli}
    \mathbb{E}_{X \sim \operatorname{Bern}(p)}[e^{tX}] = p e^t + (1-p) = 1 + p(e^{t} -1) \leq e^{p(e^t-1)}. 
\end{align}
Using Eq.~\eqref{eq: moment-generating-bernoulli} and Markov inequality, we give the standard proof of the multiplicative Chernoff bound. 

\begin{proposition}[Chernoff bounds]\label{prop: chernoff-bound}
      Let $X_1, \ldots, X_S$ be $S$ identically distributed independent random variables in $\{ 0,1 \}$ satisfying $\mathbb{E}\left[X_i\right]=\mu$ for all $i$, and denote $\bar{\mu}=\frac{1}{S} \cdot \sum_{i=1}^S X_i$. Then for any $\tau > 0$, we have that 
     \begin{align}
         \operatorname{Pr}[\bar{\mu} \geq (1 + \tau) \mu] \leq \left( \frac{e^\tau}{(1+\tau)^{1+\tau}}\right)^{\mu S}, \\
         \operatorname{Pr}[\bar{\mu} \leq (1 - \tau) \mu] \leq \left( \frac{e^{-\tau}}{(1-\tau)^{1-\tau}}\right)^{\mu S}.
     \end{align}
     Moreover, for $\tau \in (0, 1)$, we have that $\bar{\mu} \approx_{\tau, 0} \mu$ with probability at least $1-e^{-\tau^2\mu S/3}$. 
     \begin{proof}
         Markov inequality applied to the moment generating function of the Bernoulli distribution. 
     \end{proof}
\end{proposition}

It can be generalized further to any bounded random variables from an arbitrary distribution. To state the result, we first state the following lemma on the moment generating function for a bounded distribution. 

\begin{lemma}[Hoeffding's lemma]
     Let $X$ is a real-valued and bounded random variable in $[a, b]$. Then we have the moment generating for $t > 0$ and $\mathbb{E}[X] = \mu$, 
    \begin{align}
       M_X(t) = \mathbb{E}[e^{tX}] \leq \exp(\frac{t^2(b-a)^2}{8} + t\mu).
    \end{align}
    \begin{proof}
        The main observation is the fact that $e^{tX}$ is convex (actually it is convex for any $t \in \mathbb{R}$). Hence, we have that 
        \begin{align}
            e^{tX} \leq \frac{b -X}{b-a} e^{ta} + \frac{X -a}{b-a} e^{tb}.
        \end{align}
        Taking the expectations, we arrive at, 
        \begin{align}
           \begin{aligned}
                M_X(t) &= \mathbb{E}[e^{tX}] \leq \frac{b-\mu}{b-a}e^{ta} + \frac{\mu-a}{b-a}e^{tb} .
           \end{aligned}
        \end{align}
        Let's denote that $\delta = \frac{\mu-a}{b-a}$ and $1 - \delta = \frac{b-\mu}{b-a}$. Then we arrive at, 
        \begin{align}
            e^{-t\mu} \mathbb{E}[e^{tX}] \leq (1-\delta) e^{-\delta t (b-a)} + \delta e^{(1-\delta) t(b-a)}.
        \end{align}
        Note that we have that $\delta \in (0, 1)$ and set $u := t (b-a)$, we now give a bound on $g(u) = (1-\delta) e^{-\delta u} + \delta e^{(1-\delta)u}$ and $g(u) > 0$. Denote $h(u) = \log g(u)$ and we wish to bound $h(u)$, which in turn bounds $g(u)$ by the monotonicity of logarithm. Note that $h(0) = 0$ and $h'(0) = g'(0)/g(0) = g'(0) = \delta(1-\delta) (e^{(1-\delta)u} - e^{-\delta u})|_{u=0} = 0$. We now compute its second derivative, $h''(u) = (-g'^2(u) + g(u) g^{''}(u)) / g^2(u)$. We compute, 
        \begin{align}
            \begin{aligned}
                g(u)g^{''}(u) - g^{'2}(u)  = \delta(1-\delta)e^{(1-2\delta)\mu}.
            \end{aligned}
        \end{align}
        This implies that we have 
        \begin{align}
            \begin{aligned}
                h''(u) = \frac{\delta(1-\delta)e^{(1-\delta)\mu} e^{-\delta \mu}}{\left( \delta e^{(1-\delta)\mu} + (1-\delta)e^{-\delta \mu} \right)^2} \leq \frac{1}{4},
            \end{aligned}
        \end{align}
        where the inequality is given by the AM-GM inequality. By the Taylor's theorem and combine with the zeroth and first-order derivatives, this implies that for all $u \in \mathbb{R}$, 
        \begin{align}
            h(u) \leq \frac{u^2}{8} \implies g(u) \leq \exp(\frac{u^2}{8}) = \exp(\frac{t^2(b-a)^2}{8}).
        \end{align}
        Plugging back the $e^{t\mu}$, we conclude as desired. 
    \end{proof}
\end{lemma}

We state a variant of the Chernoff-Hoeffding bound that we will use later. 

\begin{proposition}[Chernoff-Hoeffding bound]
     Let $X_1, \ldots, X_S$ be $S$ identically distributed independent real-valued random variables in $[a, b]$ satisfying $\mathbb{E}\left[X_i\right]=\mu$ for all $i$, and denote $\bar{\mu}=\frac{1}{S} \cdot \sum_{i=1}^S X_i$. Then for any $\tau > 0$, we have that
\begin{align}
     \operatorname{Pr}[\bar{\mu} \geq (1 + \tau) \mu] \leq \exp(-\frac{2S\tau^2 \mu^2}{(b-a)^2}), \\
         \operatorname{Pr}[\bar{\mu} \leq (1 - \tau) \mu] \leq \exp(-\frac{2S\tau^2 \mu^2}{(b-a)^2}).
\end{align}
As a result, $\bar{\mu} \approx_{\tau} \mu $ with probability at most $2 \exp(-\frac{2S\tau^2 \mu^2}{(b-a)^2})$.
\begin{proof}
    The Markov inequality and the above Hoeffding's lemma. 
\end{proof}
\end{proposition}

\begin{corollary}\label{cor: hoeffding-bound-relative}
    Let $X_1, \ldots, X_S$ be $S$ identically distributed independent real-valued random variables in $[-a, a]$ satisfying $\mathbb{E}\left[X_i\right]=\mu$ for all $i$, and denote $\bar{\mu}=\frac{1}{S} \cdot \sum_{i=1}^S X_i$. Then for any $\tau > 0$, we have that the probability $\bar{\mu} \approx_{\tau, 0} \mu $ is at least $1 - 2 e^{-S\tau^2 \mu^2 / (2a^2)}$.
\end{corollary}

It can be converted to the additive precision.

\begin{corollary}\label{cor: hoeffding-bound-additive}
    Let $X_1, \ldots, X_S$ be $S$ identically distributed independent real-valued random variables in $[-a, a]$ satisfying $\mathbb{E}\left[X_i\right]=\mu$ for all $i$, and denote $\bar{\mu}=\frac{1}{S} \cdot \sum_{i=1}^S X_i$. Then for any $\varepsilon > 0$, we have that the probability $\bar{\mu} \approx_{0, \varepsilon} \mu $ is at least $1 - 2 e^{-S\varepsilon^2 / (2a^2)}$.
\end{corollary}

We can generalize it to the case of complex-valued random variables. 

\begin{proposition}[Chernoff-Hoeffding Bounds on complex-valued RVs]\label{prop: hoeffding-complex-|a|}
      Let $X_1, \ldots, X_N$ be $N$ identically distributed independent complex-valued random variables satisfying $\left|X_i\right| \leq a$ and $\mathbb{E}\left[X_i\right]=\mu$ for all $i$, and denote $\bar{\mu}=\frac{1}{S} \cdot \sum_{i=1}^S X_i$. Then for any $\varepsilon > 0$, we have that the probability $\bar{\mu} \approx_{0, \varepsilon} \mu $ is at least $1 - 4 e^{-S\varepsilon^2  / a^2}$
      \begin{proof}
          Apply the above Corollary~\ref{cor: hoeffding-bound-relative} to $\operatorname{Im}(X)$ and $\operatorname{Re}(X)$ respectively.  Then we have that $|X| \leq a$ when $|\operatorname{Re}(X)| \leq a/\sqrt{2}$ and $|\operatorname{Im}(X)| \leq a/\sqrt{2}$, from triangle inequality. Then it follows from the union bound, the probability $\bar{\mu} \approx_{0, \varepsilon} \mu$ at least with probability $1 - 4e^{-S\varepsilon^2 / a^2}$. 
      \end{proof}
\end{proposition}

=

\section{Identifiability from Fourier analysis}\label{sec: learnability-general}

In this section, we provide detailed and formal statements of the learnability results, and discuss efficiency via compressed sensing. The results apply to the general setting with stabilizer code and circuit alike. Hence, we adopt a more abstract view in the hope of maintaining generality as well as rigor, while minimizing changes in notation.

\subsection{Reparameterizing the fault distributions}\label{subsec: parameterization-fault-distribution}

In what follows, we consider exclusively the Boolean group $\mathcal{A}$. A \emph{distribution} $f: \mathcal{A} \rightarrow \mathbb{R}_{+ \cup 0}$ is a function evaluated on $\mathcal{A}$ that is normalized $\sum_{a \in \mathcal{A}} f(a) = 1$. In the case of Boolean group $\mathcal{A} \cong \mathbb{F}^n_2$ , the characters are $\{ +1, -1\}$-valued. Hence, it suffices to work in the \emph{real-valued} group algebra $\mathbb{R}[\mathcal{A}]$. In this case, we can also associate the symmetric, non-degenerate bilinear form $\beta: \mathcal{A} \times \mathcal{A} \rightarrow \mathbb{F}_2$ and the isomorphism $\mathcal{A} \cong \mathcal{A}^\perp$ via the association of character $\chi_a = (-1)^{\beta(\cdot, a)}$. It is desired to write its Fourier-transformed version in the \emph{character basis}
\begin{align}\label{eq: pauli-character-fourier}
    \widehat{f} = \prod_{a \in \mathcal{K}} ((1-q_a) + q_a\chi_{a}); \quad \widehat{f} = *_{a \in \mathcal{K}} ((1-q_a) + q_a\chi_{a}),
\end{align}
where the notation $*$ denotes the convolution operator. Let $\mathcal{K} \subseteq \mathcal{A}$ denotes the independent fault processes. A key result we will also use is 
\begin{lemma}
    Let Boolean group $\mathcal{A}$ be endowed with the bilinear symmetric and non-degenerate form $\beta$. Denote $\beta$ with the matrix form whose matrix entry $\beta(a, b)$ for $a, b \in \mathcal{A}$ and $a, b \neq 0$. Then the matrix form $\beta$ has full rank. 
    \begin{proof}
    Note that the matrix entry $\beta(a, b) = -\frac{1}{2}(\chi_a(b) -1)$.  Let's define $h(s) = \sum_{b \neq 0} u_b \chi_b(s)$ and for some $u$ such that  $-2\beta u  =0$. This implies that $0 = \sum_{b \neq 0} u_b(\chi_b(s) -1) = h(s) - h(0)$. Since characters are linearly independent from each other, this implies that $u=0$.  Similarly, we can prove that the columns, where we define $h(a) = \sum_{s \neq 0} u_s \chi_a(s) =\sum_{s \neq 0} u_s \chi_s(a) $. This is possible since there always exists such a symmetric bilinear form. Then, by the symmetry of the argument, columns are linearly independent.
    \end{proof}
\end{lemma}

\begin{lemma}\label{lemma: fourier-character-basis-sign}
For any distribution $f$ defined on a Boolean group $A$, then we can always write $\widehat{f}$ as in Eq.~\eqref{eq: pauli-character-fourier}, provided $|\widehat{f}(a)| > 0$ for all $a \in \mathcal{A}$.
\begin{proof}
Taking $-\log |\widehat{f}| =:y$ expressed in the character basis~Eq.~\eqref{eq: pauli-character-fourier}, we arrive at the following set of linear equations,
\begin{align}
y = \beta x,
\end{align}
where $x_a = - \log(1 - 2q_a)$. By the above lemma, $\beta$ is full-rank for $a, b \neq 0$. Hence, there must exist a unique set of choices $x$ satisfying the above linear equations. In the case that some of the Fourier coefficients are negative, we need to adjust some parameters to above $0.5$. The issue is that if this adjustment can uniquely correspond to the sign pattern of $\widehat{f}$, provided none is vanishing. Let’s denote the vector $z$ such that $z_a = 1$ if $\widehat{f}(a) < 0$; otherwise $z_a =0$. Write $\varepsilon(a) = \operatorname{sgn} \widehat{f}(a)$ and $\sigma_a = \operatorname{sgn} (1-2q_a) = (-1)^{m_a}$ for $a \in \mathcal{K}$, where $m_a$ takes value of ${0, 1}$.  Then we arrive at the following equations,
\begin{align}
\varepsilon_b  = \prod_{a\in \mathcal{K}: \beta(b, a) =1} \sigma_a \iff z_b = \sum_{a\in \mathcal{K}: \beta(b, a) =1} m_a = \sum_{a \in \mathcal{K}} \beta_{ba} m_a.
\end{align}
Again by the full (column) rank of $\beta$, the sign is determined uniquely as a result.
\end{proof}

\end{lemma}

In the case with $\widehat{f}(b) = 0$ for some $b \in \mathcal{A}$, which relates to tuning some parameters $q_a = 1/2$ for $a \in \mathcal{A}$. We ask if the parameterization in Eq.~\eqref{eq: pauli-character-fourier} still holds. The understanding is given as follows,

\begin{lemma}
Let $f$ be a distribution of the Boolean group $\mathcal{A}$ and let $W := { a\in \mathcal{A}: \widehat{f}(a) \neq 0 }$. Then the following are equivalent,
\begin{enumerate}
\item $f$ admits a parametrization of the form in Eq.~\eqref{eq: pauli-character-fourier}.
\item $W$ is a subgroup (linear subspace) of $\mathcal{A}$, defined by all elements that commute with all $b \in \mathcal{K}: q_b = 1/2$.
\end{enumerate}
\begin{proof}
Let $f$ be written as in Eq.~\eqref{eq: pauli-character-fourier} and suppose that for some $a \in \mathcal{A}$, $\widehat{f}(a) = 0$. Then there must exist some $b \in \mathcal{K}$ such that $\beta(a, b) = 1$ and $q_b = 1/2$. In this case, we have
\begin{align}
\left( (1-q_b)I + q_b \chi_b \right) = \frac{I + \chi_b}{2} \equiv \Pi_b.
\end{align}
The final expression indicates that this is given by the projection operator $\Pi_b$, which enforces that any $a’ \in \mathcal{A}$, if $\beta(a’, b) =1$, then $\widehat{f}(a’) = 0$. In other words,  $W := { a\in \mathcal{A}: \widehat{f}(a) \neq 0 }$ must lie within the subgroup (subspace) identified with the projection,
\begin{align}
\prod_{b \in \mathcal{K}: q_b = 1/2} \frac{I + \chi_b}{2}.
\end{align}
By construction, $W$ must also contain this subgroup, which proves in one direction. For the other direction, we denote the orthogonal complement of $W$ by $W^{\perp}$ with respect to $\beta$. Note that $W$ is again a Boolean group, so that we can use the above lemma to conclude the results.
\end{proof}
\end{lemma}

\begin{remark}
    The parametrization set $\mathcal{K}$ in the above case needs not be unique, which in a line with our characterization the effective distribution in Section~\ref{subsec: learnability-general-fault}. 
\end{remark}

Hence, the parameterization in Eq.~\eqref{eq: pauli-character-fourier} cannot capture every distribution (or Pauli channel) over $\mathbb{R}[\mathcal{A}]$. On the practical regime, it is not very harmful as we can always assume $\widehat{f}(a) = O(\epsilon)$ for some arbitrary small perturbative factor $\epsilon$. We now give an example for illustration. 

\begin{example}[Distribution with zero Fourier coefficients]
    We consider the single-qubit Pauli channel viewed as a symplectic form $f$, 
    \begin{align}
        f(I) =0.4, \quad f(Z) = 0.2, \quad f(X) =0.1, \quad f(Y) =0.3. 
    \end{align}
    Then we have that $\widehat{f}(X) = f(I) - f(Z) + f(X) - f(Y) =0$. Then we could either choose $f(Y) =1/2$ or $f(Z) = 1/2$. However, we have that $\widehat{f}(Z) \neq 0$ and $\hat{f}(Y) \neq 0$, so that the parameterization Eq.~\eqref{eq: pauli-character-fourier} does not hold. 
\end{example}

\begin{example}[Pauli-Lindblad parameterization]
The parameterization bears close resemblance to the \emph{Pauli-Lindblad} form of Pauli channel form~Refs.~\cite{van_den_Berg_2023, van_den_Berg_2024}. More precisely, the Pauli-Lindblad form stipulates that the coefficient $1/2 > q_a > 0$ for $a \in \mathcal{K}$. This assumption in general does not hold, and we allow the case with a negative coefficient. A simple example is given with the single-qubit Pauli channel
\begin{align}
    \begin{aligned}
        \Lambda &= p_{I} I \cdot I  + p_{X} X \cdot X + p_Y Y \cdot Y + p_Z Z \cdot Z, \\
    &= ((1-q_X) I \cdot I + q_X X \cdot X ) *  ((1-q_Y) I \cdot I + q_Y Y \cdot Y ) * ((1-q_Z) I \cdot I + q_Z Z \cdot Z ).
    \end{aligned}
\end{align}
The Pauli eigenvalues are given respectively by, 
\begin{align}
    \begin{aligned}
        \lambda_X &= 1 - 2p_Y - 2p_Z = (1 - 2q_Y) (1-2q_Z), \\
        \lambda_Y &= 1 - 2p_X - 2p_Z  = (1 - 2q_X) (1-2q_Z), \\
        \lambda_Z &= 1- 2p_X -2 p_Y= (1 - 2q_X) (1-2q_Y). 
    \end{aligned}
\end{align}
This provides a simple relation between these two parameterizations. Before characterizing the general condition, note that if we take $p_X = 0.19, p_Y = 0.05, p_Z = 0.23$ (while $p_I = 0.53$), then we would have $q_Y \approx -0.1$. 
\end{example}

Recall that any Boolean group $\mathcal{A} \cong \mathbb{F}^n_2$, there exists another way parameterizing the Fourier distribution $\widehat{f}$ via the so-called the canonical moments via the inlcusion-exclusion principle. 

\begin{definition}[String partial ordering] For any $a, b \in \mathbb{F}^n_2$. We say $a \geq b$ if either $b_i = 0$ or $b_i = a_i$ for all $i \in [n]$. 
\end{definition}

\begin{definition}[M\"{o}bius function and M\"{o}bius inversion] The Möbius function $\mu$ of $\mathbb{F}^n_2$ is the function $\mu: \mathbb{F}^n_2 \times \mathbb{F}^n_2 \mapsto \mathbb{R}$ such that for any two functions $f, g: \mathbb{F}^n_2 \mapsto \mathbb{R}$,

$$
f(t)=\prod_{s \leq t} g(s),
$$

if and only if

$$
g(t)=\prod_{s \leq t} f(a)^{\mu(s, t)}.
$$

\end{definition}

\begin{lemma}[Canonical moments]\label{lemma: canonical-moment}
     Given $\widehat{f}$ the moments of the error channel. Its canonical moment is defined by for $a \in \mathbb{F}^n_2$, 
    \begin{align}
        F(a) = \prod_{b \leq a}\widehat{f}(b)^{\mu(b, a)},
    \end{align}
    where $\mu(b, a)$ is given by the M\"{o}bius function such that
    \begin{align}
        \mu(b, a)  \;=\;
\begin{cases}
(-1)^{|a| - |b|}, & \text{if }\, b \leq a\\[6pt]
0,   & \text{otherwise}. 
\end{cases}
    \end{align}
    \begin{proof}
        We first prove that $\mu(b, a)$ is indeed a M\"{o}bius function. First we observe a useful identity, 
        \begin{align}
            \begin{aligned}
                &\sum_{b: c \leq b \leq a}\mu(b, c) = \sum_{b: c \leq b \leq a}  (-1)^{|b|-|c|} \\
                &= \sum^{|a|}_{i=|c|} \binom{|a| -|c|}{i-|c|} (-1)^{i -|c|} = (1-1)^{|a| -|c|} = \delta_{ca}.
            \end{aligned}
        \end{align}
        Then it follows, 
        \begin{align}
            \begin{aligned}
                \prod_{b\leq a} F(b) &= \prod_{b \leq a} \prod_{c \leq b} \widehat{f}(c)^{\mu(c, b)} = \prod_{b, c \:: c \leq b \leq a}\widehat{f}(c)^{\mu(c, b)} = \prod_{c \:: c \leq b \leq a} \widehat{f}(c)^{\sum_{b: c\leq b \leq a }\mu(c, b)} \\
                &= \prod_{c \:: c \leq b \leq a} \widehat{f}(c)^{\delta_{c,a}} = \widehat{f}(a).
            \end{aligned}
        \end{align}
    \end{proof}
\end{lemma}

We now formally establish the connection between the canonical moments and the parameterization in Eq.~\eqref{eq: pauli-character-fourier} and establish their formal equivalence. 

\begin{lemma}
    Suppose that $\widehat{f}$ with $\widehat{f}(s) > 0$ for all $s \in A$ is written in Eq.~\eqref{eq: pauli-character-fourier}, then there exists a one-to-one correspondence between the canonical moments and parameters from Eq.~\eqref{eq: pauli-character-fourier}. 
    \begin{proof}
        As standard, we take the logarithm over the canonical moments. From Lemma~\ref{lemma: canonical-moment}, we consider that 
        \begin{align}
            \begin{aligned}
                F(a) = \prod_{b \leq a} \widehat{f}(b)^{\mu(b, a)} &= \prod_{b \leq a} \left(  \prod_{c\in A, c \neq 0}  ((1-q_c) + q_c \chi_c(b)) \right)^{\mu(b, a)} \\
                &=\prod_{b \leq a} \left(  \prod_{c\in A, c \neq 0}  (1-2q_c)^{\langle b, c \rangle} \right)^{\mu(b, a)}.
            \end{aligned}
        \end{align}
        Note $F(a) > 0$.  
        Taking the logarithm, we obtain a system of linear equations. 
        \begin{align}
            \log F(a) = \sum_{b \leq a} \sum_{c \neq 0, c \in A} \mu(a, b) \beta(b, c) \log\left(   (1-2q_c) \right). 
        \end{align}
        Hence, it is equivalent to show that the matrix $T_{ac} = \sum_{b \leq a} \mu(a, b)\beta(b, c)$ is full-rank. Note that we have established that $\beta$ has full-rank, so it remains to prove $\mu$ attains full-rank as well. To see that, noted from the definition of the M\"{o}bius transformation over the Boolean lattice, it is upper-triangular matrix with non-vanishing diagonals. Hence, we arrive at the desired conclusion.
    \end{proof}
\end{lemma}

\begin{remark}
    In the case $\beta$ is given by the standard $\mathbb{F}_2$-inner product or the symplectic inner product (which can be viewed as standard $\mathbb{F}_2$-inner product with element transpose), we can directly prove this result. To show this, it suffices to show that $T$ can be written in a triangular form with nonzero diagonal entry. To see this, compute 
        \begin{align}
            \begin{aligned}
                T_{ac} &= \sum_{b \leq a} \mu(b, a)\beta(b, c) = \sum_{b \leq a} (-1)^{|a|- |b|} \beta(b, c ) \\
                &= \sum_{\substack{b \leq a \\ |c \cap b| \text{ is odd}}} (-1)^{|a| - |b|} \\
                &= \sum_{\substack{d \leq s \\
                d \text{ is odd}}} \sum_{e \leq r} (-1)^{|a| - |d| - |e|} = (-1)^{|a|} \sum_{\substack{d \leq s \\
                d \text{ is odd}}} (-1)^{|d|} \sum_{e \leq r}  (-1)^{|e|}. 
            \end{aligned}
        \end{align}
        Where $s = a \cap c$ and $r = a \setminus s$. Notice that we have that $\sum_{e \leq r} (-1)^{|e|} = (1-1)^{|r|} = 0$ otherwise $r =0$. Hence, $s = a$ and $a \leq c$. Then it is easy to observe that this partial ordering ensures that this matrix $T$ is given in the triangular form. To notice that the diagonal term never vanishes, we compute that 
        \begin{align}
            T_{aa} = (-1)^{|a|} \sum_{\substack{b \leq a \\
            b \text{ is odd}}} (-1)^{|b|} = (-1)^{|a|} \sum_{j \text{ is odd}} (-1)^j \binom{|a|}{j}.
        \end{align}
        Note that for $a =0$, $T_{aa} =1$. Using the relation for $|a| > 0$, 
        \begin{align}
            \sum_{j \text{ is odd}} (-1)^j \binom{|a|}{j} =- \frac{1}{2} ((1+1)^{|a|} - (1-1)^{|a|}) = -2^{|a|-1}, 
        \end{align}
        which implies $T$ is full-rank. 
\end{remark}

This result should not come as unexpected, as both give a way to parameterize the same Pauli channels, which by construction, should be equivalent to each other. One can also prove the above result from another direction by expression $\widehat{f}(b)$ in terms of the canonical momentum.


\subsection{Learnability of fault distribution}\label{subsec: learnability-general-fault}

We state our results on the elementary-$2$ Abelian group, which encompasses the learnability questions of the Pauli channels and classical-Pauli noise discussed in~\cite{Wagner_2022, Wagner_2023} such as the data-syndrome code~\cite{ashikhmin2020quantum}. Furthermore, it encompasses consideration of the circuit noise in the spacetime formalism in Section~\ref{subsec: code-circuit-iso}, which we will discuss in later sections. We start with a generic elementary-$2$ group $\mathcal{A} \cong \mathbb{F}^n_2$ and consider a distribution $f \in \mathbb{R}[\mathcal{A}]$. The learnability question can be given in two parts, both are intertwined with the concept of error-correcting codes.

\begin{definition}
    A binary linear error-correcting code $\mathcal{C}$ is a $k$-dimensional subspace of $\mathbb{F}^n_2$, denoted by $[n, k]$. Its dual (parity-check) code $\mathcal{C}^{\perp}$ $[n, n-k]$ is the orthogonal complement of $\mathcal{C}$. 
\end{definition}

 It is worth noting that classical linear error-correcting codes bear a close relation to the quantum stabilizer codes. We refer to $\mathcal{C}$ as a 
weakly self-dual code if $\mathcal{C}\subset \mathcal{C}^\perp$. To promote this into a quantum code, we place $X$ and $Z$ stabilizers both with the weakly self-dual code and the encoded logical information is given by $\mathcal{C}^\perp \setminus \mathcal{C}$. Mathematically, this can be represented as a $2$-term chain complex and is an instance of the Calderbank–Shor–Steane (CSS) code. It is also proven that any quantum stabilizer codes can be mapped to this instance of codes \cite{Bravyi_2009}. In this sense, the self-duality of the underlying code (subgroup of $\mathcal{A}$) plays a pivotal role, and a physical explanation is that this gives a set of simultaneous measurement operators. Hence, to formally motivate our study, let's consider $\mathscr{G} \subseteq \mathcal{A}$ to be a subgroup. We denote its center by $\mathscr{M} = \mathscr{G}^\perp \cap \mathscr{G}$. This subgroup $\mathscr{G}$ is referred to as the gauge group. It is more convenient to ask the learnability question taking the Fourier transform. 

\begin{enumerate}
    \item Knowing the Fourier coefficient $\widehat{f}1_{\mathscr{M}}$, we aim to learn the entire Fourier coefficient $\widehat{f}$. \label{question: physical-learn-fourier}
    
    \item Knowing the Fourier coefficient $\widehat{f}1_{\mathscr{M}}$, we aim to learn the Fourier coefficients $\widehat{f}1_{\mathscr{G}^\perp}$. \label{question: logical-learn-fourier}
\end{enumerate}

For obvious reasons, we refer to Question~\ref{question: physical-learn-fourier} as the learnability question and Question~\ref{question: logical-learn-fourier} as learnability up to logical equivalence. Performing the inverse Fourier transform, we obtain
\begin{align}\label{eq: eff-distribution}
    \mathscr{F}^{-1}[\hat{f}1_{\mathscr{G}^\perp}] = \frac{1}{|\mathscr{G}|}f * 1_{\mathscr{G}} \equiv f^{\operatorname{eff}},
\end{align}
where we obtain a distribution that is invariant under action from $\mathscr{G}$, which corresponds to the effective distribution. A closely related concept is referred to as the \emph{prior} distribution (See Definition~\ref{def: prior-distribution}), a convenient reparameterization of $f$ by removing the degrees of freedom associated with logical equivalence $\sim_{\mathscr{G}}$. Here we build mathematical machinery to tackle these problems. Note that for Boolean groups, the characters are real-valued $\{-1, +1\}$ and hence, it is safe to restrict any distribution $f \in \mathbb{R}[\mathcal{A}]$ with only the real-coefficients for $\widehat{f}$. As we have seen, in many technical adaptations, we take the logarithm of the Fourier distribution. In what follows, we assume $\widehat{f}(a) > 0$ for all $a \in \mathcal{A}$ so that taking the logarithm provides a natural bijection. Note that it is a mild assumption generally assumed across the literature~\cite{chen2024efficient, chen2025disambiguatingpaulinoisequantum, Chen_2023_learnability, van_den_Berg_2023, van_den_Berg_2024, Flammia_2020, Zhang_2025, Wagner_2023, Wagner_2022} as the probability for which a fault occurs is typically small. Denote the Boolean group $\mathcal{A}$ with group order $N= 2^n$ so that $\mathcal{A} \cong \mathbb{F}^n_2$ equipped with the symmetric and non-degenerate bilinear form $\beta$. Let $\mathscr{G} \subseteq \mathcal{A}$ be a subgroup. Its \emph{syndrome map} $\sigma(a) $ is given as a vector over $\mathbb{F}_2$ whose entry is given by $\sigma(a)_s = \beta(a, s)$ for $s \in \mathscr{M}$. Then we say that elements $a, b$ correspond to the same syndrome if $\sigma(a) = \sigma(b)$ and if they are logically equivalent if $a \sim_{\mathscr{G}} b$ (that is, there exists $g \in \mathscr{G}$ such that $a = gb$). Denote $K = |\mathcal{K}_\mathscr{M}|$, records the number of independent fault processes from $\mathcal{K}$ that share the same syndrome (excluding the zero syndrome). We state the following intuitive result, 

\begin{lemma}\label{lemma: learn-all-distribution-from-syndrome}
    For a distribution $\widehat{f} \in \mathbb{R}[\mathcal{A}]$ written in the parametrization of form Eq.~\eqref{eq: pauli-character-fourier}. Then for any nonzero element $a \in \mathcal{K}$, we write a matrix $A_{\mathscr{M}} \in \mathbb{R}^{(|\mathscr{M}|-1) \times |\mathcal{K}| }$ whose entry is given by $D[s, a] = \beta(s, a)$ for $s \neq 0 $ and $s \in \mathscr{M}$, $a \in \mathcal{K}$. Then $A_{\mathscr{M}}$ achieves full (column)-rank if and only if each element in $\mathcal{K}$ corresponds to a unique non-zero syndrome. 
    \begin{proof}
        We first prove the forward direction. Note that to achieve full-column rank $|\mathcal{K}|\leq |\mathscr{M}|-1$ and suppose that there are two elements $a, b \in \mathcal{K}$ that correspond to the same syndrome. Then $A_{\mathscr{M}}$ would necessarily have two identical columns, which implies that $A_{\mathscr{M}}$ cannot achieve full-column rank. For the other direction, we prove this by a systemic restriction of $\beta$ to a subgroup where it preserve the non-degeneracy. Recall that $\mathscr{M}$ is a subgroup of $\mathcal{A}$ so it can be generated by $r$ generators $\{g_1, \cdots, g_r\}$. Hence, the homomorphism, 
        \begin{align}
            \phi: \mathscr{M} \rightarrow \mathbb{F}^r_2; \quad \phi(g^{\alpha_1}_1 \cdots g^{\alpha_r}_r) = (\alpha_1, \cdots, \alpha_r),
        \end{align}
        is an isomorphism with respect to the group rule, where $\alpha_1, \cdots, \alpha_r$ are $\mathbb{F}_2$-valued variables. Furthermore, we can define another map, 
        \begin{align}
            \psi: \mathcal{A} / \mathscr{M}^\perp \rightarrow \mathbb{F}^r_2; \quad \psi([a]) = \sigma(a).
        \end{align}
        It is easy to see that $\psi$ is another isomorphism. Hence, there must exist a symmetric and non-degenerate bilinear form $\gamma$ defined on $\mathbb{F}^r_2$ such that $\gamma(u, v) = \beta(\phi^{-1}(u), \psi^{-1}(v))$. Hence, the result follows from the full-(column) rankness of $\gamma$ (excluding the identity). 
    \end{proof}
\end{lemma}

\begin{example}
    We now give a concrete example to the above lemma. Using the stabilizer group $\mathscr{M}$ and Pauli group $\mathrm{P}^n$. Define the following map
    \begin{align}
        \sigma: \mathrm{P}^n \rightarrow \mathbb{F}^r_2.
    \end{align}
For any Pauli element $a \in \mathrm{P}^n$, we can express it by $\sigma(a) \in \mathbb{F}^r_2$ where $\sigma(a)_i = 1$ if $a$ anti-commutes with the stabilizer generator $s \in \mathscr{M}$. Note that this identification is not injective but satisfies the "homomorphism property": $\sigma(ab) = \sigma(a) + \sigma(b)$ where the addition is over $\mathbb{F}_2$. One can construct an isomorphism map which is given by the quotient $\bar{\sigma}: \mathrm{P}^n / \mathscr{M}^{\perp} \rightarrow \mathbb{F}^r_2 $ where the commutant $\mathscr{M}^\perp = \mathscr{L}$. Furthermore, for elements from the stabilizer group $\mathscr{M}$, we can always write $s = \prod^r_{j=1}s^{\psi(s)_i}_i$ where $\psi(s)_i \in \mathbb{F}_2$ where it equals to $1$ if $s$ consists of product of $s_i$. This induces a natural bijection from $s \mapsto \psi(s) \in \mathbb{F}^r_2$. Then it is easy to see that for any coset $[a] \in \mathrm{P}^n / \mathscr{M}^\perp$ and any $a \in [a]$, $s \in \mathscr{M}$, 

\begin{align}
     \beta(a, s) = \sum^r_{j=1}\beta(a, s^{\psi(s)_i}_i)  = \sum^r_{j=1}\langle \sigma(a), {\psi(s)_i} \rangle_{\mathbb{F}^r_2} = \langle \sigma(a), \psi(s) \rangle_{\mathbb{F}^r_2},
\end{align}
where $\beta(\cdot, \cdot)$ denotes the standard symplectic inner product for the Pauli group and $\beta(a, s^{\psi(s)_i}_i) =1$ if $\psi(s)_i =1$ and $a, s_i$ anticommutes, which explains the second equality. The third equality follows by re-expressing $s$ by $\psi(s)$.   
\end{example}

We now are able to state the result of learnability. 

\begin{theorem}\label{thm: learn-all-distribution-from-syndrome}
    For a Boolean group $\mathcal{A}$ associated with a symmetric and non-degenerate bilinear form $\beta$. Denote distribution $f \in \mathbb{R}[A]$ and its Fourier distribution $\widehat{f} > 0$. Then $\widehat{f}$ can be learned from $\widehat{f}1_{\mathscr{M}}$ if and only if $\mathcal{K}_\mathscr{M} = \mathcal{K}$. 
    \begin{proof}
        Parameterizing $\widehat{f}$ in the character basis. Take the logarithm $\log \widehat{f} 1_{\mathscr{M}} = D x$ for $x_a = -\log (1-2q_a)$ for $a \in \mathcal{K}$. By the above lemma $D$ achieves full-column rank; hence, $x$ can be solved uniquely knowing $\log \widehat{f} 1_\mathscr{M}$, which implies the result. 
    \end{proof}
\end{theorem}

\begin{remark}
    In this case, it is possible to further relax the condition that $|\widehat{f}| > 0$, analogous to the treatment in Lemma~\ref{lemma: fourier-character-basis-sign}, where $A_{\mathscr{M}}$ is of full-column rank. 
\end{remark}

Here we investigate the question~\ref{question: logical-learn-fourier}, specializing to a Boolean group $\mathcal{A} \cong \mathbb{F}^n_2$ with a symmetric and non-degenerate bilinear form $\beta$. Let $\mathscr{G} \subseteq  \mathcal{A}$ and its dual $\mathscr{G}^\perp$ by elements $b \in \mathscr{G}^\perp$ such that $\beta(a, b) =0$ for all $a \in \mathscr{G}$. Note that we denote the group center is given by $\mathscr{M}$. It is desirable to represent the bilinear form $\beta$ by the following matrices,

\begin{align}
        \begin{aligned}
            D_{\mathscr{G}^\perp} = \left(\begin{array}{c|c} 
  A_{\mathscr{M}} & B_{\mathscr{M}} \\ 
  \hline 
  A_{\mathscr{G}^\perp / \mathscr{M}} & B_{\mathscr{G}^\perp / \mathscr{M}} 
\end{array}\right)  \quad D_{\mathscr{M}} = \begin{pmatrix}
    &A_{\mathscr{M}} & | & B_{\mathscr{M}}&
\end{pmatrix}.
        \end{aligned}
    \end{align}
The subscripts in matrices $A, B$ categorize the row indices. The columns of $A$ are the elements $a \in \mathcal{K}$ that correspond to distinct syndromes of $\mathscr{G}$ associated with $\mathscr{M}$ and the columns of $B$ record the elements in $\mathcal{K}$ whose syndromes are repeated from columns in $A$.

\begin{lemma}
    We have that $\operatorname{rk} D_{ \mathscr{G}^\perp} = \operatorname{rk} D_\mathscr{M}$ if and only if, for all $a, b \in \mathcal{K}$, $\sigma(a) = \sigma(b) \iff a \sim_{\mathscr{G}} b$. 
    \begin{proof}
        We first show the forward direction. We perform a rank-preserving column operation 
        \begin{align}
            D_{\mathscr{M}} = \begin{pmatrix}
    &A_{\mathscr{M}} & | & B_{\mathscr{M}}&
\end{pmatrix} \mapsto  D'_{\mathscr{M}} = \begin{pmatrix}
    &A_{\mathscr{M}} & | & \bvec{0} &
\end{pmatrix},
        \end{align}
where columns in $B_{\mathscr{M}}$ record only the redundant syndromes from columns of $A_{\mathscr{M}}$. Similarly, we could perform the column operation,
    \begin{align}
         D_{\mathscr{G}^{\perp}} = \left(\begin{array}{c|c} 
  A_{\mathscr{M}} & B_{\mathscr{M}} \\ 
  \hline 
  A_{\mathscr{G}^\perp / \mathscr{M}} & B_{\mathscr{G}^\perp / \mathscr{M}} 
\end{array}\right) \mapsto D'_{\mathscr{G}^\perp} = \left(\begin{array}{c|c} 
  A_{\mathscr{M}} & \bvec{0} \\ 
  \hline 
  A_{\mathscr{G}^\perp / \mathscr{M}} & B'_{\mathscr{G}^\perp/\mathscr{M}} 
\end{array}\right).
    \end{align}
Suppose that there exists a pair $a, b \in \mathcal{K}$ with $\sigma(a) = \sigma(b)$ and $a = b c$ for some $c \in \mathscr{M}^{\perp} \setminus \mathscr{G}$. Then there must exist some element $c' \in \mathscr{G}^{\perp} \setminus \mathscr{M}$ such that $\beta(c', c) = 1$ so that $\beta(c', a) \neq \beta(c', b)$. That is true since otherwise $c \in \mathscr{M}^\perp \setminus \mathscr{G} \cap (\mathscr{G}^{\perp} \setminus \mathscr{M})^{\perp}$. If $c \in(\mathscr{G}^{\perp} \setminus \mathscr{M})^{\perp} $, then $\beta(c, l) =0$, $\forall l \in \mathscr{G}^\perp \setminus \mathscr{M}$. However, since $c \in \mathscr{M}^{\perp}$, which implies that $\beta (c, m) =0 $ for all $m \in \mathscr{M}$ so that $c \in (\mathscr{G}^{\perp})^{\perp} = \mathscr{G}$ (by the fact that $\beta$ is non-degenerate), which contradicts the assumption that $c \notin \mathscr{G}$.  This implies that $B'_{\mathscr{G}^\perp} \neq 0$, which shows that $\operatorname{rk} D_{\mathscr{G}^\perp} = \operatorname{rk} D'_{\mathcal{\mathscr{G}^\perp}} \neq \operatorname{rk} D'_{\mathscr{M}} = \operatorname{rk} D_{\mathscr{M}}$. The converse direction follows similarly where we observe that $A_{\mathscr{M}}$ is always full-column rank. 
    \end{proof}
\end{lemma}

\begin{remark}
    Note that the assumption $\beta$ is non-degenerate is very important, otherwise, we would only have $\mathscr{G} \subseteq (\mathscr{G}^{\perp})^{\perp}$ and the above necessary condition might not work. Luckily, the existence of non-degenerate (symmetric) bilinear form for Boolean group is guaranteed by Lemma~\ref{lemma: exist-nondegenerate-symmetric-bilinear-forms}. 
\end{remark}

Taking the exponential from logarithm whereupon a bijection is achieved provided $\widehat{f} > 0$. 

\begin{theorem}\label{thm: learn-eff-distribution-from-syndrome}
    For any distribution $f \in \mathbb{R}[\mathcal{A}]$ such that its Fourier distribution $\widehat{f}(a) > 0$ for all $a \in \mathcal{A}$. Then $\widehat{f}1_{\mathscr{G}^\perp}$ is learnable from $\widehat{f}1_{\mathscr{M}}$ if and only if it can be parameterized in character basis such that, $\forall a, b \in \mathcal{K}$, $\sigma(a) = \sigma(b) \iff a \sim_{\mathscr{G}} b$. 
    \begin{proof}
        Taking for any $b \in \mathscr{G}^\perp$ and denote that 
        \begin{align}
            \log \widehat{f}(b) = \sum_{m \in \mathscr{M} } e_m \log \widehat{f}(m), 
        \end{align}
        where the coefficient $e_m$ exists if and only if $\log \widehat{f}(b) \in \operatorname{rs} D_{\mathscr{M}}$.
        Taking the exponential, we obtain the desired result. 
    \end{proof}
\end{theorem}

\begin{remark}
    It seems that we have reached a more relaxed condition in parameterization $\mathcal{K}$, if we wish only to learn up to logical equivalence. The extent to which the condition is relaxed depends on the intersection $\mathscr{G}^\perp \cap \mathscr{G}$. In particular, in the case $\mathscr{G}^\perp \cap \mathscr{G} = \{ 0 \} $, we arrive at the exactly the same condition as in the physical learnability. This is, fortunately, never the case in the quantum setting, as the group center describes precisely the measurement we are able to perform. In this view, the set of commuting measurement operators plays a fundamental role in our study of learnability. 
\end{remark}

\begin{remark}
    We now comment in detail the relation to the learnability results utilizing the parameterization in M\"{o}bius inversion, presented in \cite{Wagner_2022, Wagner_2023}. In particular, for a Boolean group $\mathcal{A} \cong \mathbb{F}^n_2$, we can define the \emph{zeta} function $\zeta(a, b) = 1_{b \leq a}$. Viewing as matrices, we have that $\zeta \cdot T = \zeta \cdot \mu \cdot \beta = \beta$, where we have that $\zeta = \mu^{-1}$. Hence, the learnability matrix defined in terms of the Zeta function is formally related by a simple basis transformation. 
\end{remark}

In a later section, we will emphasize its application in $(i)$ learning static Pauli noise subject to perfect measurement results in Section~\ref{sec: ds-sub-system code}, $(ii)$ spacetime circuits in Section~\ref{subsec: code-circuit-iso}.

\subsection{Low-rank construction via compressed sensing} \label{subsec: compressed-sensing-theory}

Another aspect of developed learning formalism using the character theory is the natural question: what is the sample complexity required if we wish to learn all $K$-degrees of freedom in our parameterized distribution Eq.~\eqref{eq: pauli-character-fourier}. It turns out that this question can be elegantly solved using the technique of compressed sensing, which we now explain. Let $\phi, \phi' \in \mathbb{C}^N$, we denote $\langle \phi, \phi' \rangle = \frac{1}{N}\sum^N_{i=1} \phi^*_i \phi'_i$. We denote $\|\phi\|_\infty := \max_i |\phi_i |$. 

\begin{definition}
    [Bounded orthogonal system] We call $\Phi = \{ \phi^{(1)}, \cdots, \phi^{(K)} \}$ a bounded orthonormal system (BOS) with a constant $C$ if $(1)$ $\langle \phi^{(j)}, \phi^{(k)} \rangle = \delta_{jk}$ and $(2)$
    $\|\phi^{(j)}\|_\infty \leq C$ for all $1 \leq j \leq K$. 
\end{definition}
Note that here we assume $\langle \phi^{(j)}, \phi^{(k)}\rangle = (1/N)\sum_{i} \phi^{*(j)} \phi^{(k)}$, where $\phi^{*(j)}$ denotes the complex conjugate. A close example is given below. 

\begin{example}
    [Hadamard matrix] A canonical example for BOS that is of center of our study is the Hadamard matrix. Let $\mathcal{A}$ be a finite Boolean group equipped with a non-degenerate, symmetric bilinear form $\beta$. For each $a \in \mathcal{A}$, we can associate $\chi_a = (-1)^{\beta(\cdot, a)}$. Let $\Phi = \{ \chi_a, a\in \mathcal{A} \}$. Note that we have
    \begin{align}
        \begin{aligned}
            \langle \chi_a, \chi_b \rangle &= \frac{1}{|\mathcal{A}|} \sum_{\xi \in \mathcal{A}} (-1)^{\beta(\xi, a)} (-1)^{\beta(\xi, b)} = \frac{1}{|\mathcal{A}|} \sum_{\xi \in \mathcal{A}} (-1)^{\beta(\xi, a+b)}, \\
            &= \frac{1}{|\mathcal{A}|} \sum_{\xi \in \mathcal{A}} (-1)^{\beta(\xi + \xi', a+b)} = \frac{1}{|\mathcal{A}|} (-1)^{\beta(\xi', a+b)}\sum_{\xi \in \mathcal{A}} (-1)^{\beta(\xi , a+b)}, \\
            &=\delta_{a, b}.
        \end{aligned}
    \end{align}
    The third equality is given by the group property. For $a + b \neq 0$, then we can always choose some $\xi'$ such that $\langle \xi', a+b \rangle = -1$. Hence, the last equality follows. As a result, $\Phi = \{ \chi_a: a \in \mathcal{A} \}$ forms a BOS with the BOS constant $1/\sqrt{N}$ and we write the matrix form $\Phi$ whose column vectors are given by $\chi_a, a \in \mathcal{A}$. 
\end{example}

\begin{definition}[Restricted isometry property]
     Let $H \in \mathbb{C}^{q \times N}$ for some integer $q, N$. We say that $H$ satisfies the restricted isometry property of order $K$ with constant $\varepsilon$, if for every $K$-sparse $x \in \mathbb{C}^N$   
    \begin{align}
        (1-\epsilon)\|x\|^2_{2} \leq \|Hx\|^2_2 \leq (1+\epsilon)\|x\|^2_2.
    \end{align}
     Furthermore, we say that $H$ satisfies the restricted isometry property of order $K$ associated with some subset $\mathcal{K} \in [N]$ with constant $\epsilon$, if the above relation holds for every $K$-sparse $x \in \mathbb{C}^N$ supported on $\mathcal{K}$. 
\end{definition}

\begin{remark}
     $\epsilon$ is called the $K+1$th-order \emph{restricted isometric constant} if it is the infimum of all possible such $\epsilon$ for a given $H \in \mathbb{C}^{q \times N}$. In our analysis, we do not fixate on taking the infimum. 
\end{remark}
This definition can be stated in following equivalent forms. 

\begin{lemma}
    The following are equivalent. 
    \begin{enumerate}
        \item $H \in \mathbb{C}^{q \times N}$ satisfies the restricted isometry property of order $K$ with constant $\varepsilon$.  
        \item For any $q \times K$ submatrix $H_K$ of $H$ we have 
        \begin{align}
           \| H^\dagger_K H_K - I_{K} \|_{2 \mapsto 2} \leq \varepsilon.
        \end{align}
        \item All the eigenvalues of $H^\dagger_K H_K$ are in the interval $[1-\varepsilon, 1+\varepsilon]$
    \end{enumerate}
    \begin{proof}
        We first prove from $(1)$ to $(2)$. Let $z \in \mathbb{C}^{K}$ with $\|z\|_2 =1$. Then we have that 
        \begin{align}
            z^\dagger (H^\dagger_KH_K  - I_K)z  = \|H_K z \|_2 - \|z\|_2 = \|H_K z\|_2 - 1.
        \end{align}
        By the RIP property, it follows that $|z^\dagger (H^\dagger_KH_K  - I_K)z| \leq \epsilon$. The converse can be proven in a similar fashion. Furthermore, by the property of the induced $2$-norm, $(2) \iff (3)$, so we have proven as desired. 
    \end{proof}
\end{lemma}

In what follows, we denote $\delta_K$ as the $K+1$th-order restricted isometric constant for the reason that will be clear below. We pay special attention to the Hadamard matrix $\Phi$ defined over characters of some Boolean group $\mathcal{A}$. For a distribution $\widehat{f} \in \mathbb{R}[\mathcal{A}]$ parameterized in the character basis Eq.~\eqref{eq: pauli-character-fourier}, we denote $\mathcal{K}_{\mathscr{M}}$ to be the subset of $\mathcal{K}$ consisting of elements with nonzero distinct syndromes. 

\begin{theorem}[RIP for row-subsampled Hadamard matrix $\Phi$]\label{thm: rip-partial-hadamard-naive} 
      Let $\Phi \in \mathbb{C}^{N \times N}$ be the Hadamard matrix associated with $\mathcal{A}$. For a sufficiently small $\delta_K > 0$ and some $q = O(\delta^{-2}_K   \max(K^2 \delta^{-2}_K\log(K) \log(1/\eta), \log(1/\delta))) $, the following holds. Let $H \in \mathbb{C}^{q \times N}$ be a matrix whose $q$ rows are chosen uniformly and independently from the rows, rescaled by $\sqrt{N/q}$. Then with probability at least $1-\delta$, we have that $H$ satisfies the restricted isometry property with the $K+1$th-order restricted isometric constant $\delta_K$ on the supports of $\mathcal{K}_{\mathscr{M}} \cup \{ I\} $ with $K = |\mathcal{K}_{\mathscr{M}}|$. 
\end{theorem}

Note that this proof gives a very loose bound on $q$. In fact, following closely the result by \cite{haviv2015restrictedisometrypropertysubsampled}, we could arrive at the near-optimal upper-bound.

\begin{theorem}\label{thm: rip-partial-hadamard-refined} 
     Let $\Phi \in \mathbb{C}^{N \times N}$ be the Hadamard matrix associated with $\mathcal{A}$. For a sufficiently small $\delta_K > 0$ and some $q$
    \begin{align}
        q = O(\delta^{-2}_K \max( \log^2(K/\delta_K) \log^2(1/\delta_K) \log(K)), \log(1/\delta)).
    \end{align}
    The following holds. Let $H \in \mathbb{C}^{q \times N}$ be a matrix whose $q$ rows are chosen uniformly and independently from the rows, rescaled by $\sqrt{N/q}$. Then with probability at least $1 - \delta$, we have that $H$ satisfies the restricted isometry property with the $K+1$th-order restricted constant $\delta_K$ on the supports of $\mathcal{K}_{\mathscr{M}} \cup \{ I \}$ with $K = |\mathcal{K}_{\mathscr{M}}|$.
\end{theorem}

To apply to our learning framework, we state the following corollary. 
\begin{corollary}\label{cor: compressed-sensing-sample-theory}
Let $\mathcal{A}$ be a Boolean group of order $N$ and $\beta$ the symmetric and non-degenerate bilinear form. Let $Q \subseteq \mathcal{A}$ be a uniformly and independently sampled subset consisting of $q$ elements. Suppose that the row-subsampled and rescaled matrix $H \in \mathbb{C}^{q \times N}$ satisfies the restricted isometry property with the $K+1$th-order restricted isometric constant $\delta_K <1$ on the supports of $\mathcal{K}_{\mathscr{M}}$. Denote $A_{\operatorname{res}} \in \mathbb{C}^{q \times K}$ to be the column restricted submatrix of $A$ over the supports $\mathcal{K}_{\mathscr{M}}$, where $A_{\mathscr{M}, \operatorname{res}}[a, b] = \beta(a, b)$ for $a, b \neq 0$ and $a \in \mathscr{M}$ and $b \in \mathcal{K}_{\mathscr{M}}$. Then $A_{\mathscr{M}, \operatorname{res}}$ has full-column rank. 
\end{corollary}

\subsubsection{Proof of Theorem~\ref{thm: rip-partial-hadamard-naive} \& Theorem~\ref{thm: rip-partial-hadamard-refined}}

We now present the proof of the above theorems. Note that for the case of Theorem~\ref{thm: rip-partial-hadamard-refined}, we only give a proof sketch and comment on (a minor) modification tailored to our case with a preselected region of supports $\mathcal{K}_M \cup \{ I \}$ and some notational conversion.  In what follows, we assume the BOS matrix with normalization unless stated otherwise.

\begin{lemma}\label{lemma: rip-restatement}
    Let $H \in \mathbb{C}^{q \times N}$ be row subsampled matrix from $\Phi$, indexed from the set $Q$, 
    \begin{align}
        H = \sqrt{\frac{N}{q}} R_Q \Phi,  
    \end{align}
    where $R_Q$ is the row selector. Then The following are equivalent, for all $x \in \mathbb{C}^N$. 
    \begin{enumerate}
        \item $  (N/\|x\|^2_2) \cdot \mathbb{E}_{j \in Q}[|(\Phi x)_j|^2] \approx_{\varepsilon/2, \varepsilon/2} (N/\|x\|^2_2) \cdot \mathbb{E}_{j \in [N]}[|(\Phi x)_j|^2]$. 
        \item $\|Hx\|^2_2 \approx_{\varepsilon, 0} \|x\|^2_2$.
    \end{enumerate}
    \begin{proof}
        Note that $\Phi$ is a BOS with normalization. Note that for $x$ with $\|x\|_2=1$, we apply the condition $(1)$ where 
        \begin{align}
            \|Hx\|^2_2 = N \cdot \mathbb{E}_{j \in Q}[|(\Phi x)_j|^2] \approx_{\varepsilon/2, \varepsilon/2} N \cdot \mathbb{E}_{j \in [N]}[|(\Phi x)_j|^2] = \|\Phi x\|^2_2 =1.
        \end{align}
        Hence, we have $\|Hx\|^2_2 \in [(1-\varepsilon/2) - \varepsilon/2, (1+ \varepsilon/2) + \varepsilon/2] = [1 - \varepsilon, 1+\varepsilon]$ and $\|Hx\|^2_2 \approx_{\varepsilon, 0} \|x\|^2_2$ for every $x \in \mathbb{C}^N$, where we added back the normalization. 
    \end{proof}
\end{lemma}

In the light of Lemma~\ref{lemma: rip-restatement}, we will prove the equivalent reformulation of the restricted isometry property as a probabilistic statement.

\subsubsection*{A naive analysis}

We now give an elementary proof which leads to a worse bound on $q$ but suffices in spirit. Let's denote a BOS matrix $\Phi$ with BOS constant $\| \Phi \|_\infty = 1$. We consider how to build an estimator for $\Phi x$, where we can write 
\begin{align}
    \Phi x = \sum^N_{\ell =1} x_\ell \Phi^{(\ell)} = \sqrt{2} \sum^N_{\ell =1} \sum^3_{s=0}p_{\ell, s} (-1)^{s/2} \Phi^{(\ell)},
\end{align}
where for any complex number $x_\ell = r e^{i\theta}$, we can always find $p_{\ell, 0}, p_{\ell, 1}, p_{\ell, 2}, p_{\ell, 3} \geq 0$ such that $\sum^3_{s=0} p_{\ell, s} = r = |x_\ell|$. Then it follows that $\sqrt{2} \sum^3_{s=0} p_{\ell, s} (-1)^{s/2} = x_l$. A concrete parametrization choice is given by $p_{\ell, 0} = \max(r \cos \theta /\sqrt{2}, 0) + L/4$, $p_{\ell, 0} = \max(-r \cos \theta / \sqrt{2}, 0) + L/4$, $p_{\ell, 1} = \max(r \sin \theta /\sqrt{2}, 0) + L/4$, $p_{\ell, 3} = \max(-r\sin \theta /\sqrt{2}, 0) + L/4$., where $L = \frac{r}{\sqrt{2}}(\sqrt{2} - |\sin \theta| - |\cos \theta|) \geq 0 $.  Let's further assume that $\|x\|_1 =1$, so that it follows that $p_{\ell, s}$ gives a probability distribution $\mathcal{D} := \{ (\ell, s, p_{\ell, s}), \: (\ell, s) \in [N] \times (0, 1, 2, 3)\}$. In other words, 
\begin{align}
    \Phi x = \sqrt{2} \sum^N_{\ell =1} \sum^3_{s=0}p_{\ell, s} (-1)^{s/2} \Phi^{(\ell)} = \mathbb{E}_{(\ell, s) \in \mathcal{D}}[\sqrt{2} \cdot (-1)^{s/2} \cdot \Phi^{(\ell)}].
\end{align}
The idea is that we wish to define some empirical distribution $\widehat{p}_{(\ell, s)}$ to estimate the above expectation. We now make this intuition precise. Let's consider the empirical mean 
\begin{align}
    g = \frac{\sqrt{2}}{S} \sum_{(\ell, s) \in \mathcal{F}} (-1)^{s/2} \Phi^{(\ell)}, 
\end{align}
where $\mathcal{F}$ contains $S \geq  \tau^{-2}\log(16/\gamma) $ samples drawn i.i.d from $\mathcal{K} \times (0, 1, 2, 3)$ for $\mathcal{K} \subseteq [N]$.

\begin{lemma}\label{lemma: Maurey-empirical-method}
    Suppose that $S  \geq \tau^{-2} \log(16/\gamma)$ for $\tau, \gamma > 0$. Denote any set $Q \subseteq [N]$. For $j \in Q$ we denote the set $Q'$ consists of all indices such that $|g_j| \approx_{0, \tau}  |\mathbb{E}_{(\ell, s) \in \mathcal{D}}[\sqrt{2} \cdot (-1)^{s/2} \cdot \Phi^{(\ell)}_j]|$. Then we have that with probability at least $3/4$
    \begin{align}
        \frac{|Q'|}{|Q|} = \Omega(1 -\gamma).
    \end{align}
    \begin{proof}
        From the empirical distribution, we have that for every $j \in Q$. Note that $|g_j| \leq 1, |(\Phi x)_j| \leq 1$ due to the assumption that $\|x\|_1 = 1$ and $\|\Phi \|_\infty =1$. Applying Proposition~\ref{prop: hoeffding-complex-|a|}, 
        \begin{align}\label{eq: naive-tau-good-indices}
           |g_j| \approx_{0, \tau}  |(\Phi x)_j|, 
        \end{align}
        with probability at least $1 - 4 \cdot e^{-\log(16/\gamma)} = 1- \gamma/4 $. Hence, in expectation, the number of $j \in Q$ that do not satisfy is at most $\gamma |Q|/4$. Hence, treating the number of these indices do not satisfy Eq.~\eqref{eq: naive-tau-good-indices} as a random variable $X$. By the Markov inequality, the probability for $X > \gamma |Q|$ is bounded by $1/4$, as desired. 
    \end{proof}
\end{lemma}
 
In the light of the above lemma, we say that an index $j \in [N]$ is $\tau$-\emph{good} if $|g_j| \approx_{0, \tau} |\mathbb{E}_{(\ell, s) \in \mathcal{D}}[\sqrt{2} \cdot (-1)^{s/2} \cdot \Phi^{(\ell)}_j]|$. Otherwise, it is called \emph{bad}. The bad indices can be made very small so it would not affect our approximation in the following sense. 

\begin{definition}[One-scale net]
     Let $\mathcal{G}$ denotes all possible such numerical expectation $g$ drawn from different choices of $\mathcal{F}$ with $S$ samples. Further, we denote $\mathcal{H} = \{ h= g^2: g \in \mathcal{G} \}$. Note that we have $|\mathcal{G}| = |\mathcal{H}| \leq K^{S}$. 
\end{definition}

Applying with the Hoeffding bound and note that $0 \leq h \leq 2$, we have by Proposition~\ref{prop: hoeffding-complex-|a|}, 

\begin{lemma}\label{lemma: approx-Phi-x-h-naive}
    For every $h \in \mathcal{H}$ and let $Q \subseteq [N]$ with cardinality $q \geq 2\varepsilon^{-1} \eta^{-1} \tau^{-2} \log(K) \log(1/\gamma) $. Then we have that for every $\eta > 0$, $0 < \gamma < 1$ and $0 < \varepsilon \leq 1/2$
    \begin{align}
        \mathbb{E}_{j \in Q}[h_j] \approx_{\varepsilon, \eta} \mathbb{E}_{j \in [N]}[h_j],
    \end{align}
    with probability at least $1 - 2^{-\tau^{-2} 
            \log(K)\log(1/\gamma) } $. 
    \begin{proof}
        First, we apply Proposition~\ref{prop: hoeffding-complex-|a|} noting that $0 \leq h \leq 2$. Then we apply the union bound where we note that the probability that some $h\in \mathcal{H}$ fails to satisfy the above is at most 
        \begin{align}
            |\mathcal{H}| 2^{-\log(K)/ \gamma} = 2^{S \log K - 2\log(K)\log(1/\gamma) \tau^{-2} } = 2^{-\tau^{-2} 
            \log(K)\log(1/\gamma)},
        \end{align}
        which concludes the result. 
    \end{proof}
\end{lemma}

\begin{remark}
    One could convert the above lemma into a more familiar $1- \delta$ argument. In particular, this task can be achieved by setting 
    \begin{align}
        q \geq 2 \epsilon^{-1} \eta^{-1} \max \{ \tau^{-2} \log(K) \log(1/\gamma), \log(1/\delta)\},
    \end{align}
    for some $\delta > 0$. 
\end{remark}

\begin{lemma}\label{lemma: rip-naive-proof-one-scale-net}
    For every multiset $Q \subseteq [N]$ and every vector $x \in \mathbb{C}^N$ with $\|x\|_1 =1$, 
    \begin{align}
        \mathbb{E}_{j \in Q}[|(\Phi x)_j|^2] \approx_{0, 6\eta} \mathbb{E}_{j \in Q}[h_j],
    \end{align}
    where $3\eta = 2 \tau + \gamma$. 
    \begin{proof}
        Suppose that an index $j \in Q$ is good in the sense above, write
        \begin{align}
           | (\Phi x)_j|^2 - |h_j| = (| (\Phi x)_j|^2 + |g_j|) |(| (\Phi x)_j|^2 - |g_j|)| \leq \tau (| (\Phi x)_j|^2 + |g_j|) \leq (1 + \sqrt{2}) \tau.
        \end{align}
        Note that we have $\Phi$ is a unitary matrix so that $|(\Phi x)_j|^2 \leq \|\Phi \|_\infty \|x\|^2_1 = \|\Phi \|_\infty $ and $|g_j| \leq \sqrt{2}$. Suppose that $j \in Q$ is a bad index, then we only conclude that $| (\Phi x)_j|^2 \approx_{0, 2} |h_j| $ with the appropriate normalization. However, from Lemma~\ref{lemma: Maurey-empirical-method}, it is stated that there is at most $\gamma$ fraction of bad indices. As a result, we have that 
        \begin{align}
        \mathbb{E}_{j \in Q}[|(\Phi x)_j|^2] \approx_{0, (1+\sqrt{2})\tau(1-\gamma) + 2 \gamma } \mathbb{E}_{j \in Q}[h_j].
    \end{align}
    Denote that $(1+\sqrt{2})\tau(1-\gamma) \leq (1+\sqrt{2})\tau \leq 4 \tau $ and we write $3\eta = 2\tau + \gamma$, we conclude with the lemma. 
    \end{proof}
\end{lemma}

Combining with the above lemmas, we arrive at the restatement of the RIP condition for uniformly subsampled $A$ from $\Phi$, where we now restore the BOS constant $\| \Phi \|_\infty$ and $1$-norm normalization $\|x \|_1$.  

\begin{proposition}\label{prop: rip-bos-naive-proof}
    Let $\Phi \in \mathbb{C}^{N \times N}$ a unitary matrix for sufficiently large $N$ and sufficiently small $\epsilon, \eta > 0$. For some $q \geq 2 \epsilon^{-1} \eta^{-3} \log(K) \log(1/\eta)  $ and let $Q$ be a multiset of $q$ uniform and independent random elements of $[N]$. Then it holds for every $x \in \mathbb{C}^N$, 
    \begin{align}
        \mathbb{E}_{j \in Q}[|(\Phi x)_j|^2] \approx_{\epsilon, (13 + 6 \epsilon) \eta \cdot \|x\|^2_1 \cdot \|\Phi\|^2_{\infty}} \mathbb{E}_{j \in [N]}[|(\Phi x)_j|^2],
    \end{align}
    with probability at least $1 - 2^{- \eta^{-2}\log(K) \log(1/\eta)}$
    \begin{proof}
        We first apply Lemma~\ref{lemma: rip-naive-proof-one-scale-net} where we have that 
        \begin{align}
        \mathbb{E}_{j \in Q}[|(\Phi x)_j|^2] \approx_{0, 6\eta \cdot \|\Phi \|^2_\infty \cdot \|x\|^2_1 } \mathbb{E}_{j \in Q}[h_j],
    \end{align}
    where we set $\eta = \tau = \gamma$ and use the scaling rule by Lemma~\ref{lemma: approximation-compositions-multiplication}. Then we can apply Lemma~\ref{lemma: approx-Phi-x-h-naive}, which leads to $\mathbb{E}_{j \in Q}[h_j] \approx_{\epsilon, \eta \cdot \|\Phi \|^2_\infty \cdot \|x\|^2_1} \mathbb{E}_{j \in [N]} [h_j]$. This leads to by the composition rule Lemma~\ref{lemma: approximation-compositions-multiplication}, 
    \begin{align}
        \mathbb{E}_{j \in Q}[|(\Phi x)_j|^2] \approx_{\epsilon, 7\eta \cdot \|\Phi \|^2_\infty \cdot \|x\|^2_1} \mathbb{E}_{j \in [N]} [h_j].
    \end{align}
    Finally, we apply Lemma~\ref{lemma: rip-naive-proof-one-scale-net} again with the composition rule to arrive at 
    \begin{align}
        \mathbb{E}_{j \in Q}[|(\Phi x)_j|^2] \approx_{\epsilon, (13 + 6 \epsilon) \eta \cdot \|\Phi \|^2_\infty \cdot \|x\|^2_1 } \mathbb{E}_{j \in [N]}[|(\Phi x)_j|^2]. 
    \end{align}
    \end{proof}
\end{proposition}

Finally, combining with Lemma~\ref{lemma: rip-restatement}, we arrive a naive bounds on the restricted isometry property of row-subsampled BOS matrices, 

\begin{theorem}[RIP for row-subsampled BOS]\label{thm: rip-bos-naive-proof}
     Let $\Phi \in \mathbb{C}^{N \times N}$ be a unitary (BOS) matrix with $\|\Phi\|_\infty = O(\frac{1}{\sqrt{N}})$. For a sufficient small $\epsilon > 0$ and some $q = O(\epsilon^{-4} K^3 \log(K) \log(K/\epsilon)) $, the following holds. Let $A \in \mathbb{C}^{q \times N}$ be a matrix whose $q$ rows are chosen uniformly and independently from the rows, rescaled by $\sqrt{N/q}$. Then with probability at least $1 - 2^{\Omega(\epsilon^{-2} K^2 \log K \log(K/\epsilon) )}$, we have that $A$ satisfies the restricted isometry property of order $K$ with RIP constant $\epsilon$ on the supports of $\mathcal{K}$ with $K = |\mathcal{K}|$. 
    \begin{proof}
        We can directly apply Proposition~\ref{prop: rip-bos-naive-proof} with $\epsilon/2$ and set $\eta = \Omega(\epsilon/K)$. Then we obtain that 
\begin{align}
        \mathbb{E}_{j \in Q}[|(\Phi x)_j|^2] \approx_{\delta_K/2,  O(\delta_K \cdot \|\Phi \|^2_\infty \cdot \|x\|^2_1/K) } \mathbb{E}_{j \in [N]}[|(\Phi x)_j|^2]. 
    \end{align}
    Note that we have that $\|x\|_1 \leq \|x\|_2 \|1_x\|_2 \leq \|x\|_2 \sqrt{K}$ for $K$-sparse $x$ supported on $\mathcal{K}$. Apply this to Lemma~\ref{lemma: rip-restatement} gives the desired result. 
    \end{proof}
\end{theorem}

\begin{remark}
    One might ask if we could further reduce the $K$ dependence and answer is yes. However, it must be given as strictly more developed proof approach.  Indeed, one might wonder if it is truly necessary to set $\tau  = O(\eta) = \gamma$. In the proof it is necessary to the case as we need to ensure that the kill off normalization incurred in translating from $\|x\|_2$ to $\|x\|_1$. 
\end{remark}

For mathematical consistency, we convert it to $1 -\delta$ style argument. In this case, we slightly modify Proposition~\ref{prop: rip-bos-naive-proof} with $q \geq 2 \epsilon^{-1} \eta^{-1} \max(\eta^{-2} \log(K) \log(1/\eta), \log(1/\delta))$. By rescaling $\epsilon$ and setting $\eta = \epsilon/ K$ (since we are proving for Hadamard matrix, there is no need for $\Omega$ notation), we obtain the following bound, converting $\epsilon$ to the $K+1$th-order restricted isometric constant $\delta_K$

\begin{corollary}
     Let $\Phi \in \mathbb{C}^{N \times N}$ be the Walsh-Hadamard matrix. For a sufficient small $\delta_K > 0$ and some 
     \begin{align}
        q \geq 4 \delta_K^{-2} K   \max(K^2 \delta_K^{-2}\log(K) \log(1/\eta), \log(1/\delta)).
    \end{align}
     The following holds. Let $A \in \mathbb{C}^{q \times N}$ be a matrix whose $q$ rows are chosen uniformly and independently from the rows, rescaled by $\sqrt{N/q}$. Then with probability at least $1 - \delta$, we have that $A$ satisfies the restricted isometry property with $K+1$th-order restricted isometric constant $\delta_K$ on the supports of $\mathcal{K}_{\mathscr{M}} \cup \{ I \}$ with $K = |\mathcal{K}_{\mathscr{M}}|$.
\end{corollary}

\subsubsection*{Refined analysis}

The preceding section we provide a naive proof on the row-subsampled RIP properties of the Fourier/BOS matrix. The bound on the number of samples $q$ is far from optimal. We now follow from~\cite{haviv2015restrictedisometrypropertysubsampled} where we give a refined analysis. The main idea is to define a level sets of $g^{(i)}$ that progressively gives a finer estimation of $\Phi x$. We will adapt the proof strategy presented in~\cite{haviv2015restrictedisometrypropertysubsampled} where we fix a preselected support $\mathcal{K}_{\mathscr{M}}$ and $K = |\mathcal{K}_{\mathscr{M}}|$ and state in a conventional $1- \delta$ fashion. As such, we will only comment and prove the parts where these adaptations apply.

\begin{definition}[Vector level set $\mathcal{G}_i$]
     For every $1 \leq i \leq t+r$, let $\mathcal{G}_i$ denote the set of all vectors $g^{(i)} \in \mathbb{C}^N$ that can be represented as

\begin{align}
    g^{(i)}=\frac{\sqrt{2}}{|F|} \cdot \sum_{(\ell, s) \in F}(-1)^{s / 2} \cdot M^{(\ell)}, 
\end{align}
for a multiset $F$ of $O\left(2^i \cdot \log (1 / \gamma)\right)$ pairs in $[N] \times\{0,1,2,3\}$. Furthermore, for a given subset $\mathcal{K}_{\mathscr{M}} \subseteq [N]$, we denote the vector level sets $\mathcal{G}^{\mathcal{K}_{\mathscr{M}}}_i$ whose pairs (columns) are drawn from $\mathcal{K}_{\mathscr{M}}$. 
\end{definition}

\begin{remark}
    Note that for every $1 \leq i \leq t+r,\left|\mathcal{G}_i\right| \leq N^{O\left(2^i \cdot \log (1 / \gamma)\right)}$ and $\left|\mathcal{G}^{\mathcal{K}_{\mathscr{M}}}_i\right| \leq |\mathcal{K}_{\mathscr{M}}|^{O\left(2^i \cdot \log (1 / \gamma)\right)}$.
\end{remark}

\begin{definition}
    For $x \in \mathbb{C}^N$ such that $\|x \|_1 =1$ and whose supported contained in $\mathcal{K}_{\mathscr{M}}$. We say that in index $j \in [N]$ is \emph{good} if, for every $1 \leq i \leq t$, there exists a vector $g^{(i+r)} \in \mathcal{G}^{\mathcal{K}_{\mathscr{M}}}_{i+r}$ such that
    \begin{align}
        |(\Phi x)_j| \approx_{0, 2^{-(i+r)/2}}|g^{(i+r)}_j|.
    \end{align}
    Otherwise, we say such an index $j$ \emph{bad}. 
\end{definition}
Similar to Lemma~\ref{lemma: Maurey-empirical-method}, we give a slightly modified version of Maurey's empirical method. 

\begin{corollary}[Maurey's empirical method]
      Suppose that $M = O(2^{i} \log1/\gamma)$ for $ \gamma > 0$. Denote any set $Q \subseteq [N]$ and every $\mathcal{G}_i$ for $1 \leq i \leq t+r$. For $j \in Q$ we denote the set $Q'$ consists of all indices such that $|g_j| \approx_{0, 2^{-i/2}}  |\mathbb{E}_{(\ell, s) \in \mathcal{D}}[\sqrt{2} \cdot (-1)^{s/2} \cdot \Phi^{(\ell)}_j]|$. Then we have that with probability at least $3/4$
    \begin{align}
        \frac{|Q'|}{|Q|} = \Omega(1 -\gamma).
    \end{align}
\end{corollary}

The refinement takes place in establishing a refiner range of approximation, where at each level, we approximate to a given precision. We first give the following definitions. 

\begin{definition}
     For a $(t+r)$-tuple of vectors $(g^{(1)}, \cdots, g^{(t+r)}) \in \mathcal{G}^{\mathcal{K}_{\mathscr{M}}}_1 \times \cdots \times \mathcal{G}^{\mathcal{K}_{\mathscr{M}}}_{t+r}$ and for $1\leq i \leq t$, let $C_i $ be the set of all $j \in [N]$ for which $i$ is the smallest index such that $|g^{(i)}_j| \geq 2 \cdot 2^{-i/2}$. For $m=i, \cdots, i +r$ define the vector $h^{(i, m)}$ by 
    \begin{align}
        h^{(i, m)}_j = |g^{(m)}_j|^2 \cdot \mathrm{1}_{j \in C_i}. 
    \end{align}
    
\end{definition}

\begin{definition}[The vector sets $\mathcal{D}_{(i, m)}$]
     For every $i \leq m \leq i+r$, let $\Delta^{(i, m)}$ be the vector defined by

\begin{align}
    \Delta_j^{(i, m)}= \begin{cases}h_j^{(i, m)}-h_j^{(i, m-1)}, & \text { if }\left|h_j^{(i, m)}-h_j^{(i, m-1)}\right| \leq 30 \cdot 2^{-(i+m) / 2} \\ 0, & \text { otherwise. }\end{cases}
\end{align}
Note that for $j \notin C_i$, $\Delta^{(i, m)} =0$. Finally we denote $\mathcal{D}_{i, m}$ be the set of all vectors $\Delta^{(i, m)}$ that can be obtained in this way. 
\end{definition}

Apply the usual Hoeffding and union bound, we can solve the following. 

\begin{lemma}
    For every $\tilde{\epsilon}, \tilde{\eta} > 0$ and some $q = O(\tilde{\epsilon}^{-1} \tilde{\eta}^{-1} \max( \log N \cdot \log(1/\gamma)), \log(1/\delta))$, let $Q$ be a multiset of $q$ uniform and independent random elements of $[N]$. Then, it holds that for every $1\leq i \leq t, m$ and a vector $\Delta^{(i, m)} \in \mathcal{D}_{i, m}$ associated with a set $C_i$, 
    \begin{align}
        \mathbb{E}_{j \in Q}[\Delta^{(i, m)}_j ] \approx_{0, b} \mathbb{E}_{j \in [N]}[\Delta^{(i, m)}_j ].
    \end{align}
    with probability $1-\delta$
    for $b = O(\tilde{\epsilon} \cdot 2^{-i} \cdot \frac{|C_i|}{N} + \tilde{\eta})$.
\end{lemma}

\begin{lemma}
    For every $j \in [N]$ that is good and let $t = \log_2(1/\eta)$ and $r = \log_2(1/\epsilon^2)$, 
    \begin{align}
        |(\Phi x)_j|^2 \approx_{O(\epsilon), O(\eta)} \sum^t_{i=1} h^{(i, i+r)}_j; \quad   |(\Phi x)_j|^2 \approx_{O(\epsilon), O(\eta)} \sum^t_{i=1} \sum^{i+r}_{m=i} \Delta^{(i, m)}_j.
    \end{align}
       
\end{lemma}

Next, we relate the approximation of $\Phi x$ to telescope sum of $h^{(i, m)}_j$. First, we show that we can estimate $\|\Phi x\|^2_2$ well within nets $C_i$ for $1\leq i \leq t$, combining with the fact that most indices are good.  

\begin{lemma}
    Suppose that $x \in \mathbb{C}^N$ with $\|x\|_1=1$. For every multiset $Q \subseteq [N]$, there exist vector collections $(\Delta^{(i, m)} \in \mathcal{D}_{i, m})_{m=i, \cdots, i+r}$ associated with the sets $C_i$ for $1\leq i \leq t$. Then we have that 
\begin{enumerate}
    \item $\mathbb{E}_{j \in [N]}[|(\Phi x)_j|^2] \geq \sum^t_{i=1} 2^{-i} \cdot \frac{C_i}{N} - \eta $.

    \item $\mathbb{E}_{j \in [N]}[|(\Phi x)_j|^2]  \approx_{O(\epsilon), O(\eta)} \mathbb{E}_{j \in Q}[\sum^t_{i=1} \sum^{i+r}_{m=i} \Delta^{(i, m)}_j]$.

\end{enumerate}

\end{lemma}

Combining the above lemmas, we can conclude the following in a similar spirit of the naive analysis. Note that we set the BOS constant back. 

\begin{proposition}\label{prop: rip-bos-rowsampled-refined}
    Let $\Phi \in \mathbb{C}^{N \times N}$ a unitary matrix for sufficiently large $N$ and sufficiently small $\epsilon, \eta > 0$. For some $q =O(  \epsilon^{-1} \eta^{-1} \max(\log(K) \log^2(1/\epsilon)\log(1/\eta), \log(1/\delta))$ and let $Q$ be a multiset of $q$ uniform and independent random elements of $[N]$. Then it holds for every $x \in \mathbb{C}^N$ supported on $\mathcal{K}_{\mathscr{M}}$
    \begin{align}
        \mathbb{E}_{j \in Q}[|(\Phi x)_j|^2] \approx_{\epsilon,  \eta \cdot \|x\|^2_1 \cdot \|\Phi\|^2_{\infty}} \mathbb{E}_{j \in [N]}[|(\Phi x)_j|^2],
    \end{align}
    with probability at least $1-\delta$. 
  
\end{proposition}

 Utilizing lemma~\ref{lemma: rip-restatement} where we can set $\eta = \Omega(\epsilon/K)$,

 \begin{theorem}[RIP for row-subsampled BOS]\label{thm: rip-bos-rowsampled-refined}
     Let $\Phi \in \mathbb{C}^{N \times N}$ be a unitary (BOS) matrix with BOS constant $\|\Phi\|_\infty = O(1/\sqrt{N})$. For a sufficiently small $\epsilon > 0$ and some $q$,
    \begin{align}
        q = O(\epsilon^{-2} K \max( \log^2(K/\epsilon) \log^2(1/\epsilon) \log(K)), \log(1/\delta)).
    \end{align}
    The following holds. Let $A \in \mathbb{C}^{q \times N}$ be a matrix whose $q$ rows are chosen uniformly and independently from the rows, rescaled by $\sqrt{N/q}$. Then with probability at least $1 - 2^{\Omega( \log K \log(K/\epsilon) )}$, we have that $A$ satisfies the $K+1$th-order restricted isometry property with RIP constant $\epsilon$ on the supports of $\mathcal{K}_{\mathscr{M}} \cup \{ I \} $ with $K = |\mathcal{K}_{\mathscr{M}}|$. 
   
\end{theorem}

\section{Learning logical information from syndrome data}\label{sec: learning-logical-information-from-syndrome}

Here, we give a self-consistent treatment of learning logical information from the syndrome distributions. The learnability result presented in Section~\ref{sec: learnability-general} is based on the classification of the symmetry of the system, and we now explain how these can be applied to various practical settings. For simplicity of discussion, we will use $\mathcal{A}$ denote for either the Pauli group or the spacetime Pauli group. We will subsume the conventional notations of Pauli channels $\Lambda$ with Pauli eigenvalues $\lambda_a$, $a \in \mathcal{A}$ throughout the rest of appendix.

\subsection{Learning with data-syndrome and subsystem codes}\label{sec: ds-sub-system code}

In this section, we review the learning scheme discussed in \cite{Wagner_2022, Wagner_2023}, focusing on some standard generalizations of the quantum stabilizer codes. It serves as a route towards considering learning logical information from a spacetime quantum error-correcting circuit. In many applications, the logically equivalent errors are defined with respect to some \emph{gauge} transform, which can be naturally described in the subsystem code. 

\begin{definition}[Subsystem code]
     A quantum stabilizer subsystem code consists of  $\mathcal{G} \subset \mathrm{P}^{n}$, referred to as the \emph{gauge group} and its center $\mathcal{Z}(\mathcal{G}) = \mathcal{S} = \mathcal{G}^\perp \cap \mathcal{G}$ is the stabilizer. The \emph{dressed} logical operators are defined by $\mathcal{S}^\perp / \mathcal{\mathcal{G}}$ and the \emph{bare} logical operators $\mathcal{G}^\perp / \mathcal{S}$. The distance of the subsystem code is defined by the minimal weight of groups in $\mathcal{S}^\perp / \mathcal{G}$. 
\end{definition}

The subsystem code is a $[[n, k, g, d]]$ code with additional parameter $g$ denotes the dimension of the gauge space. The gauge space is given by $\mathcal{G} / \mathcal{S}$. We now provide parameter counting. Let's denote $\mathcal{S} =  \langle Z_1, \cdots, Z_r \rangle $, and 
\begin{align}
    \mathcal{S}^\perp = \langle i, Z_1, \cdots, Z_n, X_{r+1}, \cdots X_{n} \rangle. 
\end{align}
The gauge group $\mathcal{G}$ can be chosen from $\langle i, Z_{1}, \cdots, Z_{r}, \cdots, Z_{r+g}, X_{r+1}, \cdots, X_{r+g} \rangle $. Then we have that the dress logical operators $\mathcal{G}^\perp / S \cong  \langle i, Z_{r+g+1}, \cdots, Z_{n}, X_{r+g+1}, \cdots X_{n} \rangle \cong \mathcal{S}^\perp / \mathcal{G} $. Hence the logical dimension of the dressed and bare logical operators is the same. Denote the \emph{bare} distance $d_{\operatorname{bare}}$ the minimal weight of $\mathcal{G}^\perp / \mathcal{S}$ and the \emph{dressed} distance $d_{\operatorname{dress}}$ by the minimal weight of $\mathcal{S}^\perp / \mathcal{G}$. It is easy to see that $d_{\operatorname{bare}} \geq d_{dress} = d$. In this case, we define $\mathcal{A} = \mathbb{F}^{2n}_2$ and associate the bilinear form $\beta$ given by the standard symplectic inner product. Let the quantum subsystem stabilizer code with gauge group $\mathcal{G}$ and $\mathcal{S} = \mathcal{G}^\perp \cap \mathcal{G}$ with $|\mathcal{S}| = r$. Let $\sigma: \mathcal{A} \cong \mathbb{F}^{2n}_2 \rightarrow \mathbb{F}^r_2$ denote the syndrome map. A simple adaptation of Theorem~\ref{thm: learn-all-distribution-from-syndrome} indicates that, 

\begin{corollary}
     The Pauli channel $\Lambda \in \mathbb{R}[\mathbb{F}^n_2]$ such that all its Pauli eigenvalues are greater than $0$ is learnable from the measurement of the syndrome if and only if for any $a \in \mathcal{K}$ and $a' \in \mathcal{K}$, $\sigma(a) \neq \sigma(a')$. 
\end{corollary}

The effective distribution gives the logically equivalent error and in this case can be written as 

\begin{align}\label{eq: eff-distribution-subsystem-code}
    p^{\operatorname{eff}}_{a} = \frac{1}{|\mathcal{\mathcal{G}}|} \sum_{g \in \mathcal{G}} p_{ag}; \quad \lambda^{\operatorname{eff}}_{a} = \begin{cases}
      \lambda_{a}  & \text{if  } a \in \mathcal{\mathcal{G}^\perp} \\
      0 & \text{otherwise}.
    \end{cases}
\end{align}
The gauge transformation defines the following equivalence relation $a \sim_{\mathcal{G}} a'$ if there exists some gauge operator $g \in \mathcal{G}$ such that $a = a'g$. 

\begin{corollary}[Logical learnability for subsystem code]
     Given the Pauli channel $\Lambda \in \mathbb{R}[\mathbb{F}^{2n}_2]$ such that all its Pauli eigenvalues are greater than $0$, then the effective distribution Eq.~\ref{eq: eff-distribution-subsystem-code} is learnable if and only if for any $a \in \mathcal{K}$ and $a' \in \mathcal{K}$, $\sigma(a) = \sigma(a') \iff a \sim_{\mathcal{G}} a' $.
    \begin{proof}
        We can define $\mathcal{S} = \mathscr{M}$, $\mathcal{G} = \mathscr{G}$, and Boolean group $\mathcal{A} = \mathbb{F}^{2n}_2$ given from the bilinear form as the symplectic inner product, which is symmetric and nondegenerate. Hence, the result follows from Theorem~\ref{thm: learn-eff-distribution-from-syndrome}. 
    \end{proof}
\end{corollary}

The above discussion assumes the perfect measurement. However, in real settings, the measurement can be faulty. Denote the number of data qubits by $n$ and the number of ancillas by $m$. We denote $N = 2n + m$ and the group $\mathcal{A} = \mathbb{F}^N_2$ and we write elements $e \in \mathbb{F}^N_2$ by a tuple $(e_d, e_a)$ and the following bilinear form $\beta(e_1, e_2) = \langle (e_1)_d, \overline{(e_2)_d} \rangle + \langle (e_1)_a, (e_2)_a \rangle_{\mathbb{F}_2} $ over the $\mathbb{F}_2$ arithmetic, where $e_d$ denotes the symplectic representation of Pauli operators on data qubits with $\overline{e_d}$ denote its conjugate (swap $X$s and $Z$s). A class of quantum error correcting codes that takes into account both information is called the the \emph{quantum data-syndrome} code~\cite{ashikhmin2020quantum}. 

\begin{definition}[Data-syndrome codes]
     Let's associate a $[m, r, d_m]$ classical linear code $\mathcal{C}$ with a canonical form 
   \begin{align}
      G=  \begin{pmatrix}
           I_{r}, \: J
       \end{pmatrix},
   \end{align}
   and a quantum stabilizer code $\mathcal{Q} = (\mathcal{S},  \mathcal{L} = \mathcal{S}^\perp)$ with $[[n, k, d]]$. Then the data-syndrome code is defined via 
   \begin{align}
      \widetilde{H} =  \begin{pmatrix}
           H & I_r & 0 \\
           J^TH & 0 & I_{m-r}
       \end{pmatrix},
   \end{align}
   where $H \in \mathbb{F}^{r \times 2n}_2$ the parity check matrix of the quantum stabilizer code. 
\end{definition}

Data-syndrome code generalizes the notion of repeating measuring syndromes and decodes via a majority vote and a natural notion of degeneracy.   For any error $e= (e_d, e_a) \in \mathbb{F}^N_2$, we denote the generalized syndrome map $\widetilde{\sigma}(e) = \widetilde{H}(e)$. 

\begin{proposition}\label{prop: ds-code-degeneracy}
    For any error $e = (e_d, e_a) \in \mathbb{F}^N_2$ there always exists $c \in \mathbb{F}^m_2$ such that $\widetilde{\sigma}(c + e_a) = \widetilde{\sigma}(e)$ and $c \in \mathcal{C}$. Furthermore, if $\widetilde{\sigma}(e_1) = \widetilde{\sigma}(e_2)$, then $e_1$ and $e_2$ are related by, 
    \begin{enumerate}
        \item \label{data-syndrome-degeneracy-condition-1} $(e_1)_a = (e_2)_a$ and $(e_1)_d = (e_2)_d + l $ for $l \in \mathcal{L}$. 
        \item \label{data-syndrome-degeneracy-condition-2} $(e_1)_a = (e_2)_a +c $ for some $c \in \mathcal{C}$ and $\widetilde{\sigma}((e_1)_d ) = \widetilde{\sigma}((e_2)_d ) + \widetilde{\sigma}(c ) $ .
    \end{enumerate}
    \begin{proof}
        To see this first part, let $\sigma(e_d) = H e_d$. It suffices to show that the classical codeword associated with $c = \sum_{i: \sigma(e_d)_i =1} g_i $ and we have that $\widetilde{\sigma}(e_d) = \widetilde{\sigma}(c)$, which implies the result. It is easy to observe that if $e_1$ and $e_2$ are related by the condition~\ref{data-syndrome-degeneracy-condition-1} or condition~\ref{data-syndrome-degeneracy-condition-2}, their syndrome must match. We now show that this is only the case. Suppose that $e, e' \in \mathbb{F}^N_2$ such that $\widetilde{\sigma}(e) = \widetilde{\sigma}(e')$. Then it must follow that $e = e' + l$ for some $l \in \ker \widetilde{H}$. Note that we have that 
        \begin{align}
            \ker \widetilde{H} = \ker \widetilde{H}' = \ker \begin{pmatrix}
           H & I_r & 0 \\
           0 & J^T & I_{m-r}
       \end{pmatrix} = \mathcal{L} \times \{0\} \cup \{ (e_d, c): c \in \mathcal{C}, c \neq 0,  \: \widetilde{\sigma}(c) = \widetilde{\sigma}(e_d) \}.
        \end{align}
        Note that the two conditions precisely correspond to these subspaces, which concludes the proof. 
    \end{proof}
\end{proposition}

The condition~\ref{data-syndrome-degeneracy-condition-2} for the above proposition~\ref{prop: ds-code-degeneracy} provides additional degeneracy than in the typical stabilizer codes, and this condition might reflect upon the error channel, $\Lambda \in \mathbb{R}[\mathbb{F}^N_2]$, again in the character basis Eq.~\eqref{eq: pauli-character-fourier}
\begin{align}
    \Lambda = \prod_{e \in \mathcal{K}}((1-q_e) + q_e \chi_e),
\end{align}
where we define the character with respect to the bilinear form $\chi_e = (-1)^{\beta(e, \cdot)}$. To apply the framework of the physical learnability, we need to ensure that syndromes uniquely correspond to errors in $\mathcal{K}$.

\begin{theorem}
    The Pauli channel $\Lambda \in \mathbb{R}[\mathbb{F}^N_2]$ such that all its Pauli eigenvalues are greater than $0$ is learnable from the measurement of the syndrome if and only if for any $e = (e_d, e_a) \in \mathcal{K}$ and $e'=(e'_d, e'_a) \in \mathcal{K}$, $\widetilde{\sigma}(e) \neq \widetilde{\sigma}(e') $. 
    \begin{proof}
        The proof follows from minimal adaptation from Theorem~\ref{thm: learn-all-distribution-from-syndrome} by setting $\mathcal{A} = \mathbb{F}^N_2$ with appropriate bilinear-form $\beta$, which is non-degenerate and symmetric. 
    \end{proof}
\end{theorem}

For quantum error-correcting applications, we are primarily concerned with learning the distribution of logical equivalence, which is still given by the effective distribution modulo the quantum stabilizers, 
\begin{align}\label{eq: eff-distribution-ds-code}
    p^{\operatorname{eff}}((e_d, e_a)) = \frac{1}{|\mathcal{S}|} \sum_{s \in \mathcal{S}} p((e_d + s, e_a)); \quad \lambda^{\operatorname{eff}}_{(e_d, e_a)} = \begin{cases}
      \lambda_{(e_d, e_a)}  & \text{if  } e_d \in \mathcal{L} \\
      0 & \text{otherwise}.
    \end{cases}
\end{align}
This gives the error rate of stabilizer-equivalent data error conditioned on the faulty-measurement outcomes given by $e_a$. Let $\widetilde{S} = \mathcal{S} \oplus \bvec{0}$ the trivial extension of the stabilizer group on the data and ancillary qubits so that $\widetilde{\mathcal{S}}^\perp = \mathcal{L} \oplus \mathbb{F}^m_2$, and the logical learnability suggests a minimal extension to the case of the stabilizer code, 

\begin{theorem}[Logical learnability for DS code]
     Given the Pauli channel $\Lambda \in \mathbb{R}[\mathbb{F}^N_2]$ such that all its Pauli eigenvalues are greater than $0$, then the effective distribution Eq.~\eqref{eq: eff-distribution-ds-code} is learnable if and only if for any $e = (e_d, e_a) \in \mathcal{K}$ and $e'=(e'_d, e'_a) \in \mathcal{K}$, $\widetilde{\sigma}(e) = \widetilde{\sigma}(e') \iff e_a = e'_a $ and $e_d = e'_d + s$ for $s \in \mathcal{S}$. 
\end{theorem}

Note that the data-syndrome code distance is given by $d = \min(d_n, d_m)$. Hence a natural sufficient condition for the parameterized Pauli channel, the Pauli error $e = (e_d, e_a)$ with $|e_a| \leq \lfloor d_n/2 \rfloor $ and $|e_d| \leq \lfloor d_m/2 \rfloor$. 


\subsection{Code-to-circuit isomorphism}\label{subsec: code-circuit-iso}
Prior to establishing a general learnability framework that incorporates the circuit-level faults, we construct a code-to-circuit isomorphism, which could be of independent interest. The circuit model we are interested in consists of Clifford unitaries and Pauli measurements that implements a base quantum (subsystem) code. 

\begin{definition}\label{def: circuit-model-from-subsystem-code}
    Let a quantum subsystem code be given with gauge group $\mathcal{G}$ and stabilizer subgroup $\mathcal{S}$ defined over $n$ qubits. Let the circuit that implements the quantum subsystem code, denoted $C_T(\mathcal{G})$, consists of
    \begin{enumerate}
        \item For each integer time $t=0, 1, \cdots, T$, we place $n$ data qubits, and $T$ is assumed to be an even integer. 
        \item For integers $t=0, \cdots, T$, we denote the \emph{measurement sequence} in implementing some (subsystem) code by subgroups $\mathcal{M}_{0.5}, \cdots, \mathcal{M}_{T+0.5}$, where each subgroup is generated by some checks we measure and $\mathcal{M}_{T+0.5} = \mathcal{S}$. Furthermore, all the measured checks altogether generate the gauge group
        \begin{align}
             \langle \mathcal{M}_{0.5}, \cdots, \mathcal{M}_{T+0.5} \rangle = \mathcal{G}.
        \end{align}
        \item Logical Clifford unitaries are compiled with $u = u_{(T, T-1)} \cdots u_{(1, 0)}$ such that they preserve the $\mathcal{G}$ as well as $\mathcal{S}$. Hence, for $t=0, \cdots, T-1$, 
        \begin{align}
            u_{(T, t)}(m_t) \in \mathcal{G}, \quad \forall m_t \in \mathcal{M}_{t+0.5}. 
        \end{align}
        
        \item The Pauli measurement gadgets and logical Clifford unitaries are implemented in half-integer steps such that they share disjoint supports, i.e.
        \begin{align}
            u_{(t+1, t)}(m_t) = m_t, \forall m_t \in \mathcal{M}_{t+0.5}, 
        \end{align}
        for $t=0, \cdots, T-1$. 
    \end{enumerate}
\end{definition}

In what follows and in the main text, we focus on a simpler model, the \emph{syndrome extraction circuits}, which implements a stabilizer group $\mathcal{S}$ (with $\mathcal{G} = \mathcal{S}$). 

\begin{definition}[Syndrome extraction Clifford circuits]\label{def: syndrome-extraction-circuit}
     Let $\mathcal{S}$ be a stabilizer group over $n$ qubits. The syndrome extraction circuit $\mathcal{C}_{T}(\mathcal{S})$ consists of 
    \begin{enumerate}
        \item For each integer time $t=0, 1, \cdots, T$, we place $n$ data qubits. 
        
        \item Logical Clifford unitaries compiled with $u = u_{(T, T-1)}u_{(T-1, T-2)} \cdots  u_{(1, 0)}$ that preserve $\mathcal{S}$. Hence, for $t=0, \cdots, T-1$, 
        \begin{align}
            u_{(T, t)}(m_t) \in \mathcal{S}, \quad \forall m_t \in \mathcal{M}_{t+0.5}. 
        \end{align}
        
        \item For $t=0, 1, \cdots, T-1$, the stabilizer measurements $\mathcal{M}_{t+0.5} \subseteq \mathcal{S}$ generated by independent stabilizer generators with disjoint supports. While the terminal stabilizer measurement $\mathcal{M}_{T+0.5}$ implements all stabilizer measurements $\mathcal{M}_{T+0.5} = \mathcal{S}$. 
    \end{enumerate}
\end{definition}

Note that in all cases we have $\mathcal{M}_{T+0.5} = \mathcal{S}$, which violates the disjoint property. In what follows, we show that it is due to the assumption that the terminal syndrome measurements are perfect. Furthermore, in many practical cases, there is no need to include $\mathcal{M}_{T+0.5}$ in the measurement sequence (that is, $\mathcal{M}_{T+0.5} = I$), if the previous measurements already generate the stabilizer group $\mathcal{S}$. This can be seen in various examples below. In what follows, we detail how to map $C_T(\mathcal{G})$ ($C_T(\mathcal{S})$) into subsystem code, establishing the circuit-to-code isomorphism. 

\subsubsection*{Pauli frame invariance}

Let's describe the basic setting and a fundamental invariance.

\begin{definition}[Spacetime Pauli group]
     A spacetime Pauli group for a Clifford circuit is given by $\mathcal{A} : = \mathcal{A} \cong  \otimes^T_{i=0} \mathrm{P}^{n}$, where $a \in \mathcal{A}$ we can represent it by $a = \eta_T(a) \eta_{T-1}(a) \cdots \eta_0(a)$, modulo the phase. 
\end{definition}
For contextual consistency, we simply denote the spacetime Pauli group to be $\mathcal{A}$, where we unify the learnability of circuits and codes on the same footing, discussed in Section~\ref{sec: learnability-general}. A key notion we will introduce is a (Pauli) frame transformation. Let $a, b \in \mathcal{A}$, we denote the inner product (symmetric, nondegenerate bilinear form) $\beta(\cdot, \cdot) \equiv \langle \cdot, \cdot \rangle_{\mathcal{A}}$, inherited from the inner product of the base Pauli group $\mathrm{P}^n$
\begin{align}
    \langle a, b \rangle_{\mathcal{A}} := \sum^T_{i=0} \langle \eta_i(a), \eta_i(b) \rangle. 
\end{align}
In what follows, we make an abuse of notation by dropping the subscript $\mathcal{A}$ when referring to the spacetime inner product $\langle \cdot, \cdot \rangle = \langle \cdot \cdot \rangle_{\mathcal{A}}$. 

\begin{definition}[Forward and backward propagation]
     For any element $a \in \mathcal{A}$ and a Clifford circuit. Denote ${u}_{(t, i)} = {u}_{(t, t-1)} \circ \cdots \circ u_{(i+1, i)}$ for $t \geq i +1$, given as the Clifford gate operations, otherwise, it is an identity. Then the forward propagation is defined by 
    \begin{align}
        \eta_{t}(\overrightarrow{a}) = \prod^{t}_{i=0} u_{(t, i)}(\eta_{i}(a)).
    \end{align}
    The backward propagation is defined by 
    \begin{align}
       \eta_{t}( \overleftarrow{a}) = \prod^{T}_{i=t} u^{-1}_{(i, T)} (\eta_{i}(a)).
    \end{align}
\end{definition}
The forward and backward propagation serve as a circuit automorphism. Whenever the context is clear, we always write the time $\eta_t(\overleftarrow{a}) = \overleftarrow{a}(t)$ for simplicity.

\begin{proposition}[Corollary 1 and Proposition 3 in Ref.~\cite{delfosse2023spacetimecodescliffordcircuits}]\label{prop: spacetime-cumulant-automorphism}
    The propagation operators have the following properties, 
    \begin{itemize}[label=--]
        \item The propagation operations are automorphism within $\mathcal{A}$ which satisfies
        \begin{align}\label{eq: cumulant-automorphism-spacetime}
            \overleftarrow{ab} = \overleftarrow{a}\overleftarrow{b}; \quad  \overrightarrow{ab} = \overrightarrow{a}\overrightarrow{b}, 
        \end{align}
        for  $a, b \in \mathcal{A}$. 

        \item The forward and backward propagation operators are related by, 
        \begin{align}
            \langle \overleftarrow{a}, b \rangle  = \langle a, \overrightarrow{b}\rangle. 
        \end{align}
    \end{itemize}
    \begin{proof}
    We restate the proof presented in Ref.~\cite{delfosse2023spacetimecodescliffordcircuits} for completeness. We first show that the propagation operators are automorphisms.
       
        \begin{align}
            \begin{aligned}
                \overleftarrow{a}(t)\overleftarrow{b}(t) &= \prod^t_{i=0} u_{(t, i)}(a(i))\prod^t_{j=0}u_{(t, j)}(b(j)) \\
                &= \pm \prod^t_{i=0} u_{(t, i)} (a(i)b(i)),
            \end{aligned}
        \end{align}
        where we always quotient the phase when we define the Pauli group $\mathcal{A}$. Then we can treat the forward- and backward- propagation operators as linear maps of $\mathcal{A}$ over the symplectic representation. In this case, the forward- and backward- propagation operators are represented in lower (upper) triangular matrix respectively with identity in diagonal, which is invertible over $\mathbb{F}_2$ arithmetic. For the second part, 
        \begin{align}
            \begin{aligned}
                \langle \overleftarrow{a}, b \rangle &= \langle \prod^{T}_{t=0}\prod^t_{i=0} u_{(t, i)}(a(i), b \rangle = \sum^T_{t=0} \sum^t_{i=0} \langle u_{(t, i)}(a(i)), b(t) \rangle \\
                &= \sum^T_{t=0} \sum^t_{i=0} \langle u^{-1}_{(t, i)}\circ u_{(t, i)}(a(i)), u^{-1}_{(t, i)}(b(t)) \rangle \\
                &=  \sum^T_{t=0} \sum^T_{t=i}\langle u^{-1}_{(t, i)}\circ u_{(t, i)}(a(i)), u^{-1}_{(t, i)}(b(t)) \rangle \\
                &= \sum^T_{i=0} \langle u^{-1}_{(t, i)}\circ u_{(t, i)}(a(i)), \prod^T_{t=i}u^{-1}_{(t, i)}(b(t)) \rangle \\
                &= \sum^T_{i=0} \langle u^{-1}_{(t, i)}\circ u_{(t, i)}(a(i)), \overrightarrow{b}(i) \rangle \\
                &= \langle a,  \prod^{T}_{i=0}\overrightarrow{b}(i) \rangle = \langle a, \overrightarrow{b} \rangle,
            \end{aligned}
        \end{align}
        as desired. 
    \end{proof}
\end{proposition}

One can view the forward and backward propagation as follows. In the Pauli laboratory frame, a circuit fault $b$ can be detected by $s^{(t)} \in \mathcal{M}_{t+0.5}$ if $\langle \eta_t(s^{(t)}), \overrightarrow{b} \rangle = 1 $. This is equivalent to boost to the fault rest frame for which we have that  $\langle \eta_t(s^{(t)}), \overrightarrow{b} \rangle = 1 = \langle \overleftarrow{\eta_t(s^{(t)})}, b \rangle $. This invariance is analogous to the $4$-momentum conservation and the frame transformation to the Lorentz transformation for quantum field theories. 

\subsubsection*{Code-to-circuit isomorphism}

 Let $(\mathcal{G}, \mathcal{S})$ be the base subsystem code. We denote the base bare logical subspace by $\mathcal{L}_{\operatorname{bare}}$. For any time $t+1, t$, for $T -1 \geq t \geq 0$ we denote the Pauli-transport of $a \in \mathrm{P}^{n}$
\begin{align}\label{eq: time-transport-Pauli}
    g(a, t) = \eta_{t+1}(u_{(t+1, t)}(a)) \eta_t(a).
\end{align}
 Denote by $C(\mathcal{M}_{t+0.5})$ the set of all elements in $\mathrm{P}^{n}$ that commute with elements in $\mathcal{M}_{t+0.5}$. A useful relation can be given within the $n$-qubit Pauli group $\mathrm{P}^n$. 

\begin{lemma}[Lemma 12 in Ref.~\cite{bacon2017sparse}]\label{lemma: double-commutant-phase}
    Let $A, B \subseteq \mathrm{P}^{n}$ be some subgroup. Then we have that 
    \begin{enumerate}
        \item $C(A B ) = C(A) \cap C(B)$. 
        \item $C(C(A)) \supset A$ and $C(C(A)). = \langle A, \pm I, \pm i  \rangle$. That is the double commutant of $A$ equals to the subgroup generated by $A$ and phases. 
    \end{enumerate}
\end{lemma}

 We define the gauge group as follows, 

 \begin{definition}[Spacetime gauge group]\label{def: spacetime-gauge-group}
     Let the circuit $C_T(\mathcal{G})$ that implements the subsystem code $(\mathcal{G}, \mathcal{S})$ given in Definition~\ref{def: circuit-model-from-subsystem-code}. Define for $t=0, \cdots, T-1$,
     \begin{align}
         \mathcal{G}_{(t+1, t)} = \{ g(a, t), \forall a \in C(\mathcal{M}_{t+0.5})   \}.
     \end{align}
     Then the \emph{spacetime gauge group} is defined by, 
     \begin{align}\label{eq: spacetime-gauge-group}
         \mathscr{G} := \langle \prod^{T-1}_{t=0} \mathcal{G}_{(t+1, t)} \cdot \prod^T_{t=0} \eta_t(\mathcal{M}_{t+0.5}) \rangle.
     \end{align}
 \end{definition}

 Suppose that $p \in C(\mathcal{G})$, then its backward propagation $\overleftarrow{\eta_T(p)} \in \mathscr{G}^\perp$, since the backward propagation by definition commutes with all Pauli transports and the measurements are strictly from $\mathcal{G}$ at each time slice.
 We first show the following technical lemma, 

 \begin{lemma}\label{lemma: stabilizer-deform-backpropagation}
     Let $x \in \mathscr{G}^\perp$ and write $x = \eta_T(a_T) \cdots \eta_0(a_0)$. Then $x \sim_{\mathscr{M}} \overleftarrow{\eta_T(a_T)}$. 
     \begin{proof}
         Since $\mathscr{G}$ commutes with $x$, then for each $a_t$ for $t=0, \cdots, T-1$, we have that $a_t \in C(\mathcal{M}_{t+0.5})$. Furthermore, we have that $\langle x, g(b, t) \rangle = \langle a_t, b\rangle + \langle a_{t+1}, u_{(t+1, t)}(b) \rangle = \langle a_{t+1}  u^\dagger_{(t+1, t)}(a_t), b \rangle = 0$ for $b \in C(\mathcal{M}_{t+0.5})$. This implies that $a_{t}  u^\dagger_{(t+1, t)}(a_{t+1}) \in C(C(\mathcal{M}_{t+0.5}))$. Since we ignore the phase, by Lemma~\ref{lemma: double-commutant-phase}, we have that 
        \begin{align}
           s_t \equiv  a_{t}  u^\dagger_{(t+1, t)}(a_{t+1}) \in \mathcal{M}_{t+0.5}.
        \end{align}
        In this case, consider two adjacent time slices, $\eta_{t+1}(a_{t+1}) \eta_t(a_t) = \eta_t(s_t) g(u^\dagger_{(t+1, t)}(a_{t+1}), t)$. Note that we have, 
        \begin{align}
            \langle \eta_{t+1}(a_{t+1}) \eta_t(a_t), \eta_t(\mathcal{M}_{t+0.5}) \rangle =0.
        \end{align}
        Here we make a slight abuse of notation above by identifying $\eta_t(\mathcal{M}_{t+0.5})$ with all $\eta_t(m)$ for $m \in \mathcal{M}_{t+0.5}$, and $s_t \in \mathcal{M}_{t+0.5}$ which by definition commutes with $\mathcal{M}_{t+0.5}$ due to the disjointness of supports. In this case, it necessarily follows that $u^\dagger_{(t+1, t)}(a_{t+1}) \in C(\mathcal{M}_{t+0.5})$ and $g(u^\dagger_{(t+1, t)}(a_{t+1}), t) \in \mathcal{G}_{(t+1, t)} \subset \mathscr{G}$ as a result. Note that for any $t=0, \cdots, T-1$, we have that $\eta_{t+1}(a_{t+1}) \eta_t(a_t) = \eta_t(s_t)g(u^{\dagger}_{(t+1, t)}(a_{t+1}), t)$, which implies that $\eta_{t+1}(a_{t+1})\eta_t(a_{t}) \sim_{\mathscr{M}} g(u^{\dagger}_{(t+1, t)}(a_{t+1}), t)$. Hence, for any three consecutive terms, we have that $\eta_{t}(a_{t})\eta_{t-1}(a_{t-1}) \sim_{\mathscr{M}} g(u^{\dagger}_{(t, t-1)}(a_{t}), t)$, which holds if and only if $a_t = u^\dagger_{(t+1, t)}(a_{t+1})$. This implies that $x \sim_{\mathscr{M}} \overleftarrow{\eta_T(a_T)}$ and $\overleftarrow{\eta_T(a_T)} \in \mathscr{G}^\perp$.
     \end{proof}
 \end{lemma}

\begin{lemma}\label{lemma: squeeze-lemma-subsystem-spacetime}
    Let $x \in \mathscr{G}^\perp$, where we denote it as $x = \prod^T_{t=0} \eta_t(a_t)$. Then the following holds, 
    \begin{enumerate}
        \item Suppose that $a_T = I$, then we have that $x \in \mathscr{M} = \mathscr{G} \cap \mathscr{G}^\perp$. 
        \item If $a_T \neq I$, then $a_T \in C(\mathcal{G})$.
    \end{enumerate}
    \begin{proof}
        Since $T$ is even, and $a_T =1$. The preceding Lemma~\ref{lemma: stabilizer-deform-backpropagation} shows that $x \sim_{\mathscr{M}} I$ and, hence, $x \in \mathscr{M}$. For the second statement, Note  
        \begin{align}
            \langle \overleftarrow{\eta_T(a_T)}, \eta_t(\mathcal{M}_{t+0.5}) \rangle = 0 = \langle a_T, u_{(T, t)(m_{t})} \rangle, \quad \forall m_t \in \mathcal{M}_{t+0.5},
        \end{align}
        for all $t=0, \cdots, T$. The second equality follows from Proposition~\ref{prop: spacetime-cumulant-automorphism}. Since by assumption, the measurement sequence generates the gauge group and the Clifford unitary preserves the gauge group, which implies that $a_T \in C(\mathcal{G})$ as desired. 
    \end{proof}
\end{lemma}

Equipped with the above technical lemma, we give a characterization of the bare logical subspace $\mathscr{G}^\perp/ \mathscr{M}$.

\begin{lemma}
    The orthogonal complement $\mathscr{G}^\perp \in \mathcal{A}$ with respect to the gauge group $\mathscr{G}$ defined with respect to a general measurement sequence is given by
    \begin{align}
        \mathscr{G}^\perp = \mathscr{M} \cdot \langle \overleftarrow{\eta(l)}: l \in \mathcal{L}_{\operatorname{bare}}\rangle, 
    \end{align}
    where $\mathscr{M} = \mathscr{G}^\perp \cap \mathscr{G}$. 
    \begin{proof}
       Let's consider the following map,
       \begin{align}\label{eq: Phi-code-circuit-iso-subsystem}
           \Phi: \mathcal{L}_{\operatorname{bare}} \rightarrow \mathscr{L}_{\operatorname{bare}},
       \end{align}
       given by the mapping $\Phi([p]_{\mathcal{S}}) = [\overleftarrow{\eta_T(p)}]_{\mathscr{M}}$, where the coset is defined with respect to base stabilizer group $\mathcal{S}$ and spacetime stabilizer group $\mathscr{M}$. First, we show this is a well-defined homomorphism. $\Phi$ is clearly well-defined: let $s \in \mathcal{S}$ so that $\Phi(ps) = \overleftarrow{\eta_T(ps)} = \overleftarrow{\eta_T(p)} \overleftarrow{\eta_T(s)}$ by Proposition~\ref{prop: spacetime-cumulant-automorphism}, then $\overleftarrow{\eta}_T(s) \in \mathscr{M}$ by construction. It also satisfies the homomorphism property by straightforward inspection. We now show that $\Phi$ is an isomorphism. Take any $x \in \mathscr{G}^\perp = \eta_T(a_T) \cdots \eta_0(a_0)$ and note that $a_T \in C(\mathcal{S})$ by construction. We argue this case-by-case, according to Lemma~\ref{lemma: squeeze-lemma-subsystem-spacetime}. 
       \begin{enumerate}
           \item Suppose that $a_T \in \mathcal{S}$. Then in this case, $x \sim_{\mathscr{M}} \eta_{T-1}(a_{T-1})\cdots \eta_0(a_0)$ and by Lemma~\ref{lemma: squeeze-lemma-subsystem-spacetime}, $x \in \mathscr{G} \cap \mathscr{G}^\perp = \mathscr{M}$ so that it is given by $\Phi(1)$. 
           \item The only other case is given by $a_T \in C(\mathcal{G}) / \mathcal{S} = \mathcal{L}_{\operatorname{bare}}$. Let $y = x \overleftarrow{\eta_T(a_T)}$, which implies that $y \in \mathscr{G}^\perp$ and $\eta_T(y) = I$ by construction. This implies again by Lemma~\ref{lemma: squeeze-lemma-subsystem-spacetime} $y \in \mathscr{M}$. Hence, we have that $\Phi(a_T) = [\overleftarrow{\eta_T(a_T)}]_{\mathscr{M}} = [y\overleftarrow{\eta_T(a_T)}]_{\mathscr{M}} = [x]_{\mathscr{M}}$.
       \end{enumerate}
       These two altogether imply that $\Phi$ is surjective. To show that $\Phi$ is an injection, it suffices to show that it has trivial kernel. Suppose that $x \sim_{\mathscr{M}} \overleftarrow{\eta_T(a_T)} \sim_{\mathscr{M}} 1$, it follows that $\overleftarrow{\eta_T(a_T)}  \in \mathscr{M} \subset \mathscr{G}$. This implies that $a_T \in \mathcal{G}$, since the measurement sequence generates the base gauge group. By the preceding Lemma~\ref{lemma: squeeze-lemma-subsystem-spacetime},  $a_T \in C(\mathcal{G})$, which implies that $a_T \in C(\mathcal{G}) \cap \mathcal{G} = \mathcal{S}$, which concludes the proof. 
    \end{proof}
\end{lemma}

The establishment of isomorphism $\Phi$ in Eq.~\eqref{eq: Phi-code-circuit-iso-subsystem} implies that we could establish a general code-to-circuit isomorphism, 

\begin{theorem}\label{thm: subsystem-circuit-to-code-isomorphism}
    Let $\{ \mathcal{M}_{0.5}, \cdots, \mathcal{M}_{T+0.5} \}$ be the measurement sequence that implements a base subsystem code $(\mathcal{G}, \mathcal{S})$ and $T$ is an even integer. Let the circuit $C_T(\mathcal{G})$ be given by the measurement sequence and  logical Clifford unitaries that preserve the gauge group and stabilizer subgroup scheduled according to Definition~\ref{def: circuit-model-from-subsystem-code}. On the spacetime Pauli group $\mathcal{A}$, define the gauge group $\mathscr{G}$ in Eq.~\ref{eq: spacetime-gauge-group} and its stabilizer (measurement) subgroup $\mathscr{M} = \mathscr{G}^\perp \cap \mathscr{G}$. Then this defines a code-to-circuit isomorphism $C_T(\mathcal{G}) \rightarrow (\mathscr{G}, \mathscr{M})$ such that, 
    \begin{align}
         \Phi: \mathcal{G}^\perp / \mathcal{S} \equiv \mathcal{L}_{\operatorname{bare}} \rightarrow \mathscr{G}^\perp / \mathscr{M} \equiv \mathscr{L}_{\operatorname{bare}}; \quad \Phi(b) = \overleftarrow{\eta_T(b)}
    \end{align}
    is a well-defined isomorphism. 
\end{theorem}

\begin{figure}
    \centering
    \includegraphics[width=1.0\linewidth]{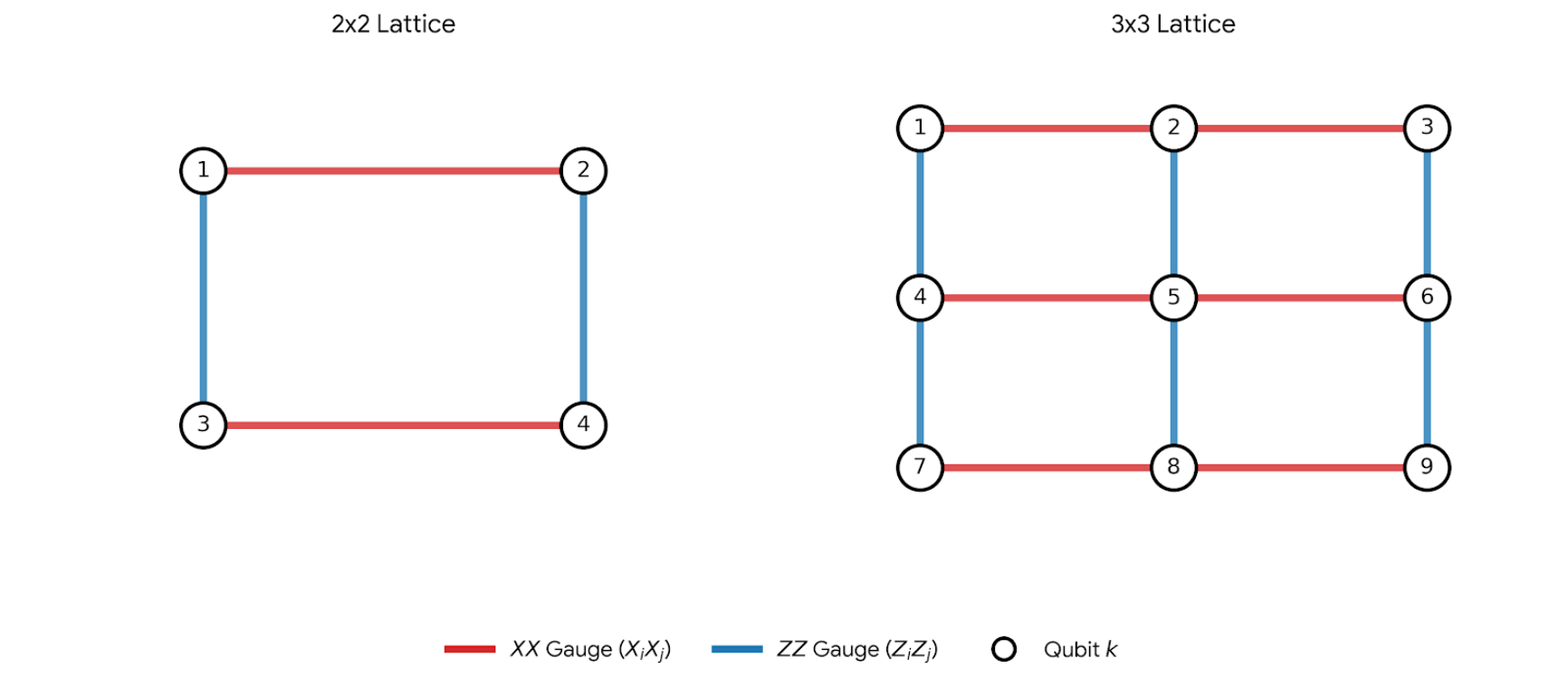}
    \caption{The Bacon-Shor code defined on the $2 \times 2$ and $3 \times 3$ lattices. The stabilizers are given as weight-$6$ operators.}
    \label{fig:BaconShorL3}
\end{figure}
\begin{example}[Bacon-Shor's code~\cite{Bacon_2006, alam2025baconshorboardgames}]\label{example: Bacon-shor}
    A typical measurement sequence in implementing the $2 \times 2$ Bacon-Shor's code is given as follows. 
    \begin{align}
        \mathcal{M}_{0.5} = \langle Z_1 Z_3, Z_2Z_4 \rangle ; \quad \mathcal{M}_{1.5} = \langle X_1 X_2, X_3X_4 \rangle . 
    \end{align}    
    The time is indexed with $t=0, 1, 2$. The gauge group $\mathscr{G}$ is constructed as, 
    \begin{align}
        \mathscr{G} = \langle \mathcal{G}_{(1, 0)},  \mathcal{G}_{(2, 1)} , \mathcal{M}_{1.5} , \mathcal{M}_{0.5} \rangle,
    \end{align}
    where $\mathcal{G}_{(1, 0)} := \{ g(a, 0): a \in \mathrm{P}^n, a \in C(\mathcal{M}_{0.5}) \}$ and $\mathcal{G}_{(2, 1)} := \{ g(a, 1): a \in \mathrm{P}^n, a \in C(\mathcal{M}_{1.5}) \}$, each is generated by $6$ independent operators. Hence, the gauge group is generated by $16$ operators, and concretely, 
    \begin{align}
        \mathcal{G}_{(1, 0)} &= \left \langle \begin{aligned}
            g(Z_1, 0), g(Z_2, 0), g(Z_3, 0), g(Z_4, 0), g(X_1X_3, 0), g(X_2X_4, 0)
        \end{aligned} \right \rangle, \\
        \mathcal{G}_{(2, 1)} &= \left \langle \begin{aligned}
            g(X_1, 1), g(X_2, 1), g(X_3, 1), g(X_4, 1), g(Z_1Z_2, 1), g(Z_3Z_4, 1)
        \end{aligned} \right \rangle.
    \end{align}
    This gives the group center $\mathscr{G}^\perp \cap \mathscr{G}$, 
    \begin{align}
        \mathscr{M} = \mathscr{G}^\perp \cap \mathscr{G} = \left\langle \begin{aligned}
            &\eta_0(Z_1Z_3), \eta_0(Z_2Z_4),  \eta_1(X_1X_2X_3X_4)\eta_0(X_1X_2X_3X_4) \\
            & \eta_2(Z_1Z_2Z_3Z_4)\eta_1(Z_1Z_2Z_3Z_4)\eta_0(Z_1Z_2Z_3Z_4), \eta_2(X_1X_2), \eta_2(X_3X_4)
        \end{aligned} \right\rangle. 
    \end{align}
    This confirms that the bare logical space is generated by precisely two elements, 
    \begin{align}
        \mathscr{L}_{\operatorname{bare}} = \langle \eta_2(X_1X_3)\eta_1(X_1X_3)\eta_0(X_1X_3), \eta_2(Z_1Z_2)\eta_1(Z_1Z_2)\eta_0(Z_1Z_2) \rangle. 
    \end{align}
    The center (measurement group) consists of three components: $i)$ the initial stabilizer group, $ii)$ the output stabilizer group, $iii)$ the dynamical stabilizers which commute with the previous history of the measurement sequence. We now examine the case with $3 \times 3$ Bacon-Shor code, with the following measurement sequence,  
    \begin{align}
        \begin{aligned}
            \mathcal{M}_{0.5} &= \langle Z_1Z_4, Z_2Z_5, Z_3Z_6 \rangle, \\
            \mathcal{M}_{1.5} &= \langle Z_4Z_7, Z_5Z_8, Z_6Z_9 \rangle, \\
            \mathcal{M}_{2.5} &= \langle X_1X_2, X_4X_5, X_7X_8 \rangle, \\
            \mathcal{M}_{3.5} &= \langle X_2X_3, X_5X_6, X_8X_9 \rangle. 
        \end{aligned}
    \end{align}
    The stabilizer generators in this case can be given by the weight-$6$ operators. 
    \begin{align}
        \begin{aligned}
& S_Z^{12}=\left(Z_1 Z_4\right)\left(Z_2 Z_5\right)\left(Z_3 Z_6\right),\\
& S_Z^{23}=\left(Z_4 Z_7\right)\left(Z_5 Z_8\right)\left(Z_6 Z_9\right), \\
& S_X^{12}=\left(X_1 X_2\right)\left(X_4 X_5\right)\left(X_7 X_8\right), \\
& S_X^{23}=\left(X_2 X_3\right)\left(X_5 X_6\right)\left(X_8 X_9\right).
\end{aligned}
    \end{align}
    The measurement (stabilizer) subgroup $\mathscr{M}$ is given, 
    \begin{align}
         \begin{aligned}\mathscr{M} =  \langle
            & \eta_0(Z_1Z_4), \eta_0(Z_2Z_5), \eta_0(Z_3Z_6), \\
            & \eta_1(Z_4Z_7) \eta_0(Z_4Z_7), \eta_1(Z_5Z_8)\eta_0(Z_5Z_8), \eta_1(Z_6Z_9)\eta_0(Z_6Z_9), \\
            & \eta_4(X_1X_2)\eta_3(X_1X_2), \eta_4(X_4X_5)\eta_3(X_4X_5), \eta_4(X_7X_8)\eta_3(X_7X_8), \\
            & \eta_4(X_2X_3), \eta_4(X_5X_6), \eta_4(X_8X_9), \\
            & \eta_4(S^{12}_Z)\eta_3(S^{12}_Z)\eta_2(S^{12}_Z)\eta_1(S^{12}_Z)\eta_0(S^{12}_Z), \eta_4(S^{23}_Z)\eta_3(S^{23}_Z)\eta_2(S^{23}_Z)\eta_1(S^{23}_Z)\eta_0(S^{23}_Z), \\
            &\eta_4(S^{12}_X)\eta_3(S^{12}_X)\eta_2(S^{12}_X)\eta_1(S^{12}_X)\eta_0(S^{12}_X), \eta_4(S^{23}_X)\eta_3(S^{23}_X)\eta_2(S^{23}_X)\eta_1(S^{23}_X)\eta_0(S^{23}_X) \rangle. 
        \end{aligned} 
    \end{align}
    There are in total $16$ independent generators. The gauge group in total has $72$ many independent generators. Then the (bare) logical subspace has $90-72 -16 = 2$ independent generators, which encode precisely one logical qubit.  
\end{example}

\subsubsection*{Application to syndrome extraction circuit}

Theorem~\ref{thm: subsystem-circuit-to-code-isomorphism} provides a code-to-circuit isomorphism in a general setting, without explicitly determining the measurement subgroup $\mathscr{M}$. In the learnability framework detailed in Section~\ref{sec: learnability-general}, it is of great interest to explicitly write down generators to (at least, subgroup) of $\mathscr{M}$. Note that, although there is no intrinsic barrier to determining $\mathscr{M}$ by analyzing the commutation relation with the generators from $\mathscr{G}$, we provide a clean characterization of generators in $\mathscr{M}$, assuming that the circuit $C_T(\mathcal{S})$ implements strictly a stabilizer code.

\begin{lemma}\label{lemma: measurement-subgroup-syndrome-extraction-circuit}
Let $C_T(\mathcal{S})$ denote the syndrome extraction circuit in Definition~\ref{def: syndrome-extraction-circuit}. The orthogonal complement $\mathscr{G}^\perp \subset \mathcal{A}$ with respect to the spacetime Pauli group is generated by  
\begin{align}
    \mathscr{G}^\perp = \langle \overleftarrow{m}: m \in \eta_t(\mathcal{M}_{t+0.5}), 0 \leq t < T \rangle \cdot \langle \overleftarrow{l}: l \in \eta_T(C(\mathcal{S})).
\end{align}
 
    \begin{proof}
    First it is easy to see that $\langle \overleftarrow{m}: m \in \eta_t(\mathcal{M}_{t+0.5}), 0 \leq t <T \rangle \cdot \langle \overleftarrow{l}: l \in \eta_T(C(\mathcal{S})) \rangle \subseteq \mathscr{G}^\perp$. This is true since we have assumed that the Clifford unitaries and Pauli measurements preserve the stabilizer space. We now show that this is in fact, the maximal subgroup. Let $x \in \mathscr{G}^\perp \subset \mathcal{A}$ such that $x = \prod^T_{t=0} \eta_t(a_t)$ and $a_t \in \mathrm{P}^{n}$. Then we have immediately that 
        $a_t \in C(\mathcal{M}_{t+0.5})$. Furthermore, we have that $\langle x, g(b, t) \rangle = \langle a_t, b\rangle + \langle a_{t+1}, u_{(t+1, t)}(b) \rangle = \langle a_{t+1}  u^\dagger_{(t+1, t)}(a_t), b \rangle = 0$ for $b \in C(\mathcal{M}_{t+0.5})$. This implies that $a_{t}  u^\dagger_{(t+1, t)}(a_{t+1}) \in C(C(\mathcal{M}_{t+0.5}))$. Since we ignore the phase, by Lemma~\ref{lemma: double-commutant-phase}, we have that 
        \begin{align}
           s_t \equiv  a_{t}  u^\dagger_{(t+1, t)}(a_{t+1}) \in \mathcal{M}_{t+0.5},
        \end{align}
        for $t=0, 1, \cdots, T-1$, and $s_T = a_T$. Note that  $\overleftarrow{\eta_t (s_t)} \in \mathscr{G}^\perp$ by definition of the gauge group~\eqref{eq: spacetime-gauge-group} and the fact that elements in $\mathcal{M}_{t+0.5}$ commute with each other due to the disjointedness of supports. To this end, let's consider 
        \begin{align}
            \begin{aligned}
                &x \prod^0_{i=T-1} \overleftarrow{\eta_i(s_i)} = (\prod^T_{j=0} \eta_{j}(a_j)) \overleftarrow{\eta_{T-1}(a_{T-1})} \overleftarrow{\eta_{T-1}(u^{\dagger}_{(T, T-1)}(a_{T}))} \overleftarrow{\eta_{T-2}(a_{T-2})} \overleftarrow{\eta_{T-2}(u^{\dagger}_{(T-1, T-2)}(a_{T-1}))} \cdots \overleftarrow{\eta_0(a_0)} \overleftarrow{\eta_0(u^{\dagger}_{(1, 0)}(a_1))},\\
                &= \overleftarrow{\eta_T(x)} = \overleftarrow{\eta_T(a_T)}.
            \end{aligned}
        \end{align}
        Note that the telescope cancellation is given by, $\eta_{t}(a_{t})\overleftarrow{\eta_{t-1}(u^{\dagger}_{(t, t-1)}(a_{t}))} = \overleftarrow{\eta_{t}(a_{t})}$ for $t=1, \cdots, T-1$, while $t=0$ holds trivially. The only terms thus left are given by the second line.  
    \end{proof}
\end{lemma}

This allows us to characterize the syndrome extraction circuits as a subsystem code.

\begin{theorem}[Circuit-(subsystem) code isomorphism]\label{thm: circuit-subsystem-code-1}
     Without loss of generality, we consider $T$ is even. Let our gauge group to be defined by Eq.~\eqref{eq: spacetime-gauge-group}.  Then this defines a subsystem code of code parameters $[[(T+1)n, k ]]$ such that
    \begin{enumerate}
        \item The measurement (stabilizer) subgroup $\mathscr{M}$ is generated by the backward propagation of each syndrome measurement from $\mathcal{M}_t, t=0, \cdots, T$
        
        \item The bare logical  $\mathscr{L}_{\operatorname{bare}}$ is generated by the backward propagation of the base logical space $C(\mathcal{S}) / \mathcal{S}$. 
    \end{enumerate}
    \begin{proof}
        By the preceding lemma and the fact that $T+1$ is odd, we have that 
        \begin{align}
            &\mathscr{M} = \langle \overleftarrow{m}: m \in \eta_t(\mathcal{M}_{t+0.5}), 0 \leq t < T \rangle \cdot \langle \overleftarrow{s}: s \in \eta_T(\mathcal{S}) \rangle, \\
            &\mathscr{L}_{\operatorname{bare}} = \langle \overleftarrow{\eta_T(l)}: l \in C(\mathcal{S}) / \mathcal{S}\rangle.
        \end{align}
        Then we have that $\mathscr{G}^\perp = \mathscr{M} \cdot \mathscr{L}_{\operatorname{bare}}$, which proves the result. 
    \end{proof}
\end{theorem}

Note that our circuit-to-code map serves as a slight variant of the formalism introduced in \cite{delfosse2023spacetimecodescliffordcircuits}. Though there is no explicit notion of outcome codes, the relation does appear implicitly, which we now examine. 

\begin{remark}[Appearance of outcome codes]
     We can examine the center of the gauge group $\mathscr{M}$ to be generated by two parts. 
    \begin{align}
        \mathscr{M} = \langle \overleftarrow{\eta_T(s)}: s \in \mathcal{S} \rangle \cdot \langle \overleftarrow{\eta_t(m_t)} : m_t \in \mathcal{M}_{t+0.5},  T-1 \geq t \geq 0\rangle.
    \end{align}
    The first part serves to be the final and ideal syndrome extraction round. Let $s \in \mathcal{M}_{t+0.5}$ and its backward propagation $\overleftarrow{\eta_t(s)}$ and suppose that there is some earlier time $t'$ such that $s \in \mathcal{M}_{t' +0.5}$, then we multiply them and replace with the generator $\overleftarrow{\eta_t(s)} \overleftarrow{\eta_{t'}(s)}$. Note that this always gives deterministic measurement outcomes, which defines a circuit detector. Furthermore, denote the number of generators for $\mathscr{M}$ is $M$. For each such generator, it is possible to record these relations, which provide a detector matrix as a parity check to $\mathcal{O}(\mathcal{C})$. This gives a formal connection to the spacetime circuit description found in \cite{delfosse2023spacetimecodescliffordcircuits} and the detector matrix in \cite{derks2024designing}, while we here focus purely on measuring the stabilizers. 
    
\end{remark}

We now give a concrete example. 

\begin{example}[3 rounds of syndrome extraction with repetition codes]\label{example: repetition-code}
     Let's consider we have $3$ rounds of syndrome measurements and our measurement group
    \begin{align}
        \mathcal{M}_{i+0.5, \text{even}}  = \langle Z_1Z_2 \rangle; \quad  \mathcal{M}_{i + 0.5, \text{odd}}  = \langle Z_2Z_3 \rangle,
    \end{align}
    for $i=0, 1, \cdots, 5$. In the case of $i =6$, we associated with the full stabilizer group $\mathcal{M}_{6.5} = \mathcal{S}$. We now construct our gauge group. First, we consider the following subgroup,  
    \begin{align}
        \begin{aligned}
            &\mathcal{G}_{(i+1, i), \text{even}}  = \langle g(X_3, i), g(X_1X_2X_3, i) , g(Z_1, i), g(Z_2, i), g(Z_3, i) \rangle, \\
        &\mathcal{G}_{(i+1, i), \text{odd}}  = \langle g(X_1, i), g(X_1X_2X_3, i) , g(Z_1, i), g(Z_2, i), g(Z_3, i) \rangle .
        \end{aligned}
    \end{align}
    Then the gauge group $\mathscr{G} \subseteq \mathrm{P}^{21}$ can be constructed 
        \begin{align}
           \mathscr{G} = \langle \mathcal{G}_{(i+1, i)}, \mathcal{M}_{i+0.5}: i=0, \cdots 5 \rangle. 
        \end{align}
      
        The center of the gauge group $\mathscr{M} = \mathscr{G}^\perp \cap \mathscr{G}$ can be given by 
    \begin{align}
        \begin{aligned}
            &\eta_6(Z_1Z_2)\eta_5(Z_1Z_2)\eta_4(Z_1Z_2)\eta_3(Z_1Z_2)\eta_2(Z_1Z_2)\eta_1(Z_1Z_2)\eta_0(Z_1Z_2), \\
            & \eta_6(Z_2Z_3)\eta_5(Z_2Z_3)\eta_4(Z_2Z_3)\eta_3(Z_2Z_3)\eta_2(Z_2Z_3)\eta_1(Z_2Z_3)\eta_0(Z_2Z_3) \\
            & \eta_5(Z_2Z_3)\eta_4(Z_2Z_3)\eta_3(Z_2Z_3)\eta_2(Z_2Z_3)\eta_1(Z_2Z_3)\eta_0(Z_2Z_3) \\
            &\eta_4(Z_1Z_2)\eta_3(Z_1Z_2)\eta_2(Z_1Z_2)\eta_1(Z_1Z_2)\eta_0(Z_1Z_2) \\
            &\eta_3(Z_2Z_3)\eta_2(Z_2Z_3)\eta_1(Z_2Z_3)\eta_0(Z_2Z_3) \\
            &\eta_2(Z_1Z_2)\eta_1(Z_1Z_2)\eta_0(Z_1Z_2) \\
            &\eta_1(Z_2Z_3)\eta_0(Z_2Z_3) \\
            & \eta_0(Z_1Z_2),
        \end{aligned}
    \end{align}
which gives $8$ desired generators. Note that by the dimension counting over symplectic space, the gauge group has $ 5 \times 6 + 2 = 32$ linearly independent operators so that $\dim \mathscr{G}^\perp = 42 - 32 = 10$. Then it follows that $\dim \mathscr{G}^\perp / \mathscr{M} = 10 - 8 = 2$, which encodes precisely one logical qubit. It is easy to see that the choice of such a logical representative is given by 
\begin{align}
    \eta_6(X_1X_2X_3) \cdots \eta_0(X_1X_2X_3),
\end{align}
as desired. 
\end{example}

In the syndrome extraction circuit, we might be interested in performing logical measurements. Treating the logical measurements as an abstract gadget, this can be cast in the same framework in Theorem~\ref{thm: subsystem-circuit-to-code-isomorphism}, treating the gauge group defined by the stabilizer and logical operators measured. Hence, a simple corollary can be stated, 
\begin{corollary}
    Let $C_T(\mathcal{G})$ be the circuit that implements the gauge group $\mathcal{G} = \langle \mathcal{S}, \mathcal{L}_M \rangle$, where $\mathcal{L}_{M}$ denotes the measured logical operators. Denote $\mathcal{L}^C_M$ denotes its orthogonal complementary span of logical operators in $C(\mathcal{S})/\mathcal{S}$. Then, the circuit-to-code isomorphism  $C_T(\mathcal{G}) \rightarrow  (\mathscr{G}, \mathscr{M})$, such that $\mathscr{G}^\perp  = \mathscr{M} \cdot \langle \overleftarrow{\eta_T(l)}: l \in \mathcal{L}^C_M \rangle $.
\end{corollary}

\subsubsection*{Syndrome extraction circuits with ancillary worldlines}

We now present a small generalization of the circuit-to (subsystem)-code mapping with ancillary worldlines by sacrificing some generality in favor of more explicit understandings, which would become beneficial in specific scenarios. For simplicity of discussion, we omit the mid-circuit logical Pauli measurements, as they require additional ancilla qubits. Here we present a protocol that generalizes that from Ref.~\cite{bacon2017sparse}, while we leave a more systematic construction of circuit-to-code mappings in future work. Denote $r_0$, $r_a$ to be respectively the projection of elements in $\mathcal{A}$ to the supports on the data qubits (respectively ancilla qubits). We now aim to define the gauge group of interests, akin to~\cite{bacon2017sparse, gottesman2022opportunitieschallengesfaulttolerantquantum}. For the simplicity of discussion, we say that $\mathcal{C}_T(\mathcal{S})$ consists of $L$ rounds of syndrome extraction. For $i$th round of syndrome extraction, we perform transversal $Z$ measurements onto the ancillae at time $t_i$ with $t_0 =0$ and $t_L =T$. We first consider the case where we only perform one round of syndrome extractions. Then we denote the resulting operator by, 
\begin{align}
    V_{t_1} := (I \otimes \langle 0 |^{\otimes r}) u_{(t_1, 0)} (I \otimes | 0 \rangle^{\otimes r}).
\end{align}
Then we have that, 
\begin{align}\label{eq: good-error-detecting-circuit}
    \begin{aligned}
         V^\dagger_{t_1} V_{t_1} &= (I \otimes \langle 0 |^{\otimes r}) u^\dagger_{(t_1, 0)} (I \otimes | 0 \rangle^{\otimes r})(I \otimes \langle 0 |^{\otimes r}) u_{(t_1, 0)} (I \otimes | 0 \rangle^{\otimes r}) \\
         &\propto \frac{1}{|\mathcal{S}|} \sum_{s \in \mathcal{S}} s. 
    \end{aligned}
\end{align}
Note that $V_{t_1}$ is referred to as the \emph{good error-detecting} circuit in Ref.~\cite{bacon2017sparse} and the addition of a logical Clifford unitary to the original construction presented in Ref.~\cite{bacon2017sparse} is harmless, as they also preserve the stabilizers. Let $u_{(t_1, 0)}$ denotes the Clifford unitary implementing this circuit. Note that by the above definition, $r_a(u^\dagger_{(t_1, 0)}(\mathcal{M}_{t_1 +0.5})) = \mathcal{M}_{t_1 +0.5} = r_a(u_{(t_1, 0)}(\mathcal{M}_{t_1 +0.5})) =  \mathcal{M}_{0.5} \equiv \mathcal{M}$. For $i=0, \cdots, L-1$, we assume the following periodic structure to simplify the discussion, 
\begin{align}\label{eq: syndrome-periodic-condition}
    \langle u_{(t_{i+1}, t_i)}(\mathcal{M}) \cup \mathcal{M} \rangle = \langle \mathcal{S}, \mathcal{M} \rangle.
\end{align}
Roughly speaking, this assumption ensures that between each interval $[t_{i}, t_{i+1}]$, we implement a full round of syndrome extraction, associated with the base stabilizer group $\mathcal{S}$. Let's denote $t_0, t_1, \cdots, t_L$ to be the (periodic) time slices for which we perform transversal $Z$ measurements $\mathcal{M}$ on the ancilla qubits. We consider the gauge group,

\begin{align}\label{eq: gauge-group-ancilla-worldlines}
    \mathscr{G} =  \langle \prod^L_{i=0} \eta_{t_i}(\mathcal{M}),  \prod^{T-1}_{i=0} \mathcal{G}_{(i+1, i)} ,  \eta_T(\mathcal{S}) \rangle,
\end{align}
where the subscripts are denoted in consistency with the above derivation in different time intervals during which a syndrome extraction round is performed. For $i=1, \cdots, L-1$,
\begin{align}
    \mathcal{G}_{(t_i+1, t_i)} = \{ \eta_{t_i+1}(u_{(t_i+1, t_i)}(a)) \eta_{t_i}(a): a \in \mathrm{P}^{n+r}; \: r_a(a) \in \mathcal{M} \}. 
\end{align}
Note that $r_a(a) \in \mathcal{M} \iff a \in C(\mathcal{M})$, so that this gauge group Eq.~\eqref{eq: gauge-group-ancilla-worldlines} is defined in the analogous fashion compared with the case without ancillary worldlines. For $t \neq t_i$ for $i=0, \cdots, L$, 
\begin{align}
    \mathcal{G}_{(t+1, t)} := \langle g(a, t), a \in \mathrm{P}^{n+r} \}.
\end{align}
Lastly, we assume again the perfect syndrome measurement at the end. similar to the above cases without ancillary worldlines, the addition of $\eta_T(\mathcal{S})$ to the generating set of gauge group is not necessary: it can be generated without explicitly assuming $\eta_T(\mathcal{S})$. The inclusion in our case, however, simplifies the assumptions (even and oddness of all the intermediate times when we perform measurements) and simplifies the proof. We state an analogous lemma to Lemma~\ref{lemma: squeeze-lemma-subsystem-spacetime}, incorporating our new assumption on the propagation of Clifford unitaries. 

\begin{lemma}
    Let $\mathscr{G} \subset \mathcal{A}$ be defined in Eq.~\eqref{eq: gauge-group-ancilla-worldlines}. Denote $\mathscr{G}^\perp$ to be its orthogonal complement and $\mathscr{M} = \mathscr{G}^\perp \cap \mathscr{G}$, then 
    \begin{align}
        \mathscr{G}^\perp = \mathscr{M} \cdot  \langle \overleftarrow{\eta_T(a)}, r_0(a) \in C(\mathcal{S})/ \mathcal{S}, r_a(a) \in \mathcal{M} \rangle. 
    \end{align}

\begin{proof}
       Let $x \in \mathscr{G}^\perp$ with the gauge group defined in Eq.~\eqref{eq: gauge-group-ancilla-worldlines} and let $x = \eta_T(a_T)\cdots \eta_0(a_0)$. By Lemma~\ref{lemma: stabilizer-deform-backpropagation}, if $x \notin \mathscr{M}$, then it must be $x \sim_{\mathscr{M}} \overleftarrow{\eta_T(a_T)} \in \mathscr{G}^\perp$. Then $\langle \overleftarrow{\eta_T(a_T)}, \eta_{t_i}(\mathcal{M}) \rangle = \langle \eta_T(a_T), \overrightarrow{\eta_{t_i}(\mathcal{M})}\rangle =  \langle \eta_T(a_T),u_{(T, t_i)}( \mathcal{M})\rangle =0$. By the assumption of the good error-detecting circuit, we have that $r_0(u_{(T, t_i)}(\mathcal{M})) = \mathcal{S}$ and $r_a(u_{(T, t_i)}(\mathcal{M})) = \mathcal{M}$. This forces that $r_0(a_T) \in C(S)$ and $r_a(a_T) \in r_a(C(\mathcal{M})) = \mathcal{M}$. By the construction of the gauge group Eq.~\eqref{eq: gauge-group-ancilla-worldlines}, we have that if $r_0(a_T) \in \mathcal{S}$, then $\overleftarrow{\eta_T(a_T)} \in \mathscr{M}$, which concludes the proof. 
\end{proof}
\end{lemma}

The above lemma already concludes for the code-to-circuit isomorphism with ancillary worldlines. We now also give a clear characterization of the measurement subgroup $\mathscr{M}$. 

\begin{corollary}
    Let the measurement subgroup be $\mathscr{M} = \mathscr{G}^\perp \cap \mathscr{G}$ with the gauge group defined in Eq.~\eqref{eq: gauge-group-ancilla-worldlines}. Then, 
    \begin{align}
        \mathscr{M} = \langle \overleftarrow{\eta_{t_i}(m)}, m \in \mathcal{M} \rangle \cdot \langle \overleftarrow{\eta_T(s)}: s \in \mathcal{S} \rangle.  
    \end{align}
    \begin{proof}
        It is clear that $\langle \overleftarrow{\eta_{t_i}(m)}, m \in \mathcal{M} \rangle \cdot \langle \overleftarrow{\eta_T(s)}: s \in \mathcal{S} \rangle \in \mathscr{M}$ and we show that it is the maximal subgroup. The proof follows by minimal modification from Lemma~\ref{lemma: measurement-subgroup-syndrome-extraction-circuit} by incorporating the ancillary worldlines. 
    \end{proof}
\end{corollary}

This allows us to conclude the following, 
\begin{theorem}[Circuit-(subsystem) code isomorphism with ancillary worldlines]\label{thm: circuit-subsystem-code-ancilla}
     Without loss of generality, we consider $T$ is even. Let our gauge group to be defined by Eq.~\eqref{eq: gauge-group-ancilla-worldlines}.  Then this defines a subsystem code of code parameters $[[(T+1)(n+r), k ]]$ such that
    \begin{enumerate}
        \item The measurement (stabilizer) subgroup $\mathscr{M}$ is generated by, 
       \begin{align*}
        \mathscr{M} = \prod^{L}_{i=1}\langle \overleftarrow{\eta_{t_i}(m)}, m \in \mathcal{M} \rangle \cdot \langle \overleftarrow{\eta_T(s)}: s \in \mathcal{S} \rangle.  
    \end{align*}
        
        \item The bare logical  $\mathscr{L}_{\operatorname{bare}}$ is generated by with one-to-one correspondence with the backward propagation of the base logical space
      \begin{align*}
          \mathscr{L}_{\operatorname{bare}} = \langle \overleftarrow{\eta_T(a)}, r_0(a) \in C(\mathcal{S})/ \mathcal{S}, r_a(a) \in \mathcal{M} \rangle. 
      \end{align*}
        
    \end{enumerate}
\end{theorem}

\begin{figure*}[pt]
\centering

\includegraphics[width=0.9\linewidth]{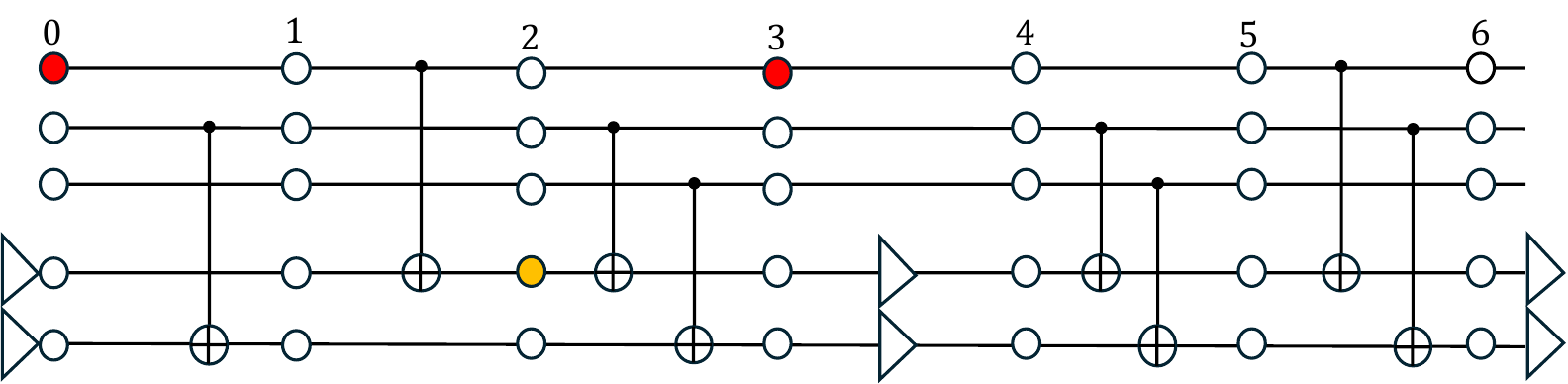}

\caption{The syndrome extraction circuit that implements the $3$-qubit repetition code. Note that the stabilizer measurement fault at $\eta_2(X_4)$ can be gauge-related to $\eta_3(X_1) \eta_0(X_1)$, which explains the modeling of measurement faults of the spacetime circuit without the ancilla. }

\label{fig: rep-circuit-ancilla}
\end{figure*}

\begin{example}\label{example: repetition-code-ancilla}
    We also give an example of the circuit-to-code mapping presented with ancilla worldlines using the repetition code. Let $n=3$ and $r=2$, we implement the following circuit with two rounds of syndrome extractions, See Figure~\ref{fig: rep-circuit-ancilla}. The gauge group $\mathscr{G}$ is given by, 
    \begin{align}
        \mathcal{G}_{(i+1, i)} = \{ g(a, i), a \in \mathrm{P}^{5} \}; \quad \mathcal{G}_{(3, 2)} = \{ \eta_{3}(u_{(3, 2)}(a)) \eta_{2}(a): a \in \mathrm{P}^{5}; \: a \in C(\mathcal{M}) \}, 
    \end{align}
    where $\mathcal{M} = \langle Z_4, Z_5 \rangle$ and for $i =0,1, 3, 4, 5$. Specifically, 
    \begin{align}
        \mathcal{G}_{(3, 2)} = \left\{ 
        \begin{aligned}
             &\eta_3(X_1)\eta_2(X_1), \eta_3(X_2)\eta_2(X_2X_2), \eta_3(X_3)\eta_2(X_3),  \eta_3(Z_j)\eta_2(Z_j), j=1, \cdots, 5
        \end{aligned}
         \right\}.
    \end{align}
    In addition to $\langle \eta_6(\mathcal{M}), \eta_2(\mathcal{M}),  \eta_0(\mathcal{M}) \rangle$. Note that in this case $\eta_T(\mathcal{S})$ can be generated given the above generated, so it is redundant. Then the stabilizers are given by, 
    \begin{align}
        \mathscr{M} = \langle \overleftarrow{\eta_2(Z_1)}, \overleftarrow{\eta_2(Z_2)}, \overleftarrow{\eta_6(Z_1Z_2Z_4)}, \overleftarrow{\eta_6(Z_2Z_3Z_5)}, \overleftarrow{\eta_6(Z_1Z_2)}, \overleftarrow{\eta_6(Z_2Z_3)}  \rangle.
    \end{align}
    Hence, in total there are $58 + 4 =62$ total number of degrees of freedom for the gauge group and $\mathscr{M}$ has $6$ independent generators. This gives in total $2^2$ for the $\mathscr{L}_{\operatorname{bare}}$ that encodes a logical qubit. Similarly to the spacetime circuit code without ancilla qubits, the mapping assumes the perfect round of final syndrome measurements. In this case, these are given by $\langle \overleftarrow{\eta_6(Z_1Z_2)}, \overleftarrow{\eta_6(Z_2Z_3)}  \rangle$, without ancilla qubits. 
\end{example}


\subsection{Learnability of circuit Pauli faults}\label{subsec: learnability-SE-circuits-pauli-faults}

Finally we apply the circuit-to-(subsystem) code isomorphism Theorem~\ref{thm: circuit-subsystem-code-1} in our learnability of circuit Pauli faults. Denote $\sigma: \mathcal{A} \rightarrow \mathbb{F}^m_2$ the syndrome map associated with $\mathscr{M}$. 

\begin{lemma}
     For two circuit Pauli faults $a, b \in \mathcal{A}$, their final accumulated errors at $T$ are given respectively by $\eta_T(\overrightarrow{a})$ and $\eta_T(\overrightarrow{b})$. Then, 
     \begin{enumerate}
         \item If $\sigma(a) = \sigma(b)$, then $\eta_T(\overrightarrow{a})$ and $\eta_T(\overrightarrow{b})$ must have the same syndrome associated with the output stabilizer group $\mathcal{M}_{T+0.5} = \mathcal{S}$.

         \item For any $a \sim_{\mathscr{G}} b$, that is, there exists some element in the gauge group $g \in \mathscr{G}$ such that $a = b  g$. Then there always exists some stabilizers $s \in \mathcal{S} = \mathcal{M}_{T+0.5}$ such that $\eta_T(\overrightarrow{a}) = \eta_T(\overrightarrow{b} )\eta_T(s)$. 
     \end{enumerate}
     \begin{proof}
         The first statement is trivial since $\mathscr{M}$ itself contains $\eta_T(\mathcal{S})$ and by the frame invariance with respect to $\beta(\cdot, \cdot)$. For the second statement. Recall that $\overrightarrow{bg} = \overrightarrow{b}\overrightarrow{g}$ and $\eta_T(\overrightarrow{g})$ vanishes unless $g \in \mathscr{M}$ or $g \in \mathcal{G}_{(T, T-1)}$ so that $\eta_T(\overrightarrow{g}) \in \mathcal{S}$ as desired. 
     \end{proof}
\end{lemma}

In either cases, we can describe a probability distribution $\Lambda \in \mathbb{R}[\mathcal{A}]$ where the spacetime Pauli (Boolean) group $\mathcal{A} $. Then we wish to determine, 
\begin{align}\label{eq: eff-distribution-spacetime}
    p^{\operatorname{eff}}_a = \frac{1}{|\mathscr{G}|} \sum_{g \in \mathscr{G}} p_{ag}.
\end{align}
Hence, we can learn the circuit Pauli faults in a way that is mathematically analogous to that of the subsystem code.  

\begin{theorem}[Learnability for circuit Pauli faults]\label{thm: circuit-logical-learnability}
     Given a fault distribution $\Lambda$ defined on the spacetime Pauli group such that all its Pauli eigenvalues are greater than $0$. Furthermore, denote $\mathcal{K}$ to be the set of spacetime Pauli noise parameters parameterized in the character basis Eq.~\eqref{eq: pauli-character-fourier}. 
    \begin{enumerate}
        \item The distribution is learnable if and only if each $a \in \mathcal{K} \subseteq \mathcal{A}$ corresponds to unique syndrome. 
        \item The effective distribution Eq.~\eqref{eq: eff-distribution-spacetime} is learnable if and only if for any $a \in \mathcal{K}$ and $a' \in \mathcal{K}$, $\sigma(a) = \sigma(a') \iff a \sim_{\mathscr{G}} a' $.
    \end{enumerate}
\end{theorem}

It is natural to ask the case of unlearnability. To see that we can explicitly characterize the learnable degrees of freedom similar to writing the matrix $D_{\mathscr{G}^\perp}$ in Theorem~\ref{thm: learn-eff-distribution-from-syndrome}. We can still parameterize the Pauli fault distribution, 
\begin{align}
    \Lambda = *_{a \in \mathcal{K}} ((1-q_a) + q_a \chi_a),
\end{align}
where there would exist $a, b \in \mathcal{K}$ with the same syndrome $\sigma(a) = \sigma(b)$. Taking the logarithm as standard and we can partition the matrix $D_{\mathscr{G}^\perp}$ as follows, 
\begin{align}
    D_{\mathscr{G}^\perp} = \left(\begin{array}{c|c|c} 
  A_{\mathscr{M}} & B_{\mathscr{M}} & C_{\mathscr{M}} \\ 
  \hline 
  A_{\mathscr{G}^\perp / \mathscr{M}} & B_{\mathscr{G}^\perp / \mathscr{M}}  & C_{\mathscr{G}^\perp / \mathscr{M}}
\end{array}\right).
\end{align}
Let $A$ denote these with unique syndromes and $B$ denote these with duplicate syndromes but related to gauge group $\mathscr{G}$. Let $C$ denote the duplicate syndromes that are not learnable (not related by the gauge group). In this case, one can always find a basis such that, 
\begin{align}
    D_{\mathscr{G}^\perp} = \left(\begin{array}{c|c|c} 
  A_{\mathscr{M}} & 0 & 0\\ 
  \hline 
  A_{\mathscr{G}^\perp / \mathscr{M}} & 0  & C'_{\mathscr{G}^\perp / \mathscr{M}}
\end{array}\right),
\end{align}
where each (nonzero) column within the partition $C$ undergoes a column reduction with the column of $A$ sharing the same syndrome. By definition of $C$, $C'_{\mathscr{G}^\perp / \mathscr{M}}$ cannot be all zero.. This basis representation defines the \emph{prior distribution}, 

\begin{definition}[Prior distribution]\label{def: prior-distribution} 
    Let $\Lambda$ be a (Pauli) fault distribution on $\mathcal{C}_T(\mathcal{S})$ be parameterized in the character basis Eq.~\eqref{eq: main-character-fourier-basis}. Its prior distribution $\Lambda^{\operatorname{prior}}$ is given by 
    \begin{align}
        \Lambda^{\operatorname{prior}} = *_{c \in \mathcal{K}_{\mathscr{M}}} \left( (1-q_{[c]}) + q_{[c]} \chi_c \right),
    \end{align}
    such that we have 
    \begin{align}
        1 - 2 q_{[c]} = \prod_{a \in \mathcal{K}: \sigma(a) = \sigma(c)} (1 - 2q_a). 
    \end{align}
\end{definition}

In the case where all the non-vanishing $q_a \in (0, 1/2)$, then the above formula can be seen as the XOR combination of fault probabilities corresponding to the same syndrome. Note that it is clear that prior distribution can always be estimated from the syndrome and it is of interest to ask to what extent it can represent the logical error probabilities. 

\begin{corollary}\label{cor: prior-distribution-learn-logical-error}
    The prior distribution in Definition~\ref{def: prior-distribution} is capable of representing the logical error probabilities if and only if $C'_{\mathscr{G}^{\perp}/\mathscr{M}}=0$, in other words, for any $a, a' \in \mathcal{K}$, $\sigma(a) = \sigma(a') \iff a \sim_{\mathscr{G}} a'$. 
    \begin{proof}
        The forward direction is straightforward. If $C'_{\mathscr{G}^{\perp}/\mathscr{M}}=0$, then we can represent $\lambda_b$ for any $b \in \mathscr{G}^\perp$ as a product of Fourier coefficients (Pauli fidelity). Hence, $\lambda_b = \prod_{c \in \mathcal{K}: \langle c, b \rangle = 1} (1-2q_c)$. To show that converse, note that if $C'_{\mathscr{G}^{\perp}/\mathscr{M}}\neq 0$ (recall notation from Eq.~\eqref{eq: main-detector-matrix-basis-change}), then there always exists some $b \in \mathscr{G}^\perp \setminus  \mathscr{M}$ such that 
        \begin{align}\label{eq: logical-prior-distribution}
    \lambda_b = ( \prod_{a \in \mathcal{K}_{\mathscr{M}}: \langle a, b \rangle = 1} (1-2q_{[a]})  ) (\prod_{c \in C} (1-2q_c)^{C'_{\mathscr{G}^\perp/ \mathscr{M}}[b, c]}),
\end{align}
where $C'_{\mathscr{G}^\perp/ \mathscr{M}}[b, c] \neq 0$ for some $c \in C$. The expression within the first parenthesis gives the product of coefficients of the prior distribution, which by definition (Eq.~\eqref{eq: main-prior-distribution}) equals $\prod_{a \in \mathcal{K}_{\mathscr{M}}: \langle a, b \rangle = 1} (1-2q_{[a]}) = \prod_{a \in \mathcal{K}_{\mathscr{M}}: \langle a, b \rangle = 1} \prod_{a' \in \mathcal{K}: \sigma(a') = \sigma(a)} (1-2q_{a'})$. The terms within the second parenthesis are independent from coefficients of the prior distribution; in other words, varying $q_c$ from $c \in C'_{\mathscr{G}^\perp / \mathscr{M}}[b, :]$ would provide different values of $\lambda_b$ for any given prior distribution, which concludes the result. 
    \end{proof}
\end{corollary}

\begin{remark}
We note that the prior distribution in our case is not defined with respect to some detector error model~\cite{derks2024designing, takou2025estimatingdecodinggraphshypergraphs, blumekohout2025estimatingdetectorerrormodels}, as the detectors are typically comprised of linear (affine) relation of generators from $\mathscr{M}$. This arises a subtle caveat: if we restrict the prior distribution onto detectors, such a restricted prior distribution might not be capable of representing logical error probabilities, even if learnability up to logical equivalence condition is met. 
\end{remark}

In what follows, we bound the error that arises in cases where some faults are not learnable up to logical equivalence. In this case where $ 0 \leq p_a \leq \varepsilon_c$ for some small parameter $\varepsilon_c$. 
\begin{align}
    \log \lambda_b &= \sum_{a \in \mathcal{K}_{\mathscr{M}}: \langle a, b \rangle =1} \log(1 - 2q_a) \sum_{a' \in B \sqcup C: \sigma(a')= \sigma(a)} \log(1-2q_{a'}) +  \sum_{c \in C} C'_{\mathscr{G}^\perp / \mathscr{M}}[b, c]\log(1-2q_c)\\
    &= \sum_{j} e_j \log \lambda_{s_j} + \sum_{c \in C} C'_{\mathscr{G}^\perp / \mathscr{M}}[b, c]\log(1-2q_c), \\
    &= \sum_{j} e_j \log \lambda_{s_j} + O(|C'_{\mathscr{G}^\perp / \mathscr{M}}[b, :]| \varepsilon_c), 
\end{align}
where $C_{\mathscr{G}^\perp / \mathscr{M}}[b, c]$ takes value in $\{-1, 0, 1 \}$ and $s_j \in \mathscr{M}$. The first equation holds from taking logarithm Eq.~\eqref{eq: logical-prior-distribution} (assuming there is no sign ambiguity). The second line holds from Theorem~\ref{thm: circuit-logical-learnability} (See main text Theorem~\ref{thm: main-logical-learnability} for a unified treatment), where we denote $\log \lambda_b$ as a vector in $\mathbb{R}^{|\mathcal{K}|}$. restricting  $\log \lambda_b$ onto $A$ and $B$, then there always exists some $e \in \mathbb{R}^{|\mathscr{M}|}$ (more precisely, we could utilize the compressed sensing technique and write $e \in \mathbb{Q}^{q}$, see Lemma~\ref{lemma: logical-error-propagation}) such that 
\begin{align}
    \log \lambda_b |_{A} =\sum_{a \in \mathcal{K}_{\mathscr{M}}: \langle a, b \rangle =1} \log(1 - 2q_a)= \sum_{j} e_j \log \lambda_{s_j} |_{A}, 
\end{align}
for indexed $s_j \in \mathscr{M}$ where we keep the detailed indexing vague for simplicity. Note that if $\sigma(a') = \sigma(a) $ for any $a' \in B \sqcup C$, $\langle s, a' \rangle = \langle s, a \rangle$ for $s \in \mathscr{M}$. Then we can extend the coefficients $e$ to columns with $B$ and $C$, which explains the second line above. Relating to the effective Pauli probability, we can perturbatively expand (say that if $\Lambda$ is given as the Pauli-Lindblad form) for the first-order approximation~\cite{Chen_2023_learnability}

\begin{align}
   \begin{aligned}
        p^{\operatorname{eff}}_a &= \frac{1}{|\mathcal{A}|}\sum_{b \in \mathscr{G}^\perp} \chi_b(a) \lambda_b \approx \frac{1}{|\mathcal{A}|} \sum_{b \in \mathscr{G}^\perp} \chi_b(a) \left( 1 + \log \lambda_b \right), \\
        &\approx \frac{1}{|\mathcal{A}|} \sum_{b \in \mathscr{G}^\perp} \chi_b(a)  + \frac{1}{|\mathcal{A}|} \sum_{b \in \mathscr{G}^\perp} \chi_b(a) \left( \sum_j e^{(b)}_j \log \lambda_{s_j} + O(|C'_{\mathscr{G}^\perp / \mathscr{M}}[b, :]| \varepsilon_c) \right), 
   \end{aligned}
\end{align}
which leads to an additive error $\varepsilon_c$. For $a \notin \mathscr{G}$, we further approximate 
\begin{align}
   \begin{aligned}
        p^{\operatorname{eff}}_a &\approx  \frac{1}{|\mathcal{A}|} \left( \sum_{b \in \mathscr{G}^\perp} \chi_b(a) \sum_j e^{(b)}_j \log \lambda_{s_j} + \sum_{b \in \mathscr{G}^\perp \setminus \mathscr{M}} \chi_b(a)O(|C'_{\mathscr{G}^\perp / \mathscr{M}}[b, :]| \varepsilon_c) \right),
   \end{aligned}
\end{align}
where $|C'_{\mathscr{G}^\perp / \mathscr{M}}[b, :]| =0$ for $b \in \mathscr{M}$. Note that, for sufficiently small $\varepsilon_c$, we might be able to bound $p^{\operatorname{eff}}_a$ to some prescribed relative precision.


\subsection{Applications to estimation of logical error probabilities}\label{subsec: spacetime-logical-noise-channel}

We now derive the logical channel in the spacetime circuit noise model, under the above formalism. We will work within the syndrome extraction circuit in Definition~\ref{def: main-syndrome-extraction-circuit} without ancillary worldlines, for simplicity. Note that the following discussion cannot be straightforwardly applied if the circuit consists of measurement sequences that are not mutually commuting; that is, the circuit is implemented from a base subsystem code, which will be systematically investigated in the next update. 

\begin{definition}[Spacetime logical noise channel] 
Let $\Lambda \in \mathbb{R}[\mathcal{A}]$ be a spacetime Pauli distribution, and the spacetime subsystem code is given by $(\mathscr{G}, \mathscr{M})$. Then the logical distribution is given by $\Lambda^{L} \in \mathbb{R}[\bar{\mathcal{A}}]$ where $\bar{\mathcal{A}} \cong \mathbb{F}^{2k}_2$ is given by the probability, 
\begin{align}\label{eq: circuit-logical-noise}
    p^{L}_{\bar{l}} &= |\mathscr{G}| \sum_{a_{z}: z \in \mathbb{F}^{M}_2} p^{\operatorname{eff}}_{a_z \overleftarrow{l} }=  \sum_{a_z: z \in \mathbb{F}^M_2} \sum_{g \in \mathscr{G}}p_{a_z \overleftarrow{l}g},
\end{align}
where $\bar{l}$ denotes the coset representation of the logical operator with respect to $\mathscr{G}$ and $\overleftarrow{l} := \overleftarrow{\eta_T(l)} \in \mathscr{G}^{\perp} / \mathscr{M}$ for $l \in C(\mathcal{S}) / \mathcal{S}$. 
\end{definition}

We can take the Fourier (Walsh-Hadamard) transform on the effective distribution

\begin{align}\label{eq: eff-logical-error-probability-rest}
            \begin{aligned}
                p^{\operatorname{eff}}_{a_{z} \overleftarrow{l}} &= \frac{1}{|\mathscr{G}|} \sum_{g \in \mathscr{G}} p_{a_z \overleftarrow{l} g} = \frac{1}{|\mathscr{G}|} \sum_{g \in \mathscr{G}} \frac{1}{|\mathcal{A}|} \sum_{b \in \mathcal{A}} \lambda_{b} \chi_b(a_z \overleftarrow{l}g) \\
                &= \frac{1}{|\mathcal{A}|} \sum_{b \in \mathcal{A}} \lambda_{b} \chi_{b}(a_z \overleftarrow{l}) \frac{1}{|\mathscr{G}|} \sum_{g \in \mathscr{G}} \chi_b(g) \\
                &= \frac{1}{|\mathcal{A}|} \sum_{b \in \mathscr{G}^\perp} \lambda_{b} \chi_{b}(a_z \overleftarrow{l}) = \frac{1}{|\mathcal{A}|} \sum_{\overleftarrow{l'} \in \mathscr{G}^{\perp} / \mathscr{M}} \sum_{\overleftarrow{m} \in \mathscr{M}} \lambda_{\overleftarrow{l'} \overleftarrow{m}} \chi_{\overleftarrow{l'} \overleftarrow{m}}(a_z \overleftarrow{l}), \\
                &= \frac{1}{|\mathcal{A}|} \sum_{\overleftarrow{l'} \in \mathscr{G}^{\perp} / \mathscr{M}} \chi_{\overleftarrow{l}'}(\overleftarrow{l}) \chi_{\overleftarrow{l'}} (a_z)\sum_{\overleftarrow{m} \in \mathscr{M}} \lambda_{\overleftarrow{l'} \overleftarrow{m}} \chi_{\overleftarrow{m}}(a_z).
            \end{aligned}
        \end{align}

The fourth equality is given by the signed average over a group, and the fifth equality, where $b = \overleftarrow{l'} \overleftarrow{m}$ for $b \in \mathscr{G}^\perp = \mathscr{M} \cdot \mathscr{L}_{\operatorname{bare}}$ is given by the circuit-to-code isomorphism of the syndrome extraction circuit Theorem~\ref{thm: circuit-subsystem-code-1}. The final equality is given by the linearity (homomorphism property) of the character $\chi_b(\cdot)$. Hence, 

\begin{align}\label{eq: logical-error-probability-rest}
     \begin{aligned}
         p^{L}_{\bar{l}} &= |\mathscr{G}| \sum_{a_{z}: z \in \mathbb{F}^{M}_2} p^{\operatorname{eff}}_{a_z \overleftarrow{l} } \\
     &= \frac{1}{|\mathscr{G}^\perp / \mathscr{M}|}  \sum_{\overleftarrow{l'} \in \mathscr{G}^{\perp} / \mathscr{M}} \chi_{\overleftarrow{l}'}(\overleftarrow{l}) \underbrace{\frac{1}{|\mathscr{M}|} \sum_{a_{z}: z \in \mathbb{F}^{M}_2}\sum_{\overleftarrow{m} \in \mathscr{M}} \lambda_{\overleftarrow{l'} \overleftarrow{m}} \chi_{a_z}({\overleftarrow{l'}\overleftarrow{m}})}_{=: \lambda^L_{\bar{l}}},
     \end{aligned}
\end{align}
where the second line we utilize the fact $\chi_a(b) = \chi_b(a)$. In other words, as long as $\lambda_b$ for $b \in \mathscr{M}$ is learnable from the syndrome data, we can estimate the logical error probability to arbitrary precision in principle. This form looks intimidating and remains intractable even for a small system. We can often evaluate this in another way, using the prior distribution, assuming the learnability up to logical equivalence condition is met

\begin{align}\label{eq: logical-error-probability}
    \begin{aligned}
       p^{L}_{\bar{l}}  &= |\mathscr{G}| \sum_{a_z: z \in \mathbb{F}^M_2} p^{\operatorname{eff}}_{a_z \overleftarrow{l}} = \frac{1}{|\mathscr{G}^\perp|} \sum_{a_z: z \in \mathbb{F}^M_2} \sum_{b \in \mathscr{G}^\perp} \lambda_b \chi_{b}(a_z \overleftarrow{l}),
        \\
        &= \frac{1}{|\mathscr{G}^\perp|}  \sum_{a_z: z \in \mathbb{F}^M_2} \sum_{b \in \mathscr{G}^\perp} \prod_{c \in \mathcal{K}_\mathscr{M}} ((1-q_{[c]}) + q_{[c]} \chi_c({b}))\chi_{b}(a_z \overleftarrow{l}), \\
        &= \frac{1}{|\mathscr{G}^\perp |}  \sum_{a_z: z \in \mathbb{F}^M_2} \sum_{b \in \mathscr{G}^\perp } \sum_{J \subset \mathcal{K}_{\mathscr{M}}} \prod_{c \in J} q_{[c]} \chi_{b}(c) \prod_{c' \notin J}(1-q_{[c']}) \chi_{b}(a_z \overleftarrow{l}), \\
        &= \frac{1}{|\mathscr{G}^\perp |}  \sum_{a_z: z \in \mathbb{F}^M_2} \sum_{J \subset \mathcal{K}_{\mathscr{M}}} \prod_{c \in J}  q_{[c]} \left( \sum_{b \in \mathscr{G}^\perp }\chi_{b}(\prod_{c \in J}c a_z \overleftarrow{l}) \right)\prod_{c' \notin J}(1-q_{[c']}), \\
        &=  \sum_{a_z: z \in \mathbb{F}^M_2} \sum_{J \subset \mathcal{K}_{\mathscr{M}}} \prod_{c \in J}  q_{[c]} \prod_{c' \notin J}(1-q_{[c']}) \mathbf{1}\{\prod_{c \in J} ca_z \overleftarrow{l} \in \mathscr{G}\} \\
        & = \sum_{J \subset \mathcal{K}_{\mathscr{M}}} \prod_{c \in J}  q_{[c]} \prod_{c' \notin J}(1-q_{[c']}) \mathbf{1}\{\prod_{c \in J} ca_{\sigma(\prod_{c \in J} c)} \overleftarrow{l} \in \mathscr{G}\}.
    \end{aligned}
\end{align}
The second line we plug in the parameterization of the prior distribution Eq.~\eqref{eq: main-prior-distribution} (Definition~\ref{def: prior-distribution}) and Corollary~\ref{cor: main-prior-distribution-learn-logical-error} (Eq.~\eqref{eq: logical-prior-distribution}). The third line follows from standard convolution terms and the fifth line follows from signed average over $\mathscr{G}^\perp = \mathscr{M} \cdot \mathscr{L}_{\operatorname{bare}}$, which is also a group (subspace) of $\mathcal{A}$. In this case, the logical error probability (given a decoder) can be learned from syndrome data via learning the coefficients of the prior distribution. Note that, we can generally write the logical error probability in this form regardless of learnability up to logical equivalence is achieved, by substituting with the general parameterization in Eq.~\eqref{eq: pauli-character-fourier} and Lemma~\ref{lemma: fourier-character-basis-sign}. We give an intuitive understanding of Eq.~\eqref{eq: logical-error-probability} as follows. In a typical quantum experiments, the logical error probability is computed by updating the logical frame with (XOR) the recovery vector, described by the detector error model~\cite{derks2024designing}. The indicator implies that $\prod_{c \in J} ca_{\sigma(\prod_{c \in J} c)} $ lies in the logical coset of $\overleftarrow{l}$, which would be able to be detected by measuring its conjugate partner. Denote the conjugate partner of $\overleftarrow{l}$ by $\overleftarrow{l'}$ such that $\langle \overleftarrow{l'}, \overleftarrow{l} \rangle = \langle \eta_T(l'), \eta_T(l)\rangle =1$ due to $T$ is even. Then a logical flip only occurs when $\langle \prod_{c \in J} ca_{\sigma(\prod_{c \in J} c)}, \overleftarrow{l'} \rangle =1 = \langle \overrightarrow{\prod_{c \in J} c} \overrightarrow{a_{\sigma(\prod_{c \in J} c)}}, \eta_T(l' )\rangle =1$. Since $\prod_{c \in J} ca_{\sigma(\prod_{c \in J} c)} \in \mathscr{M}^\perp = \mathscr{G} \cdot \mathscr{L}_{\operatorname{bare}}$, recall from construction of the gauge group in Eq.~\eqref{eq: spacetime-gauge-group} in the syndrome extraction circuit, it follows that 
\begin{align}
    \eta_T(\overrightarrow{\prod_{c \in J} ca_{\sigma(\prod_{c \in J} c)}}) = \eta_T(\overrightarrow{\prod_{c \in J} c}) \eta_T(\overrightarrow{a_{\sigma(\prod_{c \in J} c)}}) \in C(\mathcal{S}) / \mathcal{S}.  
\end{align}
This implies the the forward propagation of residual error at the terminal time is given by a logical codeword. This confirms our intuition that circuit logical fault leads to logical failure at the terminal (output) stabilizer code. Because of our spacetime formalism, we have exclusively been working in the rest frame. However, as the above short example alludes, the mechanism of circuit logical failure should not be frame-dependent. In what follows, we jump to the boosted frame where we show that the logical error probability as we defined in Eq.~\eqref{eq: eff-logical-error-probability-rest} and Eq.~\eqref{eq: logical-error-probability-rest} is frame independent.

Recall that we have the measurements taken at $\eta_t(\mathcal{M}_{t+0.5})$ for $t =0, 1, \cdots, T-1$, where we have in total $M$ generators. We define the projection operator for $z \in \mathbb{F}^{M}_2$ by 
\begin{align}
    P_{z} = \prod^{M}_{j=1} (\frac{1 + (-1)^{z_j} \eta_{t_j}(m_{i_j})}{2}),
\end{align}
where $t_j$ and $i_j$ indicate that the globally labeled measurement $j$ is the $i_j$th measurement operator performed at time $t_j + 0.5$. Translating to the backward propagation,  
\begin{align}
    \overleftarrow{P}_z = \prod^{M}_{j=1} (\frac{1 + (-1)^{z_j} \overleftarrow{\eta_{t_j}(m_{i_j})}}{2}).
\end{align}
we can always write generators of $\mathscr{M}$ as backward propagation in the syndrome extraction circuit from Theorem~\ref{thm: circuit-subsystem-code-1}. Under the rest frame, we denote the Pauli operator correction associated with this detected syndrome by $a_{z} \overleftarrow{P}_z = \overleftarrow{P}_0 a_{z}$ with $a_{z} \in \mathcal{A}$, the spacetime Pauli group. Notice that by Proposition~\ref{prop: spacetime-cumulant-automorphism} we have that 

\begin{align}
    a_{z} \overleftarrow{P}_z = \overleftarrow{P}_0 a_{z}  \iff \overrightarrow{a_{z}} P_{z} = P_0 \overrightarrow{a_{z}}.
\end{align}

Then we can define decoders in two different frames as, 
\begin{align}
    \overrightarrow{\mathcal{D}}(\cdot) &= \sum_{z \in \mathbb{F}^{m}_2} \overrightarrow{a_{z}} P_{z} \cdot P_{z} \overrightarrow{a_{z}}, \\
    \mathcal{D}(\cdot) &= \sum_{z \in \mathbb{F}^{m}_2} a_{z} \overleftarrow{P}_z \cdot \overleftarrow{P}_z a_{z} .
\end{align}

The above maps are obviously trace-preserving and define CPTP channels by directly observing their Kraus operators $\{ \overrightarrow{a_{z}} P_{z}\}, \{a_{z} \overleftarrow{P}_\alpha \}$ respectively. Next we consider the frame transformation for distributions defined on $\Lambda \in \mathbb{R}[\mathcal{A}]$. Representing the distribution in the character basis using Eq.~\eqref{eq: pauli-character-fourier}, then under the frame transformations, 
\begin{align}
    \Lambda = \prod_{a \in \mathcal{K}} ((1-q_{a}) + q_{a} \chi_a) \rightarrow \overrightarrow{\Lambda} = \prod_{\overrightarrow{a}: a \in \mathcal{K}}((1-q_{a}) + q_{a} \chi_{\overrightarrow{a}}).
\end{align}
Then it is easy to verify that 
\begin{align}
    \lambda_{\overleftarrow{b}} = \overrightarrow{\lambda}_{b}.
\end{align}

Let us consider the boosted frame at the final time $T$ where we denote ${a_{z, T}} = \eta_T(\overrightarrow{a_{z}})$ as a candidate for the final correction term. This motivates an alternative definition of logical noise distribution, induced by $\overrightarrow{\Lambda}$ followed by $\overrightarrow{D}$. Let's define the effective probability from $\overrightarrow{\Lambda}$ as $\overrightarrow{p}^{\operatorname{eff}}_a$ for $a \in \mathcal{A}$, that is, 

\begin{align}
    \begin{aligned}
         p^{\operatorname{eff}}_{a_{z} \overleftarrow{l}} &= \frac{1}{|\mathcal{A}|} \sum_{\overleftarrow{l'} \in \mathscr{G}^{\perp} / \mathscr{M}} \chi_{\overleftarrow{l}'}(\overleftarrow{l}) \chi_{\overleftarrow{l'}} (a_z)\sum_{\overleftarrow{m} \in \mathscr{M}} \lambda_{\overleftarrow{l'} \overleftarrow{m}} \chi_{\overleftarrow{m}}(a_z), 
     \\&= \frac{1}{|\mathcal{A}|} \sum_{\eta_T(l'):\overleftarrow{l'} \in \mathscr{G}^{\perp} / \mathscr{M}} \chi_{\eta_T(l')}(\overrightarrow{\overleftarrow{l}}) \chi_{\eta_T(l')} (\overrightarrow{a_z})\sum_{m:\overleftarrow{m} \in \mathscr{M}} \overrightarrow{\lambda}_{\eta_T(l') m} \chi_{m}(\overrightarrow{a_z}),\\
      &= \frac{1}{|\mathcal{A}|} \sum_{l':\overleftarrow{l'} \in \mathscr{G}^{\perp} / \mathscr{M}} \chi_{\eta_T(l')}(\eta_T(l)) \chi_{\eta_T(l')} (\overrightarrow{a_z})\sum_{m: \overleftarrow{m} \in \mathscr{M}} \overrightarrow{\lambda}_{\eta_T(l') m} \chi_{m}(\overrightarrow{a_z}),\\
      &= \overrightarrow{p}^{\operatorname{eff}}_{\overleftarrow{a_z} \eta_T(l)}.
    \end{aligned}
\end{align}
The third line is computed according to Theorem~\ref{thm: subsystem-circuit-to-code-isomorphism} or Theorem~\ref{thm: main-circuit-to-code-iso} as 
        \begin{align}
            \overleftarrow{l} = \eta_T(l) \eta_{T-1}(u^{-1}_{(T, T-1)}(l)) \cdots \eta_{0}(u^{-1}_{(T, 0)}(l)),
        \end{align}
        so that we can compute that $\eta_T(\overrightarrow{\overleftarrow{l}}) = l \cdot l \cdots l $ where it multiplies precisely $T+1$ many times. Since by assumption that $T$ is even. Then $\eta_T(\overrightarrow{\overleftarrow{l}}) = l \cdot l \cdots l = l$.
This implies that we can compute the logical error probability 
\begin{align}
      p^{L}_{\bar{l}} &= |\mathscr{G}| \sum_{a_{z}: z \in \mathbb{F}^{M}_2} p^{\operatorname{eff}}_{a_z \overleftarrow{l} } = |\mathscr{G}| \sum_{a_z: z \in \mathbb{F}^M_2} \overrightarrow{p}^{\operatorname{eff}}_{\overrightarrow{a_z} \eta_T(l)}, 
\end{align}
which implies that the logical error probability, as it should be, is frame-independent.

\section{Sample complexity and error analysis}\label{sec: sample-complexity-general}


\subsection{Error propagation via compressed sensing}\label{subsec: error-propagation-compressed-sensing}

In practice, the syndrome expectations are only sampled up to a precision. We refer to it as the \emph{imprecision} due to the finite sampling effect from the syndrome expectations. Let $Q \subset  \mathscr{M}$ from which we sample according to Theorem~\ref{thm: rip-bos-rowsampled-refined}, and we denote the vector $y_{m}:= - \log \lambda_{m}$, Let the noise vector $\omega$ be such that $\bar{y} = y + \omega$, we have $\omega \in \mathbb{R}^q$ with $\|\omega \|_1 \leq q  \varepsilon$ and $\|\omega\|_2 \leq \sqrt{q} \varepsilon$. Our goal is to examine how the sample imprecision could propagate into learning the (effective) distributions. 

\subsubsection*{Bounds on errors in learning the entire (prior) distribution}

Recall Lemma~\ref{lemma: learn-all-distribution-from-syndrome}, we can learn $\Lambda$ from syndrome data if and only if $\mathcal{K} = \mathcal{K}_{\mathscr{M}}$. Otherwise, we could consider learning the prior distribution from syndrome data
\begin{align*}
    \Lambda^{\operatorname{prior}} = *_{a \in \mathcal{K}_{\mathscr{M}}}((1-q_{[a]}) + q_{[a]}\chi_{[a]}),
\end{align*}
where the prior distribution Eq.~\eqref{eq: main-prior-distribution} is rewritten above for convenience. Apply Theorem~\ref{thm: rip-partial-hadamard-refined} and Corollary~\ref{cor: compressed-sensing-sample-theory} with $\mathcal{A}$ in the setting stated in Theorem~\ref{thm: learn-all-distribution-from-syndrome} in Section~\ref{sec: learnability-general}, let $A_{\mathscr{M}, \operatorname{res}} \in \mathbb{R}^{q \times (K-1)}$ be the row-subsampled and column restricted matrix, where $A_{\mathscr{M}, \operatorname{res}}[a, b] = \beta(a, b)$ for $a \in Q \subset \mathscr{M}$ and $b \in \mathcal{K}_{\mathscr{M}} \setminus \{ I \}$. We also denote $H'_{\mathscr{M}, \operatorname{res}} \in \mathbb{R}^{q \times (K+1)}$ by the row-subsampled Hadamard matrix restricted on the columns in $\mathcal{K}_{\mathscr{M}} \cup \{ I \} $, where $H'_{\mathscr{M}, \operatorname{res}}[a, b] = (-1)^{\beta(a, b)}$ for $a \in Q$ and $b \in \mathcal{K}_{\mathscr{M}}$. Furthermore, we denote that $H_{\mathscr{M}, \operatorname{res}}$ to exclude the identity element $\mathcal{K}_{\mathscr{M}}$. Finally, we denote $K+1$th-order restricted isometric constant by $\delta_K$, which is assumed to be small and a constant. 

\begin{lemma}\label{lemma: noisy-recovery-all-distribution}
    Consider the following recovery problem, 
    \begin{align}\label{eq: noisy-estimation-problem}
        y = A_{\mathscr{M}, \operatorname{res}} x + \omega, 
    \end{align}
    subject to the noise $\omega$, where we have that $x_a = - \log(1-2q_{[a]}), \: a \in \mathcal{K}_{\mathscr{M}}$. Then there exists a unique choice of noisy recovery $\bar{x}$ such that 
    \begin{align}
        \|\bar{x} - x \|_2 \leq  2  \frac{\sqrt{1 + \delta_{K}}}{1 - \delta_{K}} \sqrt{q} \varepsilon.
    \end{align}
    
    \begin{proof}
    Note that we can always write 
    \begin{align}
        A_{\mathscr{M}, \operatorname{res}} = \frac{1}{2}(1_{q \times (K-1)} - H_{\mathscr{M}, \operatorname{res}}),
    \end{align}
     This implies that we can always write Eq.~\eqref{eq: noisy-estimation-problem} from, 
        \begin{align}
          \begin{aligned}
                2 y & = (\sum_i x_i ) 1_{q} - H_{\mathscr{M}, \operatorname{res}}x +2\omega  \\
            &= \left( H_{\mathscr{M}, \operatorname{res}} | 1_{q} \right) \begin{pmatrix}
                -x \\
                s
            \end{pmatrix} + 2\omega \\
            &= H'_{\mathscr{M}, \operatorname{res}} \begin{pmatrix}
                -x \\
                s
            \end{pmatrix} + 2\omega.
          \end{aligned}
        \end{align}
        Where $s = \sum^{K-1}_{i=1} x_i $. Then we denote that 
        \begin{align}
            \begin{pmatrix}
                -\bar{x} \\
                \bar{s}
            \end{pmatrix} = 2\left( H^{'T}_{\mathscr{M}, \operatorname{res}} H'_{\mathscr{M}, \operatorname{res}} \right)^{-1} H^{'T}_{\mathscr{M}, \operatorname{res}} y. 
        \end{align}
        This follows that 
        \begin{align}
         \begin{aligned}
              \|\bar{x} - x \|_2 &\leq   \| \begin{pmatrix}
                -\bar{x} \\
                \bar{s}
            \end{pmatrix} - \begin{pmatrix}
                -x \\
                s
            \end{pmatrix} \|_2  \leq \| 2\left( H^{'T}_{\mathscr{M}, \operatorname{res}} H'_{\mathscr{M}, \operatorname{res}} \right)^{-1} H^{'T}_{\mathscr{M}, \operatorname{res}}  \omega \|_2 \\
            &\leq 2 \| \left( H^{'T}_{\mathscr{M}, \operatorname{res}} H'_{\mathscr{M}, \operatorname{res}} \right)^{-1} \|_2 \| H^{'T}_{\mathscr{M}, \operatorname{res}}\|_2 \|\omega\|_2 \\
            & \leq 2  \frac{\sqrt{1 + \delta_{K}}}{(1 - \delta_{K})} \frac{\varepsilon \sqrt{q}}{\sqrt{q}} = 2  \frac{\sqrt{1 + \delta_{K}}}{(1 - \delta_{K})}  \varepsilon,
         \end{aligned}
        \end{align}
        where at the final inequality, we have used the property $\|\omega\|_2 \leq \sqrt{q} \varepsilon$ and we take the normalization into consideration with the RIP constant $\delta_{K}$, when bounding the eigenvalues.  
    \end{proof}
\end{lemma}

Depending on the norm of choice, we can also cast the above theorem in $\ell_1$ norm, which results in the bound 
\begin{align}\label{eq: noisy-l1-inverse-bound}
        \|\bar{x} - x \|_1 \leq  2  \frac{\sqrt{1 + \delta_{K}}}{(1 - \delta_{K})} \sqrt{K} \varepsilon. 
\end{align}
Note that we have in general $q = \widetilde{O}(K)$ from Theorem~\ref{thm: rip-partial-hadamard-refined}, we see that this $\ell_1$ norm bound is comparable with $\|\omega \|_1$. Note that we have $\lambda_a = e^{-x_a}, a \in \mathcal{K}_{\mathscr{M}}$. We now examine how the above error translates to the product of Fourier coefficients. 

\begin{lemma}\label{lemma: noisy-recovery-physical-fourier-coeff}
    Let
    \begin{align}
        \lambda_{S} = \prod_{a \in S \subseteq \mathcal{K}_\mathscr{M}} \lambda_{a},
    \end{align}
    where $\bar{\lambda}_{S}$ be the estimator constructed from $\bar{\lambda}_a = \exp(-\bar{x}_a)$ for $a \in S \subseteq \mathcal{K}_\mathscr{M}$. Under the setting of Lemma~\ref{lemma: noisy-recovery-all-distribution}, we have that $\bar{\lambda}_{S} \approx_{ O(\sqrt{K}\varepsilon), 0} \lambda_{S}$. 
    \begin{proof}
        First we note that $\bar{\lambda}_{\mathcal{K}_{\mathscr{M}}} \approx_{ \delta, 0} \lambda_{\mathcal{K}_\mathscr{M}}$ is equivalent to say that the ratio $\bar{\lambda}_{S} /\lambda_{S} \in [1-\delta, 1 +\delta] $. Then we can compute, 
        \begin{align}
            \begin{aligned}
                \bar{\lambda}_{S} / \lambda_{S} &= \prod_{a \in \mathcal{K}_\mathscr{M}; a \neq 0} \bar{\lambda}_{a} / \lambda_a = \prod_{a \in S; a \neq 0} \exp(-(\bar{x}_a - x_a)), \\
                &= \exp(-\sum_{a \in S; a \neq 0} (\bar{x}_a - x_a)).
            \end{aligned}
        \end{align}
        We can apply our $\ell_1$-norm statement with the monotonicity of the real-valued exponential to the following bound. 
        \begin{align}
            \exp(- \| \bar{x} - x \|_1) \leq \bar{\lambda}_{S} / \lambda_{S} \leq \exp(\|\bar{x} -x \|_1).
        \end{align}
        Plug in the bound Eq.~\eqref{eq: noisy-l1-inverse-bound} and to the first order $\varepsilon$, we have reached our statement. 
    \end{proof}
\end{lemma}

Collecting all pieces,  

\begin{theorem}\label{thm: error-propagation-physical}
     Let the Fourier coefficients (Pauli eigenvalues) associated with $\mathscr{M}$ can be estimated up to  a multiplicative precision  $\varepsilon$ with some small $\varepsilon > 0$ and the parameterized distribution in Eq.~\eqref{eq: pauli-character-fourier} satisfies the learnability condition in Theorem~\ref{thm: learn-all-distribution-from-syndrome}. Let $Q \subset \mathscr{M}$, $|Q| =q$ be a uniformly subsampled set so that the resulting column- and row-restricted matrix $A_{\mathscr{M}, \operatorname{res}}$ in Corollary~\ref{cor: compressed-sensing-sample-theory} has full-column rank and denote the $K+1$th-order restricted isometry constant by $1 > \delta_K > 0$.  For any $a \in \mathcal{A}$, any Fourier coefficient $\lambda_a$ can be learned from  $Q \subset \mathscr{M}$ with a multiplicative precision $2 \frac{\sqrt{1 + \delta_K}}{1-\delta_K} \sqrt{K} \varepsilon = O(\sqrt{K} \varepsilon)$. 
\end{theorem}

\begin{remark}
    The above bound in some aspect is overly pessimistic. For generic $a \in \mathcal{A}$, $a$ might only anticommute with every sparse pattern of $a' \in \mathcal{K}$. In which case, we may further restrict columns to $a$'s support. For a sufficiently sparse support, we would obtain with the relative precision $O(\varepsilon)$ as desired.  
\end{remark}

At this stage, we could also discuss how to directly infer the (probability) coefficients of the prior distribution in Definition~\ref{def: prior-distribution} and see Eq.~\eqref{eq: main-prior-distribution}. 

\begin{lemma}\label{lemma: noisy-recovery-prior}
    Under the same setting as above, there exists a unique noisy recovery such that for every $a \in \mathcal{K}_{\mathscr{M}}$, 
    \begin{align}
        \bar{q}_{[a]} \approx_{0,O(\varepsilon) } q_a,
    \end{align} 
    where the term in the big-$O$ notation is given by $\frac{\sqrt{1 + \delta_K}}{1-\delta_K} \varepsilon$. 
    \begin{proof}
        From Lemma~\ref{lemma: noisy-recovery-all-distribution}, note that for $a \in \mathcal{K}$, we have that 
        \begin{align}
            q_a = \frac{1}{2}(e^{-x_a} - 1).
        \end{align}
        So that we have that 
        \begin{align}
            |\bar{q}_{[a]} - q_{[a]}| = \frac{1}{2} | e^{-x_a} - e^{-\bar{x}_a} | \leq \frac{e^{-\min(x_a, \bar{x}_a)}}{2} |x_a - \bar{x}_a| \leq \frac{1}{2} \|\bar{x} - x \|_2 \leq \frac{\sqrt{1 + \delta_K}}{1-\delta_K} \varepsilon,
        \end{align}
        as desired. 
    \end{proof}
\end{lemma}

The bound in Lemma~\ref{lemma: noisy-recovery-prior} is free of the dimensionality factor $K$, hence optimal. To be noted, this bound cannot be applied in Lemma~\ref{lemma: noisy-recovery-physical-fourier-coeff} due to the fact that $\lambda_a$ is a product of a combination of coefficients $p_b$ ($q_b$), where we can no longer switch trivially with $1$-norm and $2$-norm, from which the dimensionality factor occurs.

\subsubsection*{Bounds on errors in learning the effective distribution}

In many cases, a distribution parameterized in the character basis Eq.~\eqref{eq: pauli-character-fourier} admits a parameterization set $\mathcal{K}$ that contains many logical equivalent elements, related by the gauge group $\mathscr{G}$. In this case, it is not possible to learn the entire distribution from purely stabilizer measurements. Instead, we wish to learn the \emph{effective} distribution given by the quotients with the gauge group $\mathscr{G}$, and we consider how to learn the Fourier coefficient $\lambda_{b}, b \in \mathscr{G}^\perp$ $\lambda_{a}, b \in \mathscr{M}$. Suppose that such a distribution admits a parameterization set $\mathcal{K}$ such that $a, b \in \mathcal{K}$, $\sigma(a) = \sigma(b) \iff a \sim_{\mathscr{G}} b$. Then by Theorem~\ref{thm: learn-eff-distribution-from-syndrome} and Theorem~\ref{thm: rip-partial-hadamard-refined}, Corollary~\ref{cor: compressed-sensing-sample-theory}, with a high probability, $A_{\mathscr{M}, \operatorname{res}}$ is full-column rank. Hence, we can always express that $- \log(\lambda_b) \in \operatorname{rs} D_{\mathscr{M}, \operatorname{res}}$, where we similarly restrict the subsampled (with replacement) rows from $Q \subset \mathscr{M}$. Denote the vector $e \in \mathbb{R}^q$ such that

\begin{align}
    e D_{\mathscr{M}, \operatorname{res}} = - \log(\lambda_b) = \sum_i e_i y_i.
\end{align}
Our goal is thus to bound with $\ell_1$ norm using the estimated $\sum_i e_i \bar{y}_i$. Note that by application of Cauchy–Schwarz inequality, we obtain 
\begin{align}\label{eq: log-error-after-take-log-schwartz}
    |\sum_{i}e _i \bar{y}_i - \sum_i e_i y_i | = |\sum_i e_i \eta_i| \leq \|e\|_2 \|\omega\|_2 \leq \|e\|_2 \sqrt{q} \varepsilon. 
\end{align}
Hence, it remains to determine the two-norm of the vector $e \in \mathbb{R}^q$. Note that since $e$ admits real-valued entries, one cannot naively bound $\|e\|_2 \leq \sqrt{q}$, so we need to further examine this case. Let $l \in \mathbb{R}^{K}$ such that $l_a = \beta(a, b)$ for $a \in \mathcal{K}_{\mathscr{M}}$, which records the commutation relation for the logical operator $b \in \mathscr{G}^\perp$.

\begin{lemma}\label{lemma: logical-error-propagation}
    Suppose that $H'_{\mathscr{M}, \operatorname{res}}$ satisfies RIP condition with some small $K+1$th-order restricted isometric constant $0 < \delta_K < 1$. For $b \in \mathscr{G}^\perp$, there exists a unique choice of $e \in \mathbb{Q}^q \subset \mathbb{R}^q$ such that 
    \begin{align}
         \|e \|_2  \leq \frac{2}{\sqrt{q(1-\delta_K)}} \|l \|_2. 
    \end{align}
    \begin{proof}
    First, note that we have
    \begin{align}
       (e A_{\mathscr{M}, \operatorname{res}})_{a} = l_a; \quad A_{\mathscr{M}, \operatorname{res}} = \frac{1}{2}(1_{q \times (K-1)} - H_{\mathscr{M}, \operatorname{res}} ).
    \end{align}
    For $a \in \mathcal{K}_\mathscr{M}$ and $a \neq 0$. Apply the similar technique now with $H'_{\mathscr{M}, \operatorname{res}}$ where we append the all one column to the left-most of $H_{\mathscr{M}, \operatorname{res}}$. 
       \begin{align}
          \begin{aligned}
               & e(1_{q \times K} - H'_{\mathscr{M}, \operatorname{res}}) =(\sum^q_i e_i) 1_{K} - e H'_{\mathscr{M}, \operatorname{res}}, \\
           & \left(  -e, s \right) \begin{pmatrix}
               H'_{\mathscr{M}, \operatorname{res}} \\
               1_K 
           \end{pmatrix} = 2l,
          \end{aligned}
       \end{align}
       where we denote $s = \sum^m_i e_i $ and we further denote $\tilde{e} = (-e, s ) $  $M = \begin{pmatrix}
               H'_{\mathscr{M}, \operatorname{res}} \\
               1_K 
           \end{pmatrix}$. 
        Let's consider 
        \begin{align}
            \begin{aligned}
                M^\dagger M = M^TM =  H^{'T}_{\mathscr{M}, \operatorname{res}} H'_{\mathscr{M}, \operatorname{res}} + 1_{K \times K}.
            \end{aligned}
        \end{align}
Hence, by the RIP of the restricted Hadamard matrix, we have that 
\begin{align}
    \sigma_{\max}(M^TM) \leq  q(1 + \delta_{K}) + K; \quad \sigma_{\min}(M^TM) \geq q(1 - \delta_{K}) .
\end{align}
Then we could set 
\begin{align}
    \tilde{e} := 2M(M^TM)^{-1} l, 
\end{align}
which is the unique choice to the following $\ell_2$ minimization problem $\min_{M^T\tilde{e} = l } \|\tilde{e}\|_2$. Note that since $l_0 =0$ for the identity element in $\mathcal{A}$. This implies that we have that $\sum^{q+1}_{i=1} \tilde{e}_i = 0$ and by correct ordering of columns of $B$, we arrive at $\tilde{e}_{q+1} = - \sum^{q}_{i=1} \tilde{e}_i$, consistent with our constraints.  Note that by the submultiplicity, $\|e\|_2 \leq \|B(B^TB)^{-1} \|_2 \|l\|_2$.  Note that one can show $\|M(M^TM)^{-1} \|_2  =  (\sigma_{\operatorname{min}}(M^TM))^{-1/2}$, which implies
\begin{align}
    \|e \|_2 \leq \|\tilde{e} \|_2 \leq \frac{2}{\sqrt{q(1-\delta_K)}} \|l \|_2, 
\end{align}
as required. Finally, we note that $\tilde{e}$ and $e$ in this way admit strictly rational values. This follows from the fact that both $l$ and $M$ are integer-valued, and, hence, $(M^TM)^{-1}$ lie in $\mathbb{Q} \subset \mathbb{R}$.  Physically, it means that the Fourier coefficients associated with $l$ can also be given by the product of Fourier coefficients of stabilizers in $\mathscr{M}$  to some powers specified by $e$. 
    \end{proof}
\end{lemma}

Hence, suppose that the effective distribution is learnable, we can bound the Fourier coefficient $\lambda_b$, $b \in \mathscr{G}^\perp$ to a multiplicative precision.

\begin{theorem}\label{thm: error-propagation-logical}
   Let the Fourier coefficients (Pauli eigenvalues) associated with $\mathscr{M}$ can be estimated up to  a multiplicative precision  $\varepsilon$ with some small $\varepsilon > 0$ and the parameterized distribution in the character basis satisfies the learnability condition up to logical equivalence in Theorem~\ref{thm: learn-eff-distribution-from-syndrome}. Let  $Q \subset \mathscr{M}$, $|Q| =q$ be a uniformly subsampled set so that the resulting column- and row-restricted matrix $A_{\mathscr{M}, \operatorname{res}}$ in Corollary~\ref{cor: compressed-sensing-sample-theory} has full-column rank and denote the $K+1$th-order restricted isometry constant by $1 > \delta_K > 0$.  For any $b \in \mathscr{G}^\perp$, let $l \in \mathbb{R}^K$ such that $l_a = \beta(a, b)$ for $a \in \mathcal{K}_{\mathscr{M}}$.  Then any Fourier coefficient $\lambda_b$ can be learned from  $Q \subset \mathscr{M}$ with a multiplicative precision $O(\|l\|_2 \varepsilon)$. 
   \begin{proof}
       For any $b \in \mathscr{G}^\perp$, denote $\bar{\lambda}_b$ to be the estimated Fourier coefficient. Then we have that, 
       \begin{align}
           \frac{\bar{\lambda}_b}{\lambda_b} = \exp(- \sum_i e_i (\bar{y}_i - y_i)).
       \end{align}
       By the preceding lemma and  Eq.~\eqref{eq: log-error-after-take-log-schwartz}, 
       \begin{align}
           |\sum^{q}_{i=1} e_i (\bar{y}_i - y_i)| = |\sum^{q}_{i=1} e_i \eta_i| \leq \|e\|_2 \|\omega\|_2 \leq \frac{2 \|l\|_2 \varepsilon}{\sqrt{1-\delta_K}}.
       \end{align}
       Taking the exponential, it follows that, 
\begin{align}
  \exp(- \frac{2 \|l\|_2 \varepsilon}{\sqrt{1-\delta_K}} )  \leq   \frac{\bar{\lambda}_b}{\lambda_b}  \leq \exp(\frac{2 \|l\|_2 \varepsilon}{\sqrt{1-\delta_K}}),
\end{align}
where we again use the monotonicity of the exponential and the above lemma. Then we keep to the first order of $\varepsilon$, which arrives at $\bar{\lambda}_b \approx_{O(\| l \|_2 \varepsilon), 0} \lambda_b$ as desired, where $\|l \|_2 \leq \sqrt{K}$. 
       
   \end{proof}
\end{theorem}

\begin{remark}
    Note that we can always take this to the additive precision. Recall that $|\lambda_b| \leq 1$, then we can simply relax to $\bar{\lambda}_b \approx_{0, O(\|l\|_2 \epsilon)} \lambda_b$. 
\end{remark}

\subsection{Background-free syndrome extractions}\label{subsec: background-free-sampling}

We now characterize on the sample efficiency in providing an accurate estimation to the syndrome expectation, to a multiplicative precision. In what follows, we make slight abuse of notation, treating $m \in \mathscr{M}$ as well as a Bernoulli random variable $m \sim \operatorname{Bern}(\varepsilon)$, but a clear context will always be provided. To our goal, let $m^{(1)}, \cdots, m^{(S)}$ be i.i.d random variables drawn from the Bernoulli distribution $\operatorname{Bern}(\varepsilon)$. We have that $\mathbb{E}_{m^{(i)} \sim \operatorname{Bern}(\varepsilon)}[m^{(i)}] = \varepsilon$ for all $i$. Denote $\bar{\varepsilon} = \frac{1}{S} \sum^S_{i=1} m^{(i)}$ and suppose that we wish to obtain to a multiplicative precision $\tau$ to $\varepsilon$, we have that, 

\begin{lemma}\label{lemma: epsilon-multi-translation}
    Let $\bar{\varepsilon} \approx_{\tau, 0} \varepsilon$ for some $1 > \tau > 0$. Denote $\mu = 1 - 2\varepsilon$ and $\bar{\mu} = 1- 2\bar{\varepsilon}$.  Then, 
    \begin{enumerate}
    \item $\bar{\mu} \approx_{0, 2\tau \varepsilon} \mu$.
    \item $\bar{\mu} \approx_{\frac{2\tau \varepsilon}{1-2\varepsilon}, 0} \mu$, for $1/2  > \varepsilon > 0$. 
    \end{enumerate}
    \begin{proof}
        Note that we have that $\bar{\varepsilon} \approx_{\tau, 0} \varepsilon$ if and only if $|\bar{\varepsilon} - \varepsilon| \leq \tau |\varepsilon|$. Then we have that $|\bar{\mu} - \mu| = |2 \varepsilon - 2 \bar{\varepsilon}| = 2 |\varepsilon - \bar{\varepsilon}| \leq 2 \tau \varepsilon =   \frac{2\tau \varepsilon}{ 1- 2 \varepsilon} |\mu|$.  The last equality follows precisely when  $1/2  > \varepsilon > 0$. 
    \end{proof}
\end{lemma}

We convert the multiplicative precision to the additive precision for $\bar{\mu}$. We could also consider the multiplicative precision, which holds when $ 2 \tau \varepsilon/ (1 - 2\varepsilon) = 2 \tau \varepsilon+ 4 \tau \varepsilon^2 + O(\varepsilon^3)$. Hence, when $\varepsilon$ is small enough,  we can keep to the first order where $\omega \approx 2 \tau $. Using the above notations,

\begin{theorem}\label{thm: 1/epsilon-scaling-upper-bound}
    Let $m^{(1)}, \cdots, m^{(S)}$ be i.i.d random variables drawn from the Bernoulli distribution $\operatorname{Bern}(\varepsilon)$, where $\mathbb{E}_{m^{(i)} \sim \operatorname{Bern}(\varepsilon)}[m^{(i)}] = \varepsilon$ for all $i$. Denote $\bar{\varepsilon} = \frac{1}{S} \sum^S_{i=1} m^{(i)}$, $\bar{\mu} = 1 - 2 \bar{\varepsilon}$, and $\mu = 1- 2 \varepsilon$. For a sufficiently small $\varepsilon > 0$ and any $1 > \tau > 0$, let $S > 12\varepsilon^{-1} \tau^{-2} \log(1/\delta)$. Then the following holds, with probability at least $1 - \delta$, for any $\delta > 0$.
    \begin{enumerate}
        \item (Additive precision) we have $\bar{\mu} \approx_{ 0, \tau \varepsilon} \mu$.
        \item (Multiplicative precision) We have $\bar{\mu} \approx_{ \varepsilon\tau/ (1-2\varepsilon), 0} \mu$.
    \end{enumerate}

     \begin{proof}
         Applying the Chernoff bound proposition~\ref{prop: chernoff-bound}, whenever $S >  12 \varepsilon^{-1} \tau^{-2} \log(1/\delta)$,  then with probability at least $1 - \delta$, $\bar{\varepsilon} \approx_{\tau/2, 0} \varepsilon$. By Lemma~\ref{lemma: epsilon-multi-translation}
and assumption that $\varepsilon$ is sufficiently small. The additive and multiplicative precision statements regarding $\bar{\mu}$ follows.     \end{proof}
\end{theorem}

The appearance of $O(1/\varepsilon)$ scaling is specialized to $\operatorname{Bern}(\varepsilon)$, when $\varepsilon$ is small, the variance of the distribution is given by $\varepsilon(1-\varepsilon) = O(\varepsilon)$. Moreover, the above statement can be thought as a variance-aware approach, where it achieves an optimal bound one can ever hope for. Using standard arguments such as from Ref.~\cite{lee2025efficientbenchmarkinglogicalmagic} or Le Cam two-point method~\cite{gong2024samplecomplexitypurityinner, lecam1972limits, lecam1986amsdt} to show that $\Theta(\varepsilon^{-1} \tau^{-2})$ is needed to reach the targeted precision from the above theorem.


\subsection{Putting them together: sample complexity in estimating logical error probability}\label{subsec: sampling-LEP}

Finally, we put everything together in obtaining the logical error probability, where we utilize the prior distribution $\Lambda^{\operatorname{prior}}$. In this section, we assume that the fault distribution is learnable at least up to  a logical equivalence. 

\begin{theorem}\label{thm: sample-complexity-prior-distribution}
     Let the prior distribution $\Lambda^{\operatorname{prior}}$ be given, 
    \begin{align*}
        \Lambda^{\operatorname{prior}} = *_{a \in \mathcal{K}_{\mathscr{M}}}((1-q_{[a]}) + q_{[a]} \chi_a).
    \end{align*}
    Let $Q \subset \mathscr{M}$ be a subsampled stabilizer measurement set such that $A_{\mathscr{M}, \operatorname{res}}$ attains full column-rank with the $K+1$th-order restricted isometric constant $\delta_K < 1$. Then $m \in Q$, we denote that $\lambda_m = \mathbb{E}_{m \sim \operatorname{Bern}(\varepsilon)}[(-1)^m] = \mu$ and estimated $\bar{\mu}$ from 
    \begin{align}
        S \geq 12 \eta^{-2} \varepsilon^{-1}\log(1/\delta); \quad \eta = \tau (1-2\varepsilon) \frac{1- \delta_K}{\sqrt{1 + \delta_K}},
    \end{align}
    many samples. Then for any $a \in \mathcal{K}_{\mathscr{M}}$, with probability at least $1 -\delta$, the coefficient can be estimated $\bar{q}_a \approx_{0, \tau \varepsilon + O(\varepsilon^2)} q_a$, up to the first order $O(\varepsilon)$.   
    \begin{proof}
        First we utilize Theorem~\ref{thm: 1/epsilon-scaling-upper-bound} and set $y = - \log \mu $ and $\bar{y} = - \log \bar{\mu}$. Then we have that 
        \begin{align}
            (1 - \frac{\varepsilon \eta}{1-2\varepsilon}) \leq |\bar{\mu} / \mu| \leq  (1 + \frac{\varepsilon \eta}{1+2\varepsilon}).  
        \end{align}
        Since $\bar{\mu} / \mu > 0$, then the above inequality can be transformed into
        \begin{align}
            |\bar{y} - y| \leq \frac{\varepsilon \eta}{1+2\varepsilon} + O(\varepsilon^2) = \varepsilon \tau \frac{1 -\delta_K}{\sqrt{1 + \delta_K}} + O(\varepsilon^2).
        \end{align}
         Then apply the Lemma~\ref{lemma: noisy-recovery-prior}, we obtain that and set $\eta$ as prescribed, we arrive at the desired additive approximation, to the first order $O(\varepsilon)$.
    \end{proof}
\end{theorem}

Utilizing this above, we now comment on the sample complexity in computing the logical error probability. Recall the logical error probability from Eq.~\eqref{eq: logical-error-probability}, 

\begin{align}
    \begin{aligned}
       p^{L}_{\bar{l}}  &= |\mathscr{G}| \sum_{z \in \mathbb{F}^M_2} p^{\operatorname{eff}}_{a_z \overleftarrow{l}}, \\
        & = \sum_{J \subset \mathcal{K}_{\mathscr{M}}} \prod_{c \in J}  q_{[c]} \prod_{c' \notin J}(1-q_{[c']}) \mathbf{1}\{\prod_{c \in J} ca_{\sigma(\prod_{c \in J} c)} \overleftarrow{l} \in \mathscr{G}\}.
    \end{aligned}
\end{align}

We now sketch a heuristic exponential speedup in estimating the logical error probability to a relative precision. Let the coefficient 
\begin{align}\label{eq: prior-coeff-uniform}
    q_{[c]} \equiv q \in [\alpha \varepsilon, \beta \varepsilon],
\end{align}
and fix a logical coset $\bar{l}$ and $\alpha, \beta$ are some (constant) parameters where $\beta/\alpha$ is close to $1$. This assumption reflects the convenience where additive precision $ \tau \varepsilon$ can be translated easily to relative precision $O(\tau)$. Let $d_{*}$ be the minimum integer $j$ such that there exists $J \subseteq \mathcal{K}_{\mathscr{M}}$ with $|J| = j$ and
\begin{align}
    \mathbf{1}\{\prod_{c \in J} a_{\sigma(\prod_{c \in J}c)} \overleftarrow{l} \in \mathscr{G} \} =1. 
\end{align}
Define the following set 
\begin{align}
    N_j(\bar{l}) := |\{J \subseteq \mathcal{K}_{\mathscr{M}}: |J| = j, \prod_{c \in J}c a_{\sigma(\prod_{c\in J}c)} \overleftarrow{l}  \} |,
\end{align}
The number of length $j$ set that would produce a logical flip. In this case, we can write under the uniformity assumption Eq.~\eqref{eq: prior-coeff-uniform}
\begin{align}
    \begin{aligned}
        p^{L}_{\bar{l}} &= \sum^{K=|\mathcal{K}_{\mathscr{M}}|}_{j=d_*} N_j(\bar{l})q^j (1-q)^{K -j},\\
        &= q^{d_*}F(q), 
    \end{aligned}
\end{align}
where $F(q) := \sum^{K=|\mathcal{K}_{\mathscr{M}}|}_{j=d_*} N_j(\bar{l})q^{j-d_*} (1-q)^{K -j}$. This way, we assume that the logical error probability is controlled by a single parameter $q$ such that
\begin{align}
    \frac{p^L_{\bar{l}}(\bar{q})}{p^L_{\bar{l}}(q)} = (\frac{\bar{q}}{q})^{d_*} \frac{F(\bar{q})}{F(q)}.
\end{align}
In this case, 
\begin{align}
    \begin{aligned}
        \frac{|p^L_{\bar{l}}(q) - p^L_{\bar{l}}(\bar{q})|}{p^L_{\bar{l}}(q)} = \frac{|q^{d_*}F(q) - \bar{q}^{d_*} F(\bar{q})|}{q^{d_*}F(q)}. 
    \end{aligned}
\end{align}
Typically for small $q$, the fluctuation of $F(q)$ might be large, taking the derivative. However, for a sufficiently large $d_*$, let us assume that $\bar{q}^{d_*} F(\bar{q})/ F(q) = \gamma \bar{q}^{d_*}$, for some $\gamma$. Then we can re-express the above by
\begin{align}
    \begin{aligned}
         \frac{|p^L_{\bar{l}}(q) - p^L_{\bar{l}}(\bar{q})|}{p^L_{\bar{l}}(q)} = \frac{|q^{d_*}F(q) - \bar{q}^{d_*} F(\bar{q})|}{q^{d_*}F(q)} = \frac{|q^{d_*} - \gamma \bar{q}^{d_*}|}{q^{d_*}} \leq \frac{|q^{d_*} - (1+\tilde{\tau})^{d_*} \gamma q^{d_*}|}{q^{d_*}} \leq ( (1+\tilde{\tau})^{d_*}\gamma -1), 
    \end{aligned}
\end{align}
where we make of the uniformity assumption from Eq.~\eqref{eq: prior-coeff-uniform} and Theorem~\ref{thm: sample-complexity-prior-distribution} where we, by the risk of impreciseness, say that $\bar{q} \approx_{\tilde{\tau}, 0} q$. Hence, take $\tilde{\tau} = \tau /(2d_*) $, we have that (assuming $\gamma \in (0, 1]$ so that it does not deviate too much)
\begin{align}
     ( (1+\tilde{\tau})^{d_*}\gamma -1) = (\exp(d_* \log(1 + \tilde{\tau})) \gamma -1) \leq \gamma e^{\tau/2} -1, 
\end{align}
where for sufficiently small $\tau < 2 \log 2/\gamma$, this expression is less than $1$, which establishes a desired relative precision. By Theorem~\ref{thm: sample-complexity-prior-distribution}, this would use 
\begin{align}
    S = O\left(\frac{d^2_*}{\tau^2} (1-2\varepsilon)^{-2} \frac{1 }{1-\delta_K} \varepsilon^{-1} \log(1/\delta)\right) 
\end{align}
samples to achieve a relative precision with probability $1 -\delta$. As implied in the heuristic approach above, the apparent difficulty in reaching a rigorous guarantee in reaching desired relative precision of the logical error probability lies on the fact that we cannot possibly enumerate (in classical computing resources) every term in Eq.~\eqref{eq: logical-error-probability}. The leading order approximation provided above might not be valid even below the threshold and might require other assumptions such as the uniformity assumption in the case above. As a result, we leave a rigorous characterization to possible future studies.

\section{Details on numerical simulations}
\label{sec: numerical-simulations}

This section documents the numerical pipeline used to produce Fig.~\ref{Fig: simulation}. The key goal is to estimate logical error probabilities (LEPs) using only detector/syndrome data gathered from noisy syndrome-extraction circuits, without requiring direct observation of rare logical failures.

\paragraph{Notation for code identifiers.}
\makeatletter
\newcommand{\code@text}[1]{%
    \begingroup\small
    \ifdefined\nolinkurl
        \nolinkurl{#1}%
    \else
        \texttt{\detokenize{#1}}%
    \fi
    \endgroup
}
\DeclareRobustCommand{\code}[1]{\ifmmode\text{\code@text{#1}}\else\code@text{#1}\fi}
\makeatother
To avoid overfull lines caused by long monospace tokens (e.g., modulo and function names), we typeset code identifiers using a breakable typewriter font.

\subsection{Codes and circuit construction}

We study distance-$d$ surface-code instances at $d\in\{3,5,7\}$. For each code, we build a single-round syndrome-extraction circuit ($\texttt{num\_cycles}=1$) using the \texttt{Stim} simulator. The circuit is compiled to a detector error model (DEM), and we treat the DEM as defining a classical ``space--time'' parity-check problem for decoding.

In the implementation used for Fig.~\ref{Fig: simulation}, the circuit construction follows the syndrome-extraction routine in \code{sim_ldpc/dem_sim.py} (class \code{DEMSyndromeExtraction}). At a high level, each shot produces detector bits (parities of detector events) and logical observables from \texttt{Stim}'s detector sampler.

\paragraph{Noise parametrization.}
All circuit-level noise rates are specified via a base dictionary
\begin{align}
\texttt{CIRCUIT\_ERROR\_PARAMS} = \{\texttt{p\_i}=1.0,\ \texttt{p\_state\_p}=0.8,\ \texttt{p\_m}=0.9,\\
\qquad\qquad\qquad\quad \texttt{p\_CX}=0.0,\ \texttt{p\_idling\_gate}=0.0\},
\end{align}
and are scaled by a global physical error-rate parameter $p$ via
\begin{equation}
\text{(location error probability)}\ \equiv\ p\times \code{CIRCUIT_ERROR_PARAMS[location]}.
\end{equation}
Throughout Fig.~\ref{Fig: simulation}, faults are injected using \code{fault_type='DEPOLARIZE1'}.

\subsection{From Stim to a decoding instance}

From the DEM we construct a binary check matrix $H$ and an observables matrix $L$ using
\code{beliefmatching.detector_error_model_to_check_matrices(dem, allow_undecomposed_hyperedges=True)}.
The corresponding vector of fault-event priors $\{q_a\}$ is taken from \code{dem_matrix.priors}.
Allowing undecomposed hyperedges ensures that higher-weight DEM error mechanisms are included rather than being forcibly decomposed.

In all cases we flatten \texttt{Stim} loop constructs when extracting the DEM (\code{circ.detector_error_model(flatten_loops=True)}), ensuring that the resulting matrices correspond directly to the explicit detector graph used for decoding.

For each circuit shot, \texttt{Stim} returns detector outcomes $\vec d\in\{0,1\}^{m}$ and observable outcomes $\vec o\in\{0,1\}^{k}$. A decoder produces a correction $\vec c\in\{0,1\}^{n}$ (interpreted as a correction on the DEM fault variables). We record a logical failure whenever any observable is flipped after applying the correction,
\begin{equation}
\text{fail} = \mathbf{1}\left\{\left(\vec o + \vec c L^{\mathsf T}\right)\bmod 2 \neq \vec 0\right\}.
\end{equation}

\paragraph{Decoder.}
All LEP quantities in Fig.~\ref{Fig: simulation} use the same BP+LSD (OSD-type) decoder with parameters
\begin{align}
\texttt{max\_iter}=5,\ \texttt{bp\_method='min\_sum'},\ \texttt{ms\_scaling\_factor}=0.5,\\
\texttt{schedule='parallel'},\ \texttt{lsd\_method='lsd\_e'},\ \texttt{lsd\_order}=3.
\end{align}

\subsection{Quantities reported and how they are computed}

We compare three LEP estimators computed from the same circuit/decoder combination:
\begin{enumerate}
    \item \textbf{Sampled LEP} (direct Monte Carlo): sample $N$ shots, decode each shot, and estimate
    \begin{equation}
        \widehat p_{\mathrm{L}}^{\mathrm{sample}} = \frac{1}{N}\sum_{i=1}^{N} \text{fail}_i.
    \end{equation}

    \item \textbf{True LEP} (reference value given the DEM): compute $p_{\mathrm{L}}^{\mathrm{true}}$ using the exact DEM priors $\{q_a\}$ and the same decoder, via fault-enumeration truncated at \texttt{max\_order}=4.

    \item \textbf{Predicted LEP} (our learned-prior protocol): (i) learn an estimate $\{\widehat q_a\}$ from detector samples; (ii) compute
    \begin{equation}
        \widehat p_{\mathrm{L}}^{\mathrm{pred}}\equiv p_{\mathrm{L}}(\{\widehat q_a\}),
    \end{equation}
    using the same truncated fault-enumeration procedure (\texttt{max\_order}=4).
\end{enumerate}

\paragraph{Learning the priors from syndrome data.}
The prior-learning step is performed by the class \texttt{PredictPriors}. Given sampled detector bits, it constructs a (subsampled) linear system whose rows correspond to randomly selected products of detector events, and whose right-hand side is given by the corresponding empirical expectations computed from the shot data. The system is solved using the RIP-based compressed-sensing routine (\texttt{mode='rip'}), producing an estimate $\{\widehat q_a\}$ of the DEM fault-event priors.

For the scaling plots in Fig.~\ref{Fig: main-simulation}(a,b), the prior-learning accuracy is quantified by
\begin{equation}
\tau \equiv \max_a \frac{|\widehat q_a - q_a|}{q_a}.
\end{equation}

\subsection{Extraction of sample-complexity curves in Fig.~\ref{Fig: simulation}}

The combined scaling figure is generated by our data-generation and plotting scripts, which aggregate repeated-trial results and write intermediate summaries to disk for subsequent plotting.

\paragraph{Panel (a): shots vs target prior-learning accuracy.}
For a fixed $p$ (in the data shown, $p=2\times 10^{-4}$) and a fixed distance-$5$ surface-code instance, we run multiple independent trials at each shot count $N$ and record the distribution of $\tau$. In the notebook generating the stored data, the shot grid is
\begin{equation}
N\in\{10^4,\ 3\times 10^4,\ 10^5,\ 3\times 10^5,\ 10^6\},
\end{equation}
with $30$ independent trials per $N$. Fig.~\ref{Fig: main-simulation}(a) displays the resulting distributions (violin plots) and their medians as a function of $N$.

\paragraph{Panel (b): shots vs $p$ at fixed accuracy.}
Fixing a target accuracy threshold $\tau_{\mathrm{tar}}$ (in the data shown, $\tau_{\mathrm{tar}}=0.8$), we estimate the number of shots required to reach $\tau\le \tau_{\mathrm{tar}}$ by log--log interpolation of the median $\tau$ curve obtained in panel (a). Repeating this procedure at multiple values of $p$ yields a curve $N(\tau_{\mathrm{tar}},p)$, which we fit to a power law $N\propto p^{-\nu}$.

Concretely, the plotted values of $p$ are taken from the following grid:
\begin{equation}
p\in\{2\times 10^{-5},\ 5\times 10^{-5},\ 10^{-4},\ 2\times 10^{-4},\ 5\times 10^{-4}\}.
\end{equation}

\paragraph{Panel (c): shots needed for LEP estimation at fixed relative accuracy.}
At fixed $p=5\times 10^{-4}$ and distances $d\in\{3,5,7\}$, we estimate the number of circuit shots required for either sampled LEP or predicted LEP to reach a target \emph{relative} accuracy of 10\%. Concretely, for each $(d,N)$ we run multiple independent trials and compute the empirical standard deviation $\operatorname{std}(\widehat p_{\mathrm{L}})$ across trials. We define the relative error proxy
\begin{equation}
\mathrm{rel\_err}(N) \equiv \frac{\operatorname{std}(\widehat p_{\mathrm{L}})}{p_{\mathrm{L}}^{\mathrm{true}}},
\end{equation}
fit $\log_{10}(\mathrm{rel\_err})$ versus $\log_{10}(N)$ by least-squares, and solve for the $N$ such that $\mathrm{rel\_err}(N)=0.10$. This is performed separately for \emph{direct sampling at the logical level} (sampled LEP) and for \emph{learning from syndrome data} (predicted LEP).

In the combined figure generation, we optionally restrict to a conservative shot subset
\begin{equation}
N\in\{5\times 10^4,\ 10^5,\ 5\times 10^5,\ 10^6,\ 5\times 10^6\}
\end{equation}
to avoid small-$N$ finite-sample artifacts in the log--log regression.

\subsection{Bookkeeping and reproducibility}

For each configuration $(p,N,d)$ we run multiple independent trials (typically $20$ for the LEP studies, and $30$ for the prior-learning study) using independent random seeds. Each trial records (when enabled) the trio of LEP quantities \mbox{$\widehat p_{\mathrm{L}}^{\mathrm{sample}}$}, \mbox{$\widehat p_{\mathrm{L}}^{\mathrm{pred}}$}, and \mbox{$p_{\mathrm{L}}^{\mathrm{true}}$}, along with optional wall-clock runtimes for each computation. These values are aggregated across trials when extracting standard deviations, medians, and fitted shot requirements.

Overall, Fig.~\ref{Fig: main-simulation} shows that (i) the prior-learning error decreases as the number of detector samples increases, and (ii) using learned priors to predict LEPs can reduce the required circuit-shot count by orders of magnitude relative to direct logical-level sampling in the low-LEP regime.

\end{document}